\documentclass{article}

\usepackage[a4paper, margin=1in]{geometry}
\usepackage{mathtools}
\usepackage{mathpartir}
\usepackage{stmaryrd}
\usepackage{amssymb}
\usepackage{xparse}
\usepackage{xcolor}
\usepackage{todonotes}
\usepackage{pi-macros}
\usepackage{lambda-macros}
\usepackage{pftools}
\usepackage{enumitem}
\usepackage{hyperref}
\usepackage{xcolor}
\usepackage{listings}
\usepackage{tikz}
\usepackage{wrapfig}
\usepackage{natbib}
\usepackage{authblk}
\usepackage{amsthm}

\usetikzlibrary{positioning}
\usetikzlibrary{decorations.pathreplacing}
\pgfdeclarelayer{bg}    %
\pgfsetlayers{bg,main}
\usepackage[capitalise]{cleveref}

\newlist{indproof}{itemize}{5}
\setlist[indproof]{%
  itemsep=5pt,  %
  leftmargin=5pt,
  listparindent=10pt,
  font={\sc}, %
  label={}
}
\newcommand{\indcase}[1]{\item{{\color{RedDevil}{\sc Case}}(#1)}.}

\NewDocumentCommand{\defeq}{}{\stackrel{\mathclap{\mbox{\tiny def}}}{=}}
\NewDocumentCommand{\hcpcoexp}{}{CSLL}

\makeatletter
\newcommand{\bigsymbol}[1]{%
  \DOTSB
  \mathop{%
    \mathchoice{\big@symbol\displaystyle\Large{#1}}
               {\big@symbol\textstyle\large{#1}}
               {\big@symbol\scriptstyle\normalsize{#1}}
               {\big@symbol\scriptscriptstyle\small{#1}}%
    }\slimits@
}

\newcommand{\big@symbol}[3]{%
  \vcenter{%
    \sbox\z@{$#1\sum$}%
    \dimen@=0.875\dimexpr\ht\z@+\dp\z@\relax
    #2%
    \hbox{\resizebox{!}{\dimen@}{$\m@th#3$}}%
  }%
  \vphantom{\sum}%
}
\makeatother

\newcommand{\bigwith}{\bigsymbol{\with}}

\newcommand{\bigparr}{\bigsymbol{\parr}}

\hypersetup{
    colorlinks,
    linkcolor={RedDevil},
    citecolor={RegalBlue},
    urlcolor={Cowboy}
}

\makeatletter
\renewcommand{\@todonotes@textsize}{\footnotesize\sffamily}
\makeatother
\tikzset{notestyleraw/.append style={rounded corners = 1pt}}

\DeclareDocumentCommand{\Lars}{s +m}
{\IfBooleanTF{#1}
  {\todo[color=red!20]{Lars: #2}}
  {\todo[inline, color=red!20]{Lars: #2}}
}

\DeclareDocumentCommand{\Zesen}{s +m}
{\IfBooleanTF{#1}
  {\todo[color=green!50]{Zesen: #2}}
  {\todo[inline, color=green!50]{Zesen: #2}}
}

\DeclareDocumentCommand{\Alex}{s +m}
{\IfBooleanTF{#1}
  {\todo[color=blue!50]{Alex: #2}}
  {\todo[inline, color=blue!50]{Alex: #2}}
}

\DeclareDocumentCommand{\tagrule}{m}{\tag{\rulenamestyle{#1}} \label{#1}}

\crefformat{section}{\S#2#1#3}
\crefformat{subsection}{\S#2#1#3}
\crefformat{subsubsection}{\S#2#1#3}
\crefrangeformat{section}{\S\S#3#1#4--#5#2#6}
\crefmultiformat{section}{\S\S#2#1#3}{ and~#2#1#3}{, #2#1#3}{ and~#2#1#3}

\newtheorem{lemma}{Lemma}
\newtheorem{theorem}{Theorem}
\newtheorem{definition}{Definition}

\title{Client-Server Sessions in Linear Logic}

\author[1]{Zesen Qian}

\author[2]{G. A. Kavvos}

\author[1]{Lars Birkedal}

\affil[1]{
	Aarhus University, Denmark
}
\affil[2]{
	University of Bristol, UK
}

\allowdisplaybreaks

\begin{document}
\maketitle
	\begin{abstract}
		We introduce coexponentials, a new set of modalities for Classical Linear Logic. As duals to
		exponentials, the coexponentials codify a distributed form of the structural rules of weakening
		and contraction. This makes them a suitable logical device for encapsulating the pattern of a
		server receiving requests from an arbitrary number of clients on a single channel. Guided by
		this intuition we formulate a system of session types based on Classical Linear Logic with
		coexponentials, which is suited to modelling client-server interactions.
		We also present a session-typed functional programming language for server-client programming, which we translate to our system of coexponentials.
	\end{abstract}

\section{Introduction}

The programme of \emph{session types} \cite{honda1998,vasconcelos_2012} aims to formulate
behavioural type systems that capture the notion of a \emph{session}---a structured, concurrent
interaction between communicating agents. Very little is usually assumed about these agents: their
only shared resource is usually a set of \emph{channels} through which they can send and receive
messages. On the other hand, ever since its inception it has been clear that \emph{linear logic}
\cite{girard1987linear} has a deep and mystifying relationship with concurrency.
\citet{abramsky1994proofs} argued that process calculi and linear logic should be in a Curry-Howard
correspondence \cite{bellin1994p}. Consequently, one should be able to use formulas of linear logic
as types that specify concurrent interactions, thereby constructing a system of session types that
is logically motivated. Session types and linear types have recently undergone a swift rapprochement
beginning with the work of Caires and Pfenning \cite{caires2010,caires_linear_2016}.

Despite these advances, the $\pi$-calculi that have been developed as process calculi for Linear
Logic suffer from dire expressive poverty. The typable processes are free of deadlock and
nondeterminism, at the price of being unable to model even benign forms of race. One striking
omission is that it is difficult to write down a well-typed process that represents two distinct
clients being served by a server listening on a single channel. The goal of the present paper is to
introduce a logical device, namely the \emph{strong coexponential modalities}, that will allow
us to give a linear type to this extremely common pattern of concurrent interaction.

\subsection{The problem}
	\label{section:problem}

\citet{caires2010} proposed a Curry-Howard correspondence in which Intuitionistic Linear Logic is
used as a type system for $\pi$-calculus \cite{milner1992calculus}. This correspondence allows one
to interpret formulas of linear logic as \emph{session types}, i.e., as specifications of
disciplined communication over a named channel. A few years later \citet{wadler2014} extended this
interpretation to \emph{Classical Linear Logic (CLL)}. Wadler's system, which is called
\emph{Classical Processes (CP)}, perfectly corresponds to \citeauthor{girard1987linear}'s original
one-sided sequent system for CLL \citeyearpar{girard1987linear}. Its typing judgments are of the
form $\IsProc{P}{\Gamma}$, where $\Proc{P}$ is a $\pi$-calculus process, and $\Gamma$ is a list
$\Name{x_1} : A_1, \dots, \Name{x_n} : A_n$ of name-session type pairs, with $A_i$ a formula of
Classical Linear Logic. The operational semantics of CP led Wadler to the following interpretation
of the connectives.
\begin{center}
	\begin{tabular}{ccp{4cm}cccp{4cm}}
		$\otimes$ && output         && $\parr$   && input  \\
		$\with$   && offer a choice && $\oplus$  && make a choice \\
		$!$       && server         && $?$       && client
	\end{tabular}
\end{center}
We follow a convention by which the multiplicative connectives $\otimes$, $\parr$ associate to the
right. Thus a type like $A \otimes B \parr C$ can be read as: output a (channel of type) $A$, then
input a (channel of type) $B$, and proceed as $C$. While the interpretation of the first four
connectives is intuitive, something seems to have gone awry with the exponentials \cite[\S
3.4]{wadler2014}. We claim that the computational behaviour of exponentials in CP does not in fact
accommodate what we would think of as client-server interaction.

To begin, we consider the following aspects to be the main characteristics of a client-server architecture \cite[\S\S 2.3, 3.4]{van_steen_2017}:
\begin{enumerate}[label=(\roman*)] \label{enumerate:cs-crit}
	\item There is a \emph{server process}, which repeatedly provides a service any number of clients.
	\item There is a \emph{pool of client processes}, each of which requests the said service.
	\item There is a unique \emph{end point} at which the clients may issue their requests to the
	server.
	\item The underlying network is \emph{inherently unreliable}: clients may be served
	out-of-order, i.e.,  in a \emph{nondeterministic} manner.
\end{enumerate}
While Wadler's interpretation faithfully captures (i) and (iii), it does not immediately enable the
representation of (ii). Because of its deterministic behaviour, CP is incapable of modelling (iv).

A CP term $\IsProc{\Proc{S}}{\Name{x} : \ofc A}$ can indeed `serve' sessions of type $A$ over the
channel $\Name{x}$. However, the reading of a term $\IsProc{C}{\Name{y} : \whynot A}$ as a process
which behaves as a \emph{pool of clients} along channel $\Name{y}$ is not so crisp. Recall the three
rules of $?$, namely weakening, dereliction, and contraction. In CP:
\begin{mathpar}
	\inferrule*[right=$\whynot w$] {
		\IsProc{Q}{\Gamma}
	}{
		\IsProc{Q}{\Gamma, \Name{x} : \whynot A}
	}
	\and
	\inferrule*[right=$\whynot d$]{
		\IsProc{Q}{\Gamma, \Name{y} : A}
	}{
		\IsProc{\WhyD{x}{y}{Q}}{\Gamma, \Name{x} : \whynot A}
	}
	\and
	\inferrule*[right=$\whynot c$]{
		\IsProc{Q}{\Gamma, \Name{x} : \whynot A, \Name{y} : \whynot A}
	}{
		\IsProc{\Sb{Q}{y}{x}}{\Gamma, \Name{x} : \whynot A}
	}
\end{mathpar}
Wadler interprets these rules as \emph{client formation}. Weakening stands for the empty case of a
pool of no clients. Dereliction represents a single client following session $A$. Finally,
contraction enables one to aggregate two client pools together: two sessions that are both of type
$\whynot A$ can be collapsed into one, now communicating along the shared channel $\Name{x}$.

We argue that, of those interpretations, only the one for dereliction is tenable. In the case
of weakening, we see that at least one process is involved in the premise; hence, the `pool' formed
has at least one client in it, albeit one that does not communicate with the server. Likewise,
contraction does not aggregate different clients, but different sessions owned by the same client.
Beginning with a single process $\IsProc{\Proc{P}}{\Name{x} : A, \Name{y} : A}$ we can use
dereliction twice followed by contraction to obtain $\IsProc{\WhyD{w}{x}{\WhyD{w}{y}{P}}}{\Name{w} :
\whynot A}$. This process will ask for two channels that communicate with session $A$. Nevertheless,
the result is still a single process, and not a pool of clients. Dually, the type $\ofc A$ merely
connotes a \emph{shared channel}: a non-linearized, non-session channel which is used to spawn an
arbitrary number of new sessions, each one of type $A$ \cite[\S 3]{caires2010}.

More alarmingly, there is no way to combine two distinct processes $\Proc{P} \vdash \Name{z} : A$
and $\Proc{Q} \vdash \Name{w} : A$ into a single process $\Deriv{\textsf{pool}({\Name{x}; \Name{z}. \, \Proc{P},
\Name{w}. \, \Proc{Q}})} \vdash \Name{x} : \whynot A$ communicating along a shared channel. As a
remedy, Wadler introduces the \textsc{Mix} rule:
\begin{mathpar}
	\inferH{Mix}{
		\IsProc{P}{\Gamma} \\
		\IsProc{Q}{\Delta}
	}{
		\IsProc{\Par{P}{Q}}{\Gamma, \Delta}
	}
\end{mathpar}
\ruleref{Mix} was under carefully considered for inclusion in Linear Logic, but was ultimately left
out \cite[\S V.4]{girard1987linear}. Informally, it allows two completely independent,
non-intercommunicating processes to run `in parallel.' We may then use contraction to merge them
into a single client pool:
\begin{mathpar}
	\inferrule*[right=$\whynot c$]{
		\inferrule*[right=\ruleref{Mix}]{
			\inferrule*[right=$\whynot d$]{
				\IsProc{P}{\Name{z} : A}
			}{
				\IsProc{\WhyD{x}{z}{P}}{\Name{x} : \whynot A}
			} \\
			\inferrule*[right=$\whynot d$]{
				\IsProc{Q}{\Name{w} : A}
			}{
				\IsProc{\WhyD{y}{w}{P}}{\Name{y} : \whynot A}
			}
		}{
			\IsProc{\Par{\WhyD{x}{z}{P}}{\WhyD{y}{w}{Q}}}{\Name{x} : \whynot A, \Name{y} : \whynot A}
		}
	}{
		\IsProc{\Par{\WhyD{x}{z}{P}}{\WhyD{x}{w}{Q}}}{\Name{x} : \whynot A}
	}
\end{mathpar}
The operational semantics of the \ruleref{Mix} rule in CP are studied by \citet{lindley2016}. To
formulate them correctly one needs also to add the rule
\begin{mathpar}
	\inferH{Mix0}{
		\strut
	}{
		\Stop \vdash \cdot
	}
\end{mathpar}
\ruleref{Mix0} has a flavour of inconsistency to it, but is advantageous on two levels. On the
technical level, it let us show that the operational semantics, which adds a reaction
$\red{\Par{P}{Q}}{\Par{P'}{Q}}$ whenever $\red{\Proc{P}}{\Proc{P'}}$, is well-behaved (terminating,
deadlock-free, and deterministic). In terms of computational interpretation, \ruleref{Mix0}
represents a stopped process. This solves the second problem we pointed out above, viz. the
formation of a vacuously empty client pool:
\begin{mathpar}
	\inferrule*[right=$\whynot w$]{\inferrule*[right=\ruleref{Mix0}]{\strut}{\Stop \vdash \cdot}}{\Stop \vdash \Name{x} : \whynot A}
\end{mathpar}

Nevertheless, \textsc{Mix} and \textsc{Mix0} are unbecoming rules. To begin, they are respectively
equivalent to $\bot \multimap \mathbf{1}$ and $\mathbf{1} \multimap \bot$, thereby conflating the
two units. Moreover, it is well-known \cite[\S 1.1]{bellin_subnets_1997}
\cite{girard1987linear,abramsky1996interaction,wadler2014,lindley2016} that \ruleref{Mix} is
equivalent to
\begin{equation}
	\label{eq:mix-result}
	A \otimes B \multimap A \parr B \tag{$\ast$}
\end{equation}
where $C \multimap D \defeq \negg{C} \parr D$.

Admitting this implication is unwise. At first glance, \eqref{eq:mix-result} merely weakens the
separation between these connectives, and hence damages the interpretation of $\parr$ as input, and
$\otimes$ as output. However, we argue that deeper problems lurk just beneath the surface. \citet[\S
3.4.2]{abramsky1996interaction} describe a perspective on CLL which reads $A \parr B$ as
\emph{connected concurrency} (information \emph{necessarily} flows between $A$ and $B$ \cite[\S
V.4]{girard1987linear}) and $A \otimes B$ as \emph{disjoint concurrency} (no information flow
between $A$ and $B$ whatsoever). The implication $\eqref{eq:mix-result}$ makes $\otimes$ a special
case of $\parr$. Hence, flow between the components of $A \otimes B$ is \emph{permitted, but not
obligatory} \cite[\S 3.2]{abramsky_games_1994}. Thus, \eqref{eq:mix-result} allows us to
\emph{pretend} that there is flow of information between two clients.\footnote{This is evident in
the Abramsky-Jagadeesan game semantics for MLL+MIX: a play in $A \otimes B$ projects to plays for
$A$ and $B$, but the Opponent can switch components at will. The fully complete model consists of
\emph{history-free} strategies, so there can only be non-stateful Opponent-mediated flow of
information between $A$ and $B$.}

Nevertheless, generating the actual flow of information is seemingly impossible. Using \ruleref{Mix}
we can put together two clients $\IsProc{C_i}{\Name{c_i} : A}$, and get a single process
$\IsProc{\Par{C_0}{C_1}}{\Name{c_0} : A, \Name{c_1} : A}$. As the comma stands for $\parr$, we can
only cut this with a server $\IsProc{S}{\Name{s} : \negg{A} \otimes \negg{A}}$. But, by the
interpretation of $\otimes$ as disjoint concurrency, we see that the two client sessions will be
served by disjoint server components. In other words, the server will \emph{not} allow information
to flow between clients, which does not conform to our usual conception of a stateful server! To
enable this kind of flow, a server must use $\parr$. As we cannot cut a $\parr$ (in the server) with
another $\parr$ (in the client pool), we are compelled to also accept the converse implication $A
\parr B \multimap A \otimes B$ in order to convert one of the two $\parr$'s to $\otimes$. This
forces $\otimes = \parr$, which inescapably leads to deadlock \cite[\S 4.2]{lindley2016}.

Requiring $\otimes = \parr$, a.k.a. \emph{compact closure} \cite{barr1991,abramsky1996interaction},
is often deemed necessary for concurrency. In fact, \citet{lindley2016} argue that this
\emph{conflation of dual connectives} ($\mathbf{1} = \bot$, $\otimes = \parr$, and so on) is the
source of all concurrency in Linear Logic. The objective of this paper is to argue that there is
another way: we aim to augment the Caires-Pfenning interpretation of propositions-as-sessions with a
certain degree of concurrency \emph{without adding \ruleref{Mix}}. We also wish to introduce just
enough nondeterminism to convincingly model client-server interactions in a style that satisfies
points (i)--(iv).

We shall achieve both of these goals with the introduction of \emph{coexponentials}.

\subsection{Roadmap}

First, in \cref{section:fix} we discuss the expression of the usual exponential modalities of linear
logic ($!?$) as least and greatest fixed points. This leads us to a different definition of $\ofc$,
which we call the \emph{strong exponential}. By taking a `multiplicative dual' of these fixed point
expressions, we reach two novel modalities, the \emph{strong coexponentials}, for which we write
$\exc$ and $\que$. We refine coexponentials back into a weak form that is similar to the usual
exponentials, and show that they coincide with weak coexponentials in the presence of \ruleref{Mix}
and the \emph{Binary Cut} rule.

Following that, in \cref{section:processes} we introduce a process calculus with strong
coexponentials, which we call \hcpcoexp. This new system is in the style of \citet{kokke2019better},
which replaces the one-sided sequents with \emph{hypersequents}. It is argued that coexponentials
enable a new abstraction, viz. the collection of an arbitrary number of clients following session
$A$ into a \emph{client pool}, which communicates on a channel that follows session $\que A$.
Conversely, the rules for $\exc$ express the formation of a \emph{server}, which can be cut with a
client pool to serve it requests.

In \cref{section:pools} we present an extended example which illustrates the computational behaviour
of coexponentials, i.e. an implementation of the \emph{Compare-and-Set (CAS)} synchronization
primitive. Our system neatly encapsulates racy yet atomic
behaviour implicit in such operations.

In \cref{section:csgv} we explore the implications of coexponentials in a session-typed functional language.
We extend Wadler's GV with client-server interactions and translate them to coexponentials in CSLL.
We take advantage of the higher-level notation to give several examples that would be tedious to program directly in CSLL.

We survey related work in \cref{section:related-work}, and make some concluding remarks in \cref{section:conclusion}.

\section{Exponentials, fixed points, and coexponentials}
 \label{section:fix}

	\subsection{Exponentials as fixed points}
	\label{section:exp-as-fix}

  The exponential (or `of course') modality of linear logic $!$ is used to mark a replicable
	formula. While describing a combinatory presentation of linear logic, \citet[\S
	3.2]{girardlafont1987} noticed that $!A$ can be expressed as a fixed point
	\[
		\ofc A \cong \ounit \with A \with (\ofc A \otimes \ofc A)
	\]
	The three additive conjuncts on the RHS correspond to the three rules of the dual connective
	$\whynot$, namely weakening, dereliction, and contraction. As $\with$ is a \emph{negative}
	connective, the choice of conjunct rests on the `user' of the formula,\footnote{Also known as \emph{external choice}. In the language of game semantics, the \emph{opponent}.} who may pick one of the three conjuncts at will.

	One may thus be led to believe that, were we to allow fixed points for all \emph{functors}, we
	could obtain $\ofc A$ as the \emph{fixed point} of a functor. \citet[\S 2.3]{baelde2012least}
	discusses this in the context of a system of higher-order CLL with least and greatest fixed
	points. Using the functors
	\begin{align*}
		F_A(\var{X}) &\defeq \ounit \with A \with (\var{X} \otimes \var{X})
			&
		G_A(\var{X}) &\defeq \punit \oplus A \oplus (\var{X} \parr  \var{X})
	\end{align*}
	one defines
	\begin{align*}
		!A \defeq \nu F_A && ?A \defeq \mu G_A
	\end{align*}
	where $\mu$ and $\nu$ stand for the least and greatest fixed point respectively. Just by expanding
	the fixed point rules, one then obtains certain derivable rules. While those for $?$ are the usual ones---weakening, dereliction, and contraction---the rule for $!$ is radically	different:
	\begin{mathpar}
		\inferH{StrongExp}{
			\vdash \Gamma, B \\
			\vdash \negg{B}, \ounit \\
			\vdash \negg{B}, A \\
			\vdash \negg{B}, B \otimes B
		}{
			\vdash \Gamma, !A
		}
	\end{mathpar}
	As foreshadowed by the use of a greatest fixed point, this rule is \emph{coinductive}. To prove
	$!A$ from context $\Gamma$ one must use it to construct a `seed' value (or `invariant') of type $B$.
	Moreover, this value must be discardable ($\vdash \negg{B}, \ounit$), derelictable
	($\vdash \negg{B}, A$), and copyable ($\vdash \negg{B}, B \otimes B$). This is eerily reminiscent
	of the \emph{free commutative comonoids} used to build certain categorical models of Linear Logic
	\cite[\S 7.2]{mellies2009categorical}. Because of the arbitrary choice of `seed' type $B$, the
	system using this rule does not produce good behaviour under cut elimination: the normal forms do
	not satisfy the \emph{subformula property} \cite[\S 3]{baelde2012least}: not all detours are
	eliminated. We call the modality introduced by \ruleref{StrongExp} the \emph{strong exponential}.

	Baelde shows that the standard $!$ rule can be derived from \ruleref{StrongExp}. But while the
	strong exponential can simulate the standard exponential, it also enables a host of other
	computational behaviours under cut elimination. Put simply, the standard exponential ensures
	\emph{uniformity}: each dereliction of $\ofc A$ into an $A$ must be reduced to the very same proof
	of $A$ every time. This makes sense in at least two ways. First, when we embed intuitionistic
	logic into linear logic through the Girard translation, we expect that in a proof of  $(A \to B)^o
	\defeq \ofc A^o \multimap B^o$ each use of the antecedent $\ofc A$ produces the same proof of $A$.
	Second, we know that one way to construct the exponential in many `degenerate' models of linear
	logic \cite{barr1991,mellies2018} is through the formula
	\[
		\ofc A\ \defeq\ \bigwith_{n \in \mathbb{N}} A^{\otimes n} \mathbin{/} \sim_n
	\]
	where $A^{\otimes n} \defeq A \otimes \dots \otimes A$, and $A^{\mathop{\otimes} n} \mathbin{/}
	\sim_n$ stands for the equalizer of $A^{\otimes n}$ under its $n!$ symmetries. Decoding the
	categorical language, this means that we take one $\with$ component for each multiplicity $n$, and
	each component consists of exactly $n$ copies of the same proof of $A$.

	In contrast, the $!$ rules derived from their fixed point presentation merely create an infinite
	tree of occurrences of $A$, and not all of them need be proven in the same way.

	\subsection{Deriving Coexponentials}
	\label{section:deriving-coexp}

	Both exponentials (qua fixed points) are given by a tree where each fork is marked with a
	connective ($\otimes$ for $!$, $\parr$ for $?$). The leaves of the tree are either marked with
	$A$, or with the corresponding unit. Turning this process on its head leads to two dual
	modalities, which we call the \emph{coexponentials}.

	More concretely, we define two functors by dualising the connective that adorns forks. We must not
	forget to change the units accordingly: we swap $\ounit$ (the unit for $\otimes$) with $\bot$
	(the unit for $\parr$). Let
	\begin{align*}
		H_A(\var{X}) &\defeq \punit \with A \with (\var{X} \parr \var{X})
		&
		K_A(\var{X}) &\defeq \ounit \oplus A \oplus (\var{X} \otimes \var{X})
	\end{align*}
	The strong coexponentials are then defined by
	\begin{align*}
		\exc A \defeq \nu H_A && \que A \defeq \mu K_A
	\end{align*}
	We define $\negg{\left(\que A\right)} \defeq \exc \negg{A}$, and vice versa. We obtain the following derived rules.
	\begin{mathpar}
		\inferrule*[right=$\que w$]{
			\strut
		}{
			\vdash \que A
		}
		\and
		\inferrule*[right=$\que d$]{
			\vdash \Gamma, A
		}{
			\vdash \Gamma, \que A
		}
		\and
		\inferrule*[right=$\que c$]{
			\vdash \Gamma, \que A \\
			\vdash \Delta, \que A
			}{
				\vdash \Gamma, \Delta, \que A
			}
		\and
		\inferrule*[right=$\exc$]{
			\vdash \Gamma, B \\
			\vdash \negg{B}, \punit \\
			\vdash \negg{B}, A \\
			\vdash \negg{B}, B \parr B
		}{
			\vdash \Gamma, \exc A
		}
\end{mathpar}
The rules for $\que$ are \emph{distributed forms} of the structural rules, while the $\exc$ rule
gives a \emph{strong coexponential}, analogous to the strong version of $!$ described in the
previous section.. The corresponding `weak' coexponential is given by replacing the above $\exc$
rule with
\begin{mathpar}
	\inferrule*[right=$\exc$]{
		\vdash \bigotimes \que \Gamma, A
	}{
		\vdash \bigotimes \que \Gamma, \exc A
	}
\end{mathpar}
$\que \Gamma$ stands for the context obtained by applying $\que$ to every formula in $\Gamma$, and
$\bigotimes$ folds this context with a tensor. Unfortunately, the presence of this folding operation
means that this rule is not well-behaved in proof-theoretic terms.

\subsection{Exponentials vs. Coexponentials under Mix and Binary Cuts}

In fact, we can show that, in the presence of additional rules, (weak) exponentials and (weak)
coexponentials are interderivable up to provability. This result provides strong evidence that
coexponentials are the right abstraction for our purposes, essentially by showing that
\citet{wadler2014} and others implicitly use them in modelling server-client interactions within
CLL.

The requisite rules are \ruleref{Mix}, and one of the \emph{binary cut} or \emph{multicut} rules:
\begin{mathpar}
	\inferH{BiCut}{
		\vdash \Gamma, A, B \\
		\vdash \Delta, \negg{A}, \negg{B}
	}{
		\vdash \Gamma, \Delta
	}
	\and
	\inferH{MultiCut}{
		\vdash \Gamma, A_1, \dots, A_n \\
		\vdash \Delta, \negg{A_1}, \dots,  \negg{A_n}
	}{
		\vdash \Gamma, \Delta
	}
\end{mathpar}
\ruleref{BiCut} cuts two formulas at once, and \ruleref{MultiCut} an arbitrary number. These rules
were first proposed in the context of Linear Logic by \citet{abramsky1993interaction} in the compact
setting ($\otimes = \parr$). They are logically equivalent, but only the second one satisfies cut
elimination \cite[\S 4.2]{lindley2016}. We recall some folklore facts regarding the
interderivability of certain formulas and \rulenamestyle{Mix}-like inference rules. Recall that $C
\multimap D \defeq \negg{C} \parr D$. The following statements may be found across the relevant
literature
\cite{girard1987linear,abramsky1996interaction,bellin_subnets_1997,wadler2014,lindley2016}.

\begin{lemma} \label{lem:mix-folklore}
		The following rules are logically interderivable.
		\begin{enumerate}[label=(\roman*)]
			\item The axiom $\ounit \multimap \punit$ and the \rulenamestyle{Mix0} rule.
			\item The axiom $\punit \multimap \ounit$ and the \rulenamestyle{Mix} rule.
			\item The axiom $A \otimes B \multimap A \parr B$ and the \rulenamestyle{Mix} rule.
			\item The axiom $A \parr B \multimap A \otimes B$ and the \rulenamestyle{BiCut} rule.
			\item \rulenamestyle{BiCut} and \rulenamestyle{MultiCut}.
		\end{enumerate}
		Moreover, \rulenamestyle{Mix0} is derivable from the axiom rule $\vdash \negg{A}, A$ and
		\rulenamestyle{BiCut}.
\end{lemma}

Armed with this, we can prove that

\begin{theorem}
	In CLL with \ruleref{Mix} and \ruleref{BiCut}, exponentials and coexponentials coincide up to provability. That is: if we replace $\whynot$ and $\ofc$ in the rules for the exponentials with $\que$ and $\exc$ respectively, the resultant rule is provable using the coexponential rules, and vice versa.
\end{theorem}

This result extends to strong exponentials vs. strong coexponentials. The proof is even simpler:
under \ruleref{Mix} and \ruleref{BiCut} we have $\otimes = \parr$, so $F_A$, $H_A$ and $G_A$, $K_A$
are pairwise logically equivalent.

\section{Processes}
	\label{section:processes}

In the rest of the paper we will argue that the logical observations we made in \cref{section:fix}
have an interesting computational interpretation as server-client interaction. To this end we will
introduce a process calculus for CLL equipped with a bespoke form of strong coexponentials. Our
system shall introduce a certain amount of nondeterminism, yet it will remain \ruleref{Mix}-free.

We first explain how the coexponentials capture the intuitive shape of client pool formation
(\cref{section:que-means-client}). Following that, we briefly discuss three technical design
decisions that pertain to the coexponentials used in our system
(\cref{section:decision-strong,section:decision-lists,section:decision-permutation}). Finally, we
introduce the system in \cref{section:hcpcoexp}, and its metatheory in
\cref{section:hcpcoexp-metatheory}.

\subsection{$\que$ means client, $\exc$ means server}
	\label{section:que-means-client}

Recall the three rules for $\que$, namely
\begin{mathpar}
	\inferrule*[right=$\que w$]{
		\strut
	}{
		\vdash \que A
	}
	\and
	\inferrule*[right=$\que d$]{
		\vdash \Gamma, A
	}{
		\vdash \Gamma, \que A
	}
	\and
	\inferrule*[right=$\que c$]{
		\vdash \Gamma, \que A \\
		\vdash \Delta, \que A
		}{
			\vdash \Gamma, \Delta, \que A
		}
\end{mathpar}
We can read $\que A$ as the session type of a channel shared by a \emph{pool of clients}.
\begin{itemize}
	\item $\que w$ allows the vacuous formation of a empty client pool.
	\item $\que d$ allows the formation of a client pool consisting of exactly one client.
	\item $\que c$ rule can be used to \emph{aggregate} two client pools together.
\end{itemize}
The last point requires some elaboration. Each premise can be seen as a client pool with an external
interface ($\Gamma$ and $\Delta$ respectively). $\que c$ allows us to combine these into a single
process. This new process still behaves as a client pool, but it also retains \emph{both} external
interfaces. In contrast, the $\whynot c$ rule only allowed us to collapse two shared channel names
belonging to \emph{a single process}. Moreover, it did not allow us to mix two external
interfaces---one had to use \ruleref{Mix} for that.

Finally, the `weak' $\exc$ rule, i.e.,
\begin{mathpar}
	\inferrule{
		\vdash \bigotimes \que \Gamma, A
	}{
		\vdash \bigotimes \que \Gamma, \exc A
	}
\end{mathpar}
can be read as the introduction rule for a dual \emph{server session type}. It states that a process
serving $A$, and all of whose other interactions are as a client to a set of non-interacting
($\otimes$) services, can itself be `copromoted' to a \emph{server} $\exc A$.

Note that our intuitive explanations are almost identical to those of \citet{wadler2014}. The
difference is that our rules have the right branching structure to support them.

\subsection{Design Decision \#1: Server State and the Strong Rules}
	\label{section:decision-strong}

The first change with respect to the above interpretation is the switch to the \emph{strong rule}:
\begin{mathpar}
	\inferrule{
		\vdash \Gamma, B \\
		\vdash \negg{B}, \bot \\
		\vdash \negg{B}, A \\
		\vdash \negg{B}, B \parr B
	}{
		\vdash \Gamma, \exc A
	}
\end{mathpar}
This rule strongly evokes the structure of a `stateful' server serving $A$'s, with external
interface $\Gamma$. Within the server there exists an \emph{internal server protocol} $B$. This
comes with four ingredients: a process that provides a $B$, interacting along $\Gamma$
(initialization); a way to silently consume $B$ (finalization); a way to `convert' a $B$ to an $A$
(serving a client); and a way to fork one $B$ into two connected $B$'s (forking two subservers).

We decide to use this strong rule in order to avoid the \emph{uniformity} property that was
discussed in \cref{section:exp-as-fix}: the weak coexponential rule gives trivial servers providing
identical $A$'s to all clients. In contrast, this rule will allow a server to provide a different
$A$ each time it is called upon to do so.

\subsection{Design Decision \#2: Replacing Trees with Lists}
	\label{section:decision-lists}

The strong coexponential rule arose by taking the greatest fixed point of
\[
	H_A(\var{X}) \defeq \punit \with A \with (\var{X} \parr \var{X})
\]
As discussed in \cref{section:exp-as-fix,section:deriving-coexp}, this rule represents a
\emph{tree-like structure}. Nothing stops us from replacing it with a \emph{list-like}
structure.\footnote{It is worth noting that \citeauthor{girard1987linear} considered list-like
exponentials \citeyearpar[\S V.5(ii)]{girard1987linear}, but rejected them as they were not able to
reproduce contraction. This is not a requirement for modelling client-server interaction.} We use
the functors
\begin{align*}
	H'_A(\var{X}) &\defeq \punit \with (A \parr \var{X})
	&
	K'_A(\var{X}) &\defeq \ounit \oplus (A \otimes \var{X})
\end{align*}
and acquire the strong server rule derived from $H'_A$, viz.
\begin{mathpar}
	\inferrule{
		\vdash \Gamma, B \\
		\vdash \negg{B}, \punit \\
		\vdash \negg{B}, A \parr B \\
	}{
		\vdash \Gamma, \exc A
	}
\end{mathpar}

The main benefit is that the resulting system more closely reflects the pattern of server-client
interaction: clients form a queue rather than a tree, and servers no longer have to fork
subprocesses. This rule also requires fewer ingredients: an initialization of the internal protocol,
a finalization, and a component that spawns a session to serve one additional client.

To optimize this further, we make the $\parr$ implicit, and replace $\bot$ with a general $\Delta$
in the finalization:
\begin{mathpar}
	\inferH{Server}{
		\vdash \Gamma, B\\
		\vdash \negg{B}, \Delta \\
		\vdash \negg{B}, A, B
	}{
		\vdash \Gamma, \Delta, \exc A
	}
\end{mathpar}
This second rule can be immediately derived from the first one:
\begin{mathpar}
	\inferrule{
		\inferrule*{
			\vdash \Gamma, B \\
			\vdash \negg B, \Delta
		}{
			\vdash \Gamma, \Delta, B \otimes \negg B
		} \\
		\inferrule*{
			\inferrule*{
			}{
				\vdash \negg B, B
			}
		}{
			\vdash \negg B \parr B, \punit
		}\\
		\inferrule*{
			\inferrule*{
				\vdash \negg B, A, B \\
				\inferrule*{
				}{
					\vdash \negg B, B
				}
			}{
				\vdash \negg B, B, B \otimes \negg B, A
			}
		}{
			\vdash \negg B \parr B, A \parr (B \otimes \negg B)
		}
	}{
		\vdash \Gamma,\Delta,\exc A
	}
\end{mathpar}
There is a surreptitious twist here: the `new' internal server protocol is not $B$, but $B \otimes
\negg B$. This leads to internal back-and-forth communication in the server. $\Gamma$ is consumed to
produce a $B$. This is `passed' to each process serving each client. Finally, it is reflected back
to the initilization process, and `finalized' into a $\Delta$. The $\bot$ rule is invertible, so
instantiating $\Delta \defeq \bot$ in \ruleref{Server} gives back the preceding rule. Hence, these
two rules are logically equivalent.

\subsection{Design Decision \#3: Nondeterminism Through Permutation}
	\label{section:decision-permutation}

Using list-shaped rules for $\exc$ forces us to revise the rules for $\que$. To define a cut
elimination procedure they must now match the dual functor $K'_A$, and hence become
\begin{mathpar}
	\inferrule{\strut }{\vdash \que A}
	\and
	\inferrule{
		\vdash \Gamma, \que A \\
		\vdash \Delta, A
	}{
		\vdash \Gamma, \Delta, \que A
	}
\end{mathpar}
The cut elimination procedure for these rules leads to a confluent dynamics. This is unsatisfactory
from the perspective of client-server interaction: a proper model requires some nondeterminism in
the order in which clients are served. There are many ways to introduce this kind of behaviour. We
choose the simplest one: we identify derivations up to permutation of client formation in pools.
That is, we quotient them under the least congruence $\equiv$ generated from
\begin{mathpar}
	\inferrule{
		\inferrule*{
			\vdash \Gamma, \que A \\
			\vdash \Delta, A
		}{
			\vdash \Gamma,\Delta, \que A
		} \\
		\vdash \Sigma, A
	}{
		\vdash \Gamma,\Delta,\Sigma, \que A
	}
	\equiv
	\inferrule{
		\inferrule*{
			\vdash \Gamma, \que A \\
			\vdash \Sigma, A
		}{
			\vdash \Gamma,\Sigma, \que A
		} \\
		\vdash \Delta, A
	}{
		\vdash \Gamma,\Delta,\Sigma, \que A
	}
\end{mathpar}
This amounts to quotienting lists up to permutation. Thus, when a client pool interacts with a
server, the cut elimination procedure may silently choose to serve any of the constituent clients.

\paragraph{Trees and nondeterminism}

The careful reader might notice that the original, tree-like `distributed contraction' rule $\que c$
inherently supported a certain amount of nondeterminism: if we were to quotient derivations up to
permutation of the premises of $\que c$, then the cut elimination procedure would have some choice
of whether to serve the left or right subtree first. Switching to list-like functors forbids this
move, and seemingly imposes a much stricter discipline.

Nevertheless, the tree structure is awkward and rigid. For example, consider a client pool whose
tree structure is informally $[[c_0, c_1], [c_2, c_3]]$. As nondetermistic choices are only made at
each node, the clients cannot be served in any order. For example, if $c_0$ is served first then
$c_1$ must be served next, as it is in the same subtree. From a conventional client-server
perspective this is arguably \emph{not} a sufficient amount of nondeterminism. In contrast, our
formulation allows full permutations of the client pool.

\subsection{Introducing CSLL}
	\label{section:hcpcoexp}

Based on the above considerations, we introduce the system \hcpcoexp~of \emph{Client-Server Linear
Logic}.

Following recent presentation of CLL-based systems of session types \cite{kokke2019better},
\hcpcoexp~is structured around \emph{hyperenvironments}. Intuitively, the logical system underlying
\hcpcoexp~is not one-sided sequent calculus like CP, but a \emph{hypersequent system}
\cite{avron1991hypersequents}. In this kind of presentation process constructors are more finely
decoupled. For example, the original CP output/$\otimes$ constructor $\Out*{x}{y}{\Par{P}{Q}}$ h a
combination of a parallel composition with an output prefix. Hypersequent systems allow us to
separately type to these two constructs, and brings the language closer to $\pi$-calculus.

One-sided sequent systems for CLL---such as \citeauthor{girard1987linear}'s original presentation
\citeyearpar{girard1987linear}---use sequents of the form $\vdash \Gamma$ where $\Gamma$ is an
\emph{environment}, i.e. an unordered list of formulas. We assign distinct \emph{names} to each
formula. The environment $\Gamma = \Name{x_1}: A_1, \dots, \Name{x_n}:A_n$ stands for $A_1 \parr
\dots \parr A_n$. Hence, a comma stands for $\parr$. Environments are identical up to permutation.
We write $\pempty$ for the empty one.

A \emph{hyperenvironment} adds another layer: it is an unordered list of environments. We separate
environments by vertical lines. If each environment $\Gamma_i$ stands for the formula $A_i$, the
hyperenvironment $\Ctxs{G} = \Gamma_1 \mid \dots \mid \Gamma_n$ stands for the formula $A_1 \otimes
\dots \otimes A_n$. Hence, $\mid$ stands for $\otimes$. Hyperenvironments are identical up to
permutation, and we write $\emptyset$ for the empty one. We also stipulate that variable names be
distinct within and across environments.

\begin{figure}
	\input{proc-syntax}
	\input{proc-types}
	\caption{The syntax and type system of \hcpcoexp.}
	\label{figure:proc-syntax}
\end{figure}
The syntax and the type system of \hcpcoexp~are defined in \cref{figure:proc-syntax}. The types
are the formulas of CLL. Processes have the following binding structure:
\begin{itemize}
	\item $\Name{x}$ and $\Name{y}$ are bound in $\Proc{P}$ within $\New{x}{y}{P}$.
	\item $\Name{x}$ and $\Name{y}$ are bound in $\Proc{P}$ within $\In{y}{x}{P}$, $\Out{y}{x}{P}$.
	\item $\Name{x}$ is bound in $\Proc{P}$ within $\Inl{x}{P}$,$\Inr{x}{P}$.
	\item $\Name{x}$ is bound in both $\Proc{P}$ and $\Proc{Q}$ within $\Case{x}{P}{Q}$.
	\item Within $\ExcP{y}{P}{i}{f}{z}{z'}{y'}{Q}$, $\Name{i}$ and $\Name{f}$ are bound in $\Proc{P}$, $\Name{z}$,  $\Name{z'}$, and $\Name{y'}$ are bound in $\Proc{Q}$.
	\item $\Name{x}$ and $\Name{y}$ are bound in $\Proc{P}$ within $\QueA{x}{y}{P}$.
\end{itemize}
As is usual, processes are identified up to $\alpha$-equivalence.

In all processes that involve a dot not in curly braces, we call the part that precedes it the
\emph{prefix} of the process, and the part that succeeds it the \emph{continuation}. E.g. in the
process $\In{y}{x}{P}$, the prefix is $\Name{y}\Procparens{\Name{x}}$, and the continuation is
$\Proc{P}$. We write $\Pre{y}$ for an arbitrary prefix communicating on channel $\Name{y}$, and
$\bn{\Pre{y}}$ for the variables that it binds in its continuation. For example, we may write
$\Pre{y} \defeq \Name{y}\Procparens{\Name{x}}$, and $\bn{\Pre{y}} = {\Name{x}}$. We write
$\fn{\Proc{P}}$ for the free names in a process $\Proc{P}$.

A generic judgment of the type system has the shape $\IsProc{P}{\Ctxs{G}}$ where $\Proc{P}$ is a
process, and $\Ctxs{G}$ is a hyperenvironment. Most rules are identical to HCP, and in the interest
of brevity we only discuss those that are important. Similarly to HCP, the typing cannot be inferred
from the terms alone. For example, in $\ruleref{M-False}$ the term $\In{x}{}{P}$ does not specify
the environment $\Gamma$ on which it acts.

Hyperenvironment components are introduced by the nullary and binary \emph{hypermix} rules,
\ruleref{HMix0} and \ruleref{HMix2}. These rules are `Mix' rules only in name. \ruleref{HMix2} forms
the disjoint parallel composition of two processes: their environments are joined with $\mid$, which
stands for $\otimes$.\footnote{\ruleref{Mix} would join them with a comma, which would stand for a
$\parr$.} \ruleref{HMix0} is the stopped process; its hyperenvironment is the empty one, which
stands for the unit of $\otimes$, namely $\ounit$.\footnote{\ruleref{Mix0} would stand for the unit
  of $\parr$, namely $\punit$.}

The \ruleref{Cut} and \ruleref{Tensor} rules eliminate hyperenvironment components. The premises of
\ruleref{Cut} ensure that the two variables that are being connected---viz. $\Name{x}$ and
$\Name{y}$---are in different `parallel components' of $\Proc{P}$. Notice that the external
environments of these two components, namely $\Gamma$ and $\Delta$, are then brought together in the
conclusion. A similar pattern permeates the \ruleref{Tensor} and $\ruleref{M-True}$ rules. It is
instructive to follow the derivation of the original CP rules for $\otimes$ and $\ounit$, which we
will silently use:
\begin{mathpar}
	\inferrule*{
		\inferrule*{
			\IsProc{P}{\Gamma, \Name{y} : A} \\
			\IsProc{Q}{\Delta, \Name{x} : B}
		}{
			\IsProc{\Par{P}{Q}}{\Gamma, \Name{y} : A \mid \Delta, \Name{x} : B}
		}
	}{
		\IsProc{\Out*{x}{y}{\Par{P}{Q}}}{\Gamma, \Delta, \Name{x} : A \otimes B}
	}
	\and
	\inferrule*{
		\inferrule*{ }{\IsProc{\Stop}{\emptyset}}
	}{
		\IsProc{\Out{x}{}{\Stop}}{\Name{x} : \ounit}
	}
\end{mathpar}

The exponential rules \ruleref{WhyNotW}, \ruleref{WhyNotD}, \ruleref{WhyNotC} and \ruleref{OfCourse}
are formulated in the style of \citet{kokke2019better}. In \ruleref{OfCourse} we use vector notation
($\vec{-}$) as a shorthand for lists of names and types. Note that---in contrast to all previous
systems---we notate $\Proc{P}$ as a parameter rather than as the continuation in the process
$\OfCP{x}{\vec{y}}{P}$. This because $\Proc{P}$ does not behave like a continuation. For example, it
has its own distinct commuting conversion.

The coexponential rules \ruleref{QueW}, \ruleref{QueA} and \ruleref{Claro} follow the patterns
described in
\cref{section:que-means-client,section:decision-strong,section:decision-lists,section:decision-permutation}.
The rule \ruleref{QueW} (W stands for `weaken') constructs an empty client pool. The rule
\ruleref{QueA} (A stands for `absorb') combines a client and a pool into a slightly larger pool. The
interfaces of the client pool and the client are necessarily disjoint, as they are separated by a
$\mid$ in the premise. All the processes in the resultant pool race to communicate with a server at
the single endpoint $\Name{x}$.

Correspondingly, \ruleref{Claro} constructs a process that offers a service at the single endpoint
$\Name{y}$. Its continuation $\Proc{P}$ functions as both the initialization and the finalization of
the server, over channels $\Name{i}$ and $\Name{f}$ respectively. In terms of the \ruleref{Server}
rule of \cref{section:decision-lists}, it combines the premises $\vdash \Gamma, B$ and $\vdash \negg
B, \Delta$ into one process. However, these functionalities continue to be disjoint components of
$\Proc{P}$, as their interfaces are separated by a $\mid$ in the premise. The process $\Proc{Q}$ is
a `worker process' which is spawned every time a client is to be served.

\subsection{Operational Semantics and Metatheory}
	\label{section:hcpcoexp-metatheory}

\begin{definition}\label{def:canonical}
	Canonical terms are defined by the following clauses.
	\begin{itemize}
		\item $\Pre{x}[P]$ is canonical whenever $\Proc{P}$ is.
		\item $\Par{P}{Q}$ is canonical if both $\Proc{P}$ and $\Proc{Q}$ are canonical.
		\item $\Stop$ and $\Link{x}{y}$ are canonical.
		\item $\Case{y}{P}{Q}$ and $\OfCP{x}{\vec{y}}{P}$ are canonical.
	\end{itemize}
\end{definition}
\noindent In particular, $\New{x}{y}{P}$ is \emph{not} canonical; it is a cut.

The above notion of canonicity is not definitive. For example, $\Pre{x}[P]$ could have been
considered canonical regardless of the canonicity of $\Proc{P}$ (similar to weak head normal form
for $\lambda$-calculus). However, we choose to react $\Proc{P}$ further to make the `final result'
of an interaction visible in later examples. In addition, we could require terms such as $\Proc{P}$
and $\Proc{Q}$ in $\Case{y}{P}{Q}$ be canonical for the whole term to be canonical, but we choose
not to so as to reduce the number of reaction rules.

\begin{figure}
	\input{proc-equiv}
	\caption{The structural equivalence of \hcpcoexp~processes.}
	\label{figure:proc-equiv}
\end{figure}

\begin{figure}
	\input{proc-opsem}
	\caption{The operational semantics of \hcpcoexp~processes.}
	\label{figure:proc-opsem}
\end{figure}

We define the notion of \emph{structural equivalence} $\Proc{P} \equiv \Proc{Q}$ to be the least
congruence between processes induced by the clauses in \cref{figure:proc-equiv}. Furthermore, we
define the \emph{reaction relation} $\red{P}{Q}$ between processes to be the least relation induced
by the clauses in \cref{figure:proc-opsem}.

The structural equivalence and the reaction semantics largely mirror the notions of the same name in
the $\pi$-calculus \cite{milner_functions_1992,milner_communicating_1999}. Those that differ are
justified \emph{via} linear logic. \ruleref{Res-Pre} and \ruleref{Pre-Pre} can be seen as
identifications arising from \emph{proof nets}, in which the corresponding proofs would be
graphically identical. \ruleref{OfCourse-Comm} and \ruleref{With-Comm} are commuting conversions for
$\ruleref{OfCourse}$ and $\ruleref{With}$ respectively. \ruleref{Pre-Par} with
\ruleref{Res-Pre} combine into a kind of commuting conversion for prefixes. We take the former as
reaction rules, and the latter as structural equivalences. The reason is that we wish to make
equivalence preserve canonicity. For example, in \ruleref{With-Comm} the LHS is not canonical, but
the RHS is.

The overwhelming majority of these commuting conversions is used in previous works on the
relationship between linear logic and $\pi$-calculus to obtain cut elimination \citep[\S
3.6]{wadler2014} \citep[\S 3]{bellin1994p}. Perhaps the only exception is \ruleref{Pre-Pre}, which
allows us to swap any two noninterfering prefixes. It can be justified computationally as an
observational equivalence arising from the semantics of \citet[\S 5]{atkey2017observed}. Finally,
\citet{kokke2019better} view it as a session-theoretic version of \emph{delayed actions}
\cite{merro_2004}.

The structural equivalence \ruleref{Que-Que} allows us to commute the position of two clients in the
pool, thereby imitating racing---as discussed in \cref{section:decision-permutation}. Note that to
fully exploit the nondeterminism induced by \ruleref{Que-Que} the other structural equivalences are
necessary. For example, the two clients in $\QueA{x}{x_0}{\In{y}{y'}{\QueA{x}{x_1}{P}}}$ cannot be
permuted without using \ruleref{Pre-Pre} first.

\ruleref{Pre} corresponds to eliminating non-top-level cuts in Linear Logic; it is not standard in
either $\pi$-calculus or CP. Nevertheless, we choose to include it in order to strengthen our notion
of canonical form, which in turn elucidates the examples in \cref{section:pools}. In contrast, the
reaction rules for the exponentials are standard; see \cite{kokke2019better}.

The rule \ruleref{Claro-QueW} corresponds to serving an empty client pool. In this case we simply
`short-circuit' the initialization and finalization channels of $\Proc{P}$. Likewise, the rule
\ruleref{Claro-QueA} is the reaction caused by a nonempty pool of clients. The pool offers a fresh
channel $\Name{x'}$ on which the new client expects to be served. The server then spawns a worker
process $\Proc{Q}$, and the channel $\Name{y'}$ on which it will serve the new client is connected
to $\Name{x'}$, as expected. The initialization channel $\Name{i}$ of the server continuation is
connected to the $\Name{z}$ channel, on which the worker process expects to receive the `current
state' of the server. Once $\Proc{Q}$ serves the client, it will send the `next state' of the server
on $\Name{z'}$. Thus, we re-instantiate the server with $\Name{z'}$ as the new initialization
channel. Note that the `server state' we discuss here does not conform to the usual intuition of an
immutable value; it could be a session type itself, as demonstrated by the example in
\cref{section:voting-game}.

We have the following metatheoretic results.
\begin{lemma}\label{lem:equiv-preserves-types}
	If $\Equiv{P}{Q}$, then $\IsProc{P}{\Ctxs{G}}$ if and only if $\IsProc{Q}{\Ctxs{G}}$.
\end{lemma}
\begin{theorem}[Preservation]\label{thm:preservation}
	If $\IsProc{P}{\Ctxs{G}}$ and $\red{P}{Q}$, then $\IsProc{Q}{\Ctxs{G}}$.
\end{theorem}
\begin{theorem}[Progress]\label{thm:progress}
	If $\IsProc{R}{\Ctxs{G}}$ then either $\Proc{R}$ is canonical, or there exists $\Proc{R'}$ such that $\red{R}{R'}$.
\end{theorem}

\section{An example: Compare-and-Set}
	\label{section:pools}

We now wish to demonstrate the client-server features of \hcpcoexp. To do so we produce an
implementation of the quintessential example of a synchronization primitive, the
\emph{Compare-and-Set operation} (CAS) \cite[\S 5.8]{herlihy_art_2012}. Higher-level examples are
given in \cref{section:csgv}.

A register that supports compare-and-set comes with an operation $\textsc{Cas}(e, d)$ which takes two
values: the \emph{expected} value $e$, and the \emph{desirable} value $d$. The function compares the
expected value $e$ with the register. If the two differ, the value of the register remains put, and
$\textsc{Cas}(e, d)$ returns false. But if they are found equal, the register is updated with the
desirable value $d$, and $\textsc{Cas}(e, d)$ returns true. When multiple clients are trying to
perform CAS operations on the same register, these must be performed \emph{atomically}. The CAS
operation is very powerful: an asynchronous machine that supports it can implement all concurrent
objects in a wait-free manner.

We follow previous work
\cite{girard1987linear,abramsky_computational_1993,lindley2016,kokke2019better} and define the type
of Boolean sessions to be $\boom \defeq \ounit \oplus \ounit$. We have the following derivable
constants:
\begin{align*}
	\TrueP{z}  &\defeq \IsProc{\Inl{z}{\Out{z}{}{\Stop}}}{\Name{z} : \boom}
		&
	\FalseP{z} &\defeq \IsProc{\Inr{z}{\Out{z}{}{\Stop}}}{\Name{z} : \boom}
\end{align*}
Moreover, we obtain the following derivable `elimination' rule (we write derivable rules in \Deriv{blue}):
\begin{mathpar}
	\inferrule*{
		\inferrule*{
			\IsProc{P}{\Gamma}
		}{
			\IsProc{\In{z}{}{P}}{\Name{z} : \bot, \Gamma}
		 } \\
		\inferrule*{
			\IsProc{Q}{\Gamma}
		}{
			\IsProc{\In{z}{}{Q}}{\Name{z} : \bot, \Gamma}
		}
	}{
		\IfO{z}{P}{Q} \defeq
			\IsProc{\Case{z}{\In{z}{}{P}}{\In{z}{}{Q}}}{\Name{z} : \negg{\boom}}, \Gamma
	}
\end{mathpar}
Hence, we can eliminate a Boolean channel in any environment $\Gamma$. The induced reactions are
\begin{align*}
	\New{x}{y}{
		\Par*{
			\TrueP{x}
		}{
			\IfO{y}{P}{Q}
		}
	}
	&\red*{}{}
	\Par{\Stop}{P}
		\equiv
	\Proc{P}
		& 										%
	\New{x}{y}{
		\Par*{
			\FalseP{x}
		}{
			\IfO{y}{P}{Q}
		}
	}
	&\red*{}{}
	\Par{\Stop}{Q}
		\equiv
	\Proc{P}
\end{align*}

We can now implement a register with a CAS operation. To begin, each client communicates with the
register along a channel of type
\[
	A \defeq \boom \otimes \boom \otimes \negg{\boom} \parr \ounit
\]
Thus, a client outputs three channels. On the first two it shall to receive the expected and
desirable values. On the third it will input a boolean, namely the success flag of the CAS
operation. Following that, it will accept an end-of-session signal. Curiously, this last step is
necessary for our implementation to type-check.

As a minimal example we will construct a pool of two racing clients, one performing
$\textsc{Cas}(\textsf{ff}, \textsf{tt})$, and the other one $\textsc{Cas}(\textsf{tt},
\textsf{ff})$. Initially $\Name{x_1}$ is ahead in the client pool.
\begin{align*}
		\Proc{C_0} &\defeq
			\IsProc{
				\Out{x_0}{x_e}{
					\Out{x_0}{x_d}{
						\Par*{
							\FalseP{x_e}
						}{
							\Par{
								\TrueP{x_d}
							}{
								\Link{x_0}{r_0}
							}
						}
					}
				}
			}{\Name{x_0} : \boom \otimes \boom \otimes \negg \boom \parr \ounit, \Name{r_0} : \boom \otimes \punit}
		\\
		\Proc{C_1} &\defeq
			\IsProc{
				\Out{x_1}{x_e}{
					\Out{x_1}{x_d}{
						\Par*{
							\TrueP{x_e}
						}{
							\Par{
								\FalseP{x_d}
							}{
							\Link{x_1}{r_1}
							}
						}
					}
				}
			}{\Name{x_1} : \boom \otimes \boom \otimes \negg \boom \parr \ounit, \Name{r_1} : \boom \otimes \punit}
	\\
	\Procsf{clients} &\defeq
		\IsProc{
			\QueA{x}{x_1}{\QueA{x}{x_0}{\QueW{x}{\Par*{C_0}{C_1}}}}
		}{
			\Name{x} : \que \left(\boom \otimes \boom \otimes \negg \boom \parr \ounit \right),
			\Name{r_0} : \boom \otimes \punit,
			\Name{r_1} : \boom \otimes \punit
		}
\end{align*}
Note that each client forwards the result it receives to an individual channel $\Name{r_i}$.
By the \ruleref{QueA} rule these two channels are preserved in the final interface of the pool.

Next we define the CAS register process, for which we use the $\exc$ connective.
This requires two components: the initialization and finalization process $\Proc{P}$, and the worker process $\Proc{Q}$ that serves one client.
To begin, we pick the internal server state to be $B \defeq \boom$.
We initialize the register to false, and forward the final state of the register to $\Name{u}$.
\[
	\Proc{P} \defeq
		\IsProc{\Par*{\FalseP{i}}{\Link{f}{u}}}
					 {\Name{i} : \boom \mid \Name{f} : \negg \boom, \Name{u} : \boom}
\]
Finally, we define $\Proc{Q}$. We begin by receiving the input and output channels from a client, and do a
case analysis on the current state of the register:
\[
	\Proc{Q} \defeq
		\IsProc{
			\In{y'}{y_e}{\In{y'}{y_d}{\IfO{z}{\Proc{R_1}}{\Proc{R_0}}}}
		}{
			\Name{z} : \negg \boom,
			\Name{y'}: \negg \boom \parr \negg \boom \parr \boom \otimes \punit,
			\Name{z'} : \boom
		}
\]
We have carefully named the channels so that $\Name{y_e} : \negg{\boom}$ and $\Name{y_d} :
\negg{\boom}$ carry the expected and desirable values. $\Name{z'}$ and $\Name{w'}$ carry the internal register, before and after the operation. The continuations $\Proc{R_0}$ and $\Proc{R_1}$ do a case analysis on the expected and desired value:
\begin{align*}
	\Proc{R_1} \defeq\ &\IsProc{\IfO{y_e}{\IfO{y_d}{S_{111}}{\Proc{S_{110}}}}{\IfO{y_d}{\Proc{S_{101}}}{\Proc{S_{100}}}}}{\Name{y_e}:\negg \boom, \Name{y_d}:\negg \boom, \Name{y'}:\boom \otimes \punit, \Name{z'}:\boom}
	\\
	\Proc{R_0} \defeq\ &\IsProc{\IfO{y_e}{\IfO{y_d}{\Proc{S_{011}}}{\Proc{S_{010}}}}{\IfO{y_d}{\Proc{S_{001}}}{\Proc{S_{000}}}}}{
		\Name{y_e}:\negg \boom, \Name{y_d}:\negg \boom, \Name{y'}:\boom \otimes \punit, \Name{z'}:\boom
	  }
\end{align*}
Two further case analyses lead to an exhaustive eight cases, each of which is handled by a separate process
$\Proc{S_{ijk}}$. We only give $\Proc{S_{110}}$ here, the rest being analogous:
\[
	\Proc{S_{110}} \defeq
		\IsProc{\Out{y'}{y_r}{\Par*{\TrueP{y_r}}{\In{y'}{}{\FalseP{z'}}}}}
		{\Name{y'}:\boom \otimes \punit, \Name{z'}:\boom}
\]
In this case, the expected value (true) matches the register state (true), so the process outputs true to the result channel $\Name{y_r}$ (the CAS operation succeeds), and the register is set to the desired value (false).
We must not forget to receive an end-of-session signal on $\Name{y}$, as required by the session type.
We let $\Procsf{server} \defeq \IsProc{\ExcP{y}{P}{i}{f}{z}{z'}{y'}{Q}}{\Name{y}:\exc (\negg \boom \parr \negg \boom \parr \boom \otimes \punit), \Name{u}:\boom}$, and cut:
\begin{align*}
	& \New{x}{y}{\Par*{\Procsf{clients}}{\Procsf{server}}} \\
	=\ & \New{x}{y}{\Par*{\QueA{x}{x_1}{\QueA{x}{x_0}{\QueW{x}{\Par*{\Proc{C_0}}{\Proc{C_1}}}}}}
		{\Procsf{server}}} \\
	\equiv\ & \New{x}{y}{\Par*{\QueA{x}{x_0}{\QueA{x}{x_1}{\QueW{x}{\Par*{\Proc{C_0}}{\Proc{C_1}}}}}}
		{\Procsf{server}}} \tag{$\Name{x_0}$ preempts $\Name{x_1}$ using \ruleref{Que-Que}} \\
	\red{}{}\ & \New{x}{y}{\New{x_0}{y'}{
		\Par*{
			\Par{\Proc{C_0}}{\QueA{x}{x_1}{\QueW{x}{\Proc{C_1}}}}
		}{
		\ExcP*{y}{
				\New{i}{z}{
				\Par*{\Proc{P}}{\Proc{Q}}
				}
		}{z'}{f}{z}{z'}{y'}{\Proc{Q}}
		}
	}} \tag{$\Proc{C_0}$ is accepted} \\
	\red*{}{}\ & \Out{r_0}{y_r}{\Par*{
		\TrueP{y_r}
	}{
		\In{r_0}{}{
	\New{x}{y}{
		\Par*{
			\QueA{x}{x_1}{\QueW{x}{\Proc{C_1}}}
		}{
			\ExcP{y}{
				P'
			}{z''}{f}{z}{z'}{y'}{\Proc{Q}}
		}
	}}}} \tag{$\Proc{C_0}$ performs CAS}\\
	\red{}{} & \Out{r_0}{y_r}{\Par*{
		\TrueP{y_r}
	}{
		\In{r_0}{}{
			\New{x}{y}{\New{x_1}{y'}{
				\Par*{
					\Par{\Proc{C_1}}{\QueW{x}{\Stop}}
				}{
					\ExcP*{y}{
						\New*{z''}{z}{
							\Par{\Proc{P'}}{\Proc{Q}}
						}
					}{z'}{f}{z}{z'}{y'}{\Proc{Q}}
				}
			}}
		}
	}}
	\tag{$\Proc{C_1}$ is accepted} \\
	\red*{}{} & \Out{r_0}{y_r}{\Par*{
		\TrueP{y_r}
	}{
		\In{r_0}{}{
	\Out{r_1}{y_r}{\Par*{
		\TrueP{y_r}
	}{
		\In{r_1}{}{
			\New{x}{y}{
				\Par*{
					\QueW{x}{\Stop}
				}{
					\ExcP{y}{
						P''
					}{z'''}{f}{z}{z'}{y'}{Q}
				}
			}
		}
	}}}}}
	\tag{$\Proc{C_1}$ performs CAS} \\
	\red{}{} &
	\Out{r_0}{y_r}{\Par*{
		\TrueP{y_r}
	}{
		\In{r_0}{}{
	\Out{r_1}{y_r}{\Par*{
		\TrueP{y_r}
	}{
		\In{r_1}{}{
			\Par*{\Stop}{\New{z'''}{f}{P''}}
		}
	}}}}}
	 \tag{$\Procsf{server}$ starts to finalize} \\
	\red*{}{} & \IsProc{\Out{r_0}{y_r}{\Par*{
		\TrueP{y_r}
	}{
		\In{r_0}{}{
			\Out{r_1}{y_r}{\Par*{
				\TrueP{y_r}
			}{
				\In{r_1}{}{\FalseP{u}}
			}}
		}
	}
	}}{\Name{r_0}:\boom \otimes \punit, \Name{r_1}:\boom \otimes \punit, \Name{u}:\boom}
	\tag{$\Procsf{server}$ finalizes}
\end{align*}
where $	\Proc{P'} = \Par{\TrueP{z''}}{\Link{f}{u}}$ and $\Proc{P''} =
\Par{\FalseP{z'''}}{\Link{f}{u}} $. This corresponds to the scenario where $\Proc{C_0}$ wins the
first race, and hence the CAS operation of both clients suceeds. There is another reaction
sequence: if $\Proc{C_1}$ wins the first race, we end up with
$\Out{r_1}{y_r}{\Par*{
	\FalseP{y_r}
}{
	\In{r_1}{}{
		\Out{r_0}{y_r}{\Par*{
			\TrueP{y_r}
		}{
			\In{r_0}{}{\TrueP{u}}
		}}
	}
}
}$.

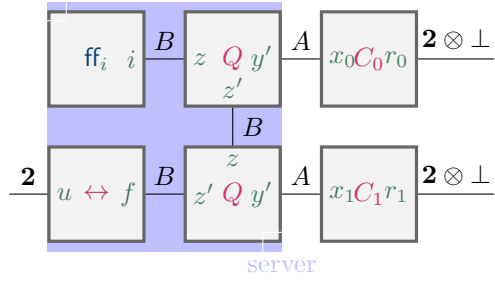
\begin{wrapfigure}{R}{0.5\textwidth}
	\begin{tikzpicture}[
		process/.style={rectangle, draw=black!60, fill=black!5, very thick, minimum size=12.5mm},
		null/.style={rectangle,draw=black!0, minimum size = 5mm},
		node distance=0.5cm
	]
	\node[process] (init) {$\FalseP{i}$};
	\node[process] (fini) [below=of init] {$\Link{}{}$};
	\node[process] (Q0) [right=of init] {$\Proc{Q}$};
	\node[process] (Q1) [right=of fini] {$\Proc{Q}$};
	\node[process] (C0) [right=of Q0] {$\Proc{C_0}$};
	\node[process] (C1) [right=of Q1] {$\Proc{C_1}$};
	\node[null] (fini-env) [left = of fini] {};
	\node[null] (C0-env) [right= 1cm of C0] {};
	\node[null] (C1-env) [right= 1cm of C1] {};
	\node[null] (topleft) [above=of init.west] {};
	\node[null] (botright) [below=of Q1.east] {};

	\draw (init.east) node [left] {$\Name{i}$} -- (Q0.west) node [right] {$\Name{z}$}
						node [midway, above] {$B$}
								   ;
	\draw (Q0.south) node [above] {$\Name{z'}$} -- (Q1.north) node [below] {$\Name{z}$}
						node [midway, right] {$B$};
	\draw (Q1.west) node [right] {$\Name{z'}$} -- (fini.east)  node [left] {$\Name{f}$}
						node [midway, above] (fini-Q1) {$B$};
	\draw (Q0.east) node [left] {$\Name{y'}$} -- (C0.west)  node [right] {$\Name{x_0}$}
						node [midway, above] (Q0-C0) {$A$};
	\draw (Q1.east) node [left] {$\Name{y'}$} -- (C1.west) node [right] {$\Name{x_1}$}
						node [midway, above] (Q1-C1) {$A$};
	\draw (fini.west) node [right] {$\Name{u}$} -- (fini-env.east) node [midway, above]  {$\boom$};
	\draw (C0-env.west) -- (C0.east) node [left] {$\Name{r_0}$} node [midway, above] (C0-C0-env) {$\boom \otimes \bot$};
	\draw (C1-env.west) -- (C1.east) node [left] {$\Name{r_1}$} node [midway, above] (C1-C1-env) {$\boom \otimes \bot$};
	\begin{pgfonlayer}{bg}    %
		\fill[blue!25] (topleft) rectangle (botright) node [below] {server};
    \end{pgfonlayer}
	\end{tikzpicture}
	\caption{Topology of Compare-and-Set protocol, after two server acceptances. Boxes represent processes. Cuts  are represented by edges connecting two channels. The dual of each session type is omitted for simplicity.}
	\label{figure:cas-layout}
\end{wrapfigure}

The coexponentials play a central r\^{o}le here: $\exc$ is used to represent the fact that this
register provides a server session at a unique end point, and $\que$ is used to collect requests for
a CAS operation to this single end point. We see that every feature of client-server interaction, as
described in points (i)--(iv) of \cref{section:problem}, is modelled. 

The fact we are able to implement a synchronization primitive like CAS shows that the server-client
rules also provide an additional safeguard, namely that \emph{server acceptance is atomic}. While
the actual CAS is \emph{not} an atomic operation---as many things are happening in parallel---the
causal flow of information ensures that the state implicitly remains atomic. To illustrate the type
of atomicity we have, consider an alternative reaction sequence where the two clients are
immediately accepted before any other reaction. \cref{figure:cas-layout} shows the process topology
of the scenario where $\Proc{C_0}$ is accepted immediately before $\Proc{C_1}$. Each client is
connected to the one of the two worker processes $\Proc{Q}$ with client protocol $A$, and the worker
processes are connected to each other and $\Proc{P}$ with internal server protocol $B$. Which
specific worker process a client connects to is determined by the client's position in the queue,
before the coexponential reaction \ruleref{Claro-QueA} takes place. The clients' positions in the
layout also determine the final result of the reaction up to structural equivalence, even before the
computation of the output takes place.

\section{A session-typed language for server-client programming}
  \label{section:csgv}

As the example of the previous section shows, \hcpcoexp~is a particularly low-level language. This
is a feature of essentially all variants of linear logic as used for session typing, including
\citeauthor{kokke2019better}'s HCP \citeyearpar[Example 2.1]{kokke2019better}, and Wadler's CP
\cite[\S 2.1]{atkey2017observed} \cite[\S 3.1]{lindley2016}. Consequently, the need arises for
higher-level notation to help us write richer examples that illustrate the degree of channel sharing
in CSLL. We follow the lead of \citet[\S 4]{wadler2014} and introduce a higher-level, session-typed
functional language, which we call CSGV.

CSGV is a linear $\lambda$-calculus augmented with session types and communication primitives. It is
based on the influential work of \citet{gay2010linear}. Over the past decade many variations of this
language have been proposed; see e.g. \citet{lindley2015semantics,
lindley2016talking,lindley_lightweight_2017} and \citet{fowler2019exceptional}. CSGV extends the
version given by \citeauthor{wadler2014} with primitives for client-server interaction. Like the
approach in \emph{loc. cit.} we do not directly endow CSGV with a semantics. Instead, we formulate a
type-preserving translation into \hcpcoexp, which indirectly provides an execution mechanism.
Naturally, the client-server primitives translate to the coexponential rules of CSLL.

\subsection{Source Language and the translation}

\paragraph{Types}

The types of CSGV consist of standard functional types and session types. While the former are used
to classify values, the latter are used to describe the behaviour of channels. Compared to
\citet{wadler2014} we have added sum types, and session types for server-client shared channels.
\begin{align*}
	T, \dots\ \Coloneqq\  \quad
					& T \multimap T \mid
				T \rightarrow T \mid
				T + T \mid
				T \otimes T \mid
				\textsf{Unit} \mid
				T_S \\
	T_S, \dots\ \Coloneqq \quad
		& \send T.T_S \tag{output value of type $T$, then behave as $T_S$}\\
		\mid \ & \recv T.T_S \tag{input value of type $T$, then behave as $T_S$}  \\
		\mid \ & T_S \oplus T_S \tag{select from options} \\
		\mid \ & T_S \with T_S \tag{offer choice} \\
		\mid \ & \endp \mid \endm \tag{end-of-session} \\
		\mid \ & \que T_S \tag{request $T_S$ session} \\
		\mid \ & \exc T_S \tag{serve $T_S$ session}
\end{align*}
Both the functional types and the session types of CSGV are translated to the linear types of CSLL.
The functional part closely follows \citeauthor{wadler2014} in using the `call-by-value' embedding
of intuitionistic logic into linear logic \cite{benton_linear_1996,maraistcall1995,maraist1999}. The
session types are translated as follows:
\begin{align*}
	\bra*{\send T.T_S} &\defeq\   \negg{\bra*{T}} \parr \bra*{T_S}
	&
	\bra*{T_S \with U_L} &\defeq\  \bra*{T_S} \oplus \bra*{U_L}
	&
	\bra*{\endm} &\defeq\   \punit
	\\
	\bra*{\recv T.T_S} &\defeq\  \bra*{T} \otimes \bra*{T_S}
	&
	\bra*{T_S \oplus U_L} &\defeq\  \bra*{T_S} \with \bra*{U_L}
	&
	\bra*{\endp} &\defeq\  \ounit
	\\
	\bra*{\que T_S} &\defeq\  \exc \bra*{T_S}
	&
	\bra*{\exc T_S} &\defeq\  \que \bra*{T_S}
\end{align*}
As noted by \citet[\S 4.1]{wadler2014}, the connectives translate to the dual of what one might
expect. The reason is that channels are used in the opposite way. Consider the session type $\send
T.S$: sending a value in CSGV is translated as inputting a channel on which you can send it in CSLL.
Similarly, providing a service $\exc S$ is interpreting as inputting a channel on which the result
will be served.

\paragraph{Duality}

We define duality on session types in the standard way; it is obviously an involution.
\begin{align*}
	\overline{\send T.T_S} &\defeq\  \recv T.\overline{T_S}
	&
	\overline{\send T.T_S} &\defeq\  \recv T.\overline{T_S}
	&
	\overline{T_S \oplus U_L} &\defeq\  \overline{T_S} \with \overline{U_L}
	\\
	\overline{T_S \with U_L} &\defeq\  \overline{T_S} \oplus \overline{U_L}
	&
	\overline{\que T_S} &\defeq\  \exc \overline{T_S}
	&
	\overline{\exc T_S} &\defeq\  \que \overline{T_S}
\end{align*}
The translation is a homomorphism of involutions:
\begin{lemma}
	$\bra*{\overline{T_S}} = \negg{\bra*{T_S}}$.
\end{lemma}
\noindent Thus, connecting channels in CSGV will be translated to cuts in linear logic.

\begin{definition}
	The set of \emph{unlimited types} is defined inductively as follows.
	\begin{itemize}
		\item $\unit$ and $T \rightarrow U$ are unlimited.
		\item $T + U$ and $T \otimes U$ are unlimited whenever $T$ and $U$ are.
	\end{itemize}
	All other types are \emph{linear}.
\end{definition}
Values of unlimited types can be discarded and duplicated, because they are translated to CSLL types
that admit weakening and contraction. Categorical considerations \cite[\S
6.5]{mellies2009categorical} lead us to consider $T \otimes U$ unlimited whenever $T$ and $U$ are,
which is finer-grained than \emph{loc. cit.}

\paragraph{Terms}
CSGV is a linear $\lambda$-calculus, extended with constructs for sending and
receiving messages.
	\begin{align*}
	\Prog{L,M,N} \Coloneqq \quad
		& \Name{x} \mid
			\Unit \mid
			\Lam{x}{N} \mid
			\Jux{M}{N} \mid
			\Pair{M}{N} \mid
			\LetP{x}{y}{M} \; \Prog{N}  \\
		\mid \ & \Left{M} \mid \Right{M} \mid
			\Match{L}{x}{M}{N}
			\tag{functional fragment}\\
		\mid \ & \Send{M}{N} \mid \Recv{M} \tag{send and receive}\\
		\mid \ & \SelectL{M} \mid \SelectR{M} \mid \SCase{L}{x}{M}{N} \tag{select options}\\
		\mid \ & \Term{M} \tag{terminate $M$}\\
		\mid \ & \Conn{x}{M}{y}{N} \tag{connect $x$ of $M$ to $y$ of $N$}\\
		\mid \ & \ReqW{x} \tag{end client pool} \\
					\mid \ & \ReqA{x}{x'}{M} \tag{extract client interface} \\
					\mid \ & \Serv{y}{L}{z}{M}{f}{N} \tag{server construction}
\end{align*}

\paragraph{Typing rules}

The environments of CSGV are given by
$\Gamma, \dots \Coloneqq\ \bullet \mid\Gamma, \Name{x} : T
$. The translation of types is extended to environments pointwise.

\begin{figure}
	\small
	\begin{flushleft}
		$
		\Trans{
		\inferH{Recv}{
			\IsProg{\Gamma}{M}{\recv T.T_S}
		}{
			\IsProg{\Gamma}{\Recv{M}}{T \otimes T_S}
		}
		}[z]
		\defeq\ 
		\IsProc{\bra[z]{M}}{\negg{\bra*{\Gamma}}, \Name{z}: \bra*{T} \otimes \bra*{T_S}}
		$
	\end{flushleft}
	\begin{flushleft}
		$
		\Trans{
		\inferH{Send}{
			\IsProg{\Gamma}{M}{T} \\
			\IsProg{\Delta}{N}{\send T.T_S}
		}{
			\IsProg{\Gamma,\Delta}{\Send{M}{N}}{T_S}
		}
		}[z]
		\defeq\ 
		$
	\end{flushleft}
	\begin{mathpar}
		\inferrule{
			\inferrule*[right=$\otimes$]{
				\IsProc{\bra[y]{M}}{\negg{\bra*{\Gamma}}, \Name{y}:\bra*{T}}
				\\
				\IsProc{\Link{x'}{z}}{\Name{x'} : \negg{\bra*{T_S}}, \Name{z}:\bra*{T_S}}
			}{
				\IsProc{\Out{x'}{y}{\Par*{
					\bra[y]{M}
				}{
					\Link{x'}{z}
				}}}{
				\negg{\bra*{\Gamma}}, \Name{x'}:\bra*{T}\otimes \negg{\bra*{T_S}}, \Name{z}:\bra*{T_S}}
			}
			\\
			\IsProc{\bra[x]{N}}{\negg{\bra*{\Delta}}, \Name{x}:\negg{\bra*{T}} \parr \bra*{T_S}}
		}{
			\IsProc{\New{x}{x'}{\Par*{
				\Out{x'}{y}{\Par*{
					\bra[y]{M}
				}{
					\Link{x'}{z}
				}}
			}{
				\bra[x]{N}
			}}}{
				\negg{\bra*{\Gamma}},\negg{\bra*{\Delta}}, \Name{z}:\bra*{T_S}}
		}
	\end{mathpar}
	\begin{flushleft}
		$
		\Trans{
		\inferH{Conn}{
			\IsProg{\Gamma, \Name{x}:T_S}{M}{\endm}\\
			\IsProg{\Delta, \Name{y}:\overline{T_S}}{N}{T}
		}{
			\IsProg{\Gamma,\Delta}{\Conn{x}{M}{y}{N}}{T}
		}
		}[z]
		\defeq\ 
		$
	\end{flushleft}
	\begin{mathpar}
		\inferrule*[right=Cut]{
			\inferrule*[right=Cut]{
				\IsProc{\bra[y]{M}}{\negg{\bra*{\Gamma}}, \Name{x}:\negg{\bra*{T_S}}, \Name{y}:\punit}
				\\
				\IsProc{\Out{z}{}{\Stop}}{\Name{z}:\ounit}
			}{
			\IsProc{\New{y}{z}{\Par*{
				\bra[y]{M}
			}{
				\Out{z}{}{\Stop}
			}}}{\negg{\bra*{\Gamma}}, \Name{x}:\negg{\bra*{T_S}}}
			}
			\\
			\IsProc{\bra[z]{N}}{\negg{\bra*{\Delta}}, \Name{y}:\bra*{T_S}, \Name{z}:\bra*{T}}
		}{
			\IsProc{\New{x}{y}{\Par*{
				\New{y}{z}{\Par*{
				\bra[y]{M}
			}{
				\Out{z}{}{\Stop}
			}}
			}{
				\bra[z]{N}
			}}}{\negg{\bra*{\Gamma}},\negg{\bra*{\Delta}}, \Name{z}:\bra*{T}}
		}
	\end{mathpar}
	\caption{CSGV Typing Rules and Translation to CSLL: linear session part}
	\label{figure:csgv-linear}
\end{figure}

\begin{figure}
	\small
	\begin{flushleft}
		$
		\Trans{
			\inferH{ReqW}{
				\strut
			}{
				\IsProg{\Name{x} : \que T_S}{\ReqW{x}}{\endm}
			}
		}[z] \defeq\
		\inferrule{
			\inferrule*{
				\IsProc{\Stop}{\emptyset}
			}{
				\IsProc{\QueW{x}{\Stop}}{\Name{x} : \que \negg{\bra*{T_S}}}
			}
		}{
			\IsProc{\In{z}{}{\QueW{x}{\Stop}}}
						 {\Name{x} : \que \negg{\bra*{T_S}}, \Name{z} : \bot}
		}
		$
	\end{flushleft}
	
	\begin{flushleft}
		$
		\Trans{
			\inferH{ReqA}{
				\IsProg{\Gamma , \Name{x'} : T_S}{M}{\endm} \\
			}{
				\IsProg{\Gamma, \Name{x} : \que T_S}{\ReqA{x}{x'}{M}}{\que T_S}
			}
		}[z] \defeq
		$
	\end{flushleft}

	\begin{mathpar}
			\inferrule*[vcenter, right=HMix2+QueA]{
				\inferrule*{
					\IsProc{\bra[u]{M}}
								 {\negg{\bra*{\Gamma}}, \Name{x'} : \negg{\bra*{T_S}}, \Name{u} : \punit}\\
					\IsProc{\Out{v}{}{\Stop}}{\Name{v} : \ounit}
				}{
					\IsProc{
						\New{u}{v}{
							\Par*{\bra[z]{M}}{\Out{v}{}{\Stop}}
						}
					}{\negg{\bra*{\Gamma}}, \Name{x'} : \negg{\bra*{T_S}}}
				} \\
				\IsProc{\Link{x}{z}}{\Name{x} : \que \negg{\bra*{T_S}}, \Name{z} : \exc \bra*{T_S}}
			}{
				\IsProc{
					\QueA{x}{x'}{
						\Par*{
							\New{u}{v}{\Par*{\Out{z}{}{\Stop}}{\bra[u]{M}}}
						}{
							\Link{x}{z}
						}
					}
				}{
					\negg{\bra*{\Gamma}}, \Name{x} : \que\negg{\bra*{T_S}}, \Name{z} : \exc \bra*{T_S}
				}
			}
	\end{mathpar}

	\begin{flushleft}
		$
		\Trans{
			\inferH{Serv}{
				\IsProg{\Delta}{L}{T} \\
				\IsProg{\Name{z}:T, \Name{y}:T_S}{M}{T} \\
				\IsProg{\Sigma, \Name{f}:T}{N}{U} 
			}{
				\IsProg{\Delta,\Sigma, \Name{y} : \exc T_S}{\Serv{y}{L}{z}{M}{f}{N}}{U}
			}
		}[u] \defeq
		$
	\end{flushleft}

	\begin{mathpar}
		\inferrule*[right=Claro]{
			\inferrule*{
				\IsProc{\bra[i]{L}}{\negg{\bra*{\Delta}}, \Name{i}:\bra*{T}} \\
				\IsProc{\bra[u]{N}}{\negg{\bra*{\Sigma}}, \Name{f}: \negg{\bra*{T}}, \Name{u}:\bra*{U}}
			}{
				\IsProc{\Par{\bra[i]{L}}{\bra[u]{N}}}{
					\negg{\bra*{\Delta}}, \Name{i} : \bra*{T}
						\mid
					\negg{\bra*{\Sigma}}, \Name{f} : \negg{\bra*{T}}, \Name{u}:\bra*{U}   
				}
			} \\
			\IsProc{\bra[z']{M}}{\Name{z}:\negg {\bra*{T}}, \Name{y} : \negg{\bra*{T_S}}, \Name{z'}: \bra*{T}}
		}{
			\IsProc{\ExcP*{y}{\Par{\bra[i]{L}}{\bra[u]{N}}}{i}{f}{z}{z'}{y}{\bra[z']{M}}}{
			\negg{\bra*{\Delta}}, \Name{y} : \exc \negg{\bra*{T_S}}, \negg{\bra*{\Sigma}}, \Name{u} : \bra*{U}}
		}
	\end{mathpar}

	\caption{CSGV Typing Rules and Translation to CSLL: shared session part}
	\label{figure:csgv-shared}
\end{figure}

Selected typing rules of CSGV are given in \cref{figure:csgv-linear,figure:csgv-shared}. Most rules
follow \citet[\S 4.1]{wadler2014} to the letter, and are therefore omitted. In the interest of
economy we also give the translation to CSLL at the same time. The translation is defined by
induction on the typing derivations of CSGV. As the purpose of a CSGV program is the computation of
a value of a distinguished type, the translation must privilege a single name over which this value
will be returned. Thus, given a choice of name $\Name{z}$ and a typing derivation
$\IsProg{\Gamma}{M}{T}$, we write $\Trans{\IsProg{\Gamma}{M}{T}}[z]$ for its translation into CSLL.
Somewhat abusively we will sometimes also write $\IsProc{\bra[z]{M}}{\negg{\bra*{\Gamma}}, \Name{z}
: \bra*{T}}$ for the translated term. This slight abuse of notation also reveals the intended
typing.

The novelty here is in the CSGV rules for client-server interaction, and their translation into
CSLL. A name of shared client type $\que T_S$ can be seen as a form of `capability' for talking to
the server. \ruleref{ReqW} discards this capability, signalling the end of the client pool.
\ruleref{ReqA} uses it to spawn a fresh channel $\Name{x'}$ on which a client $\Prog{M}$ will talk
to a server, and returns the capability back to the caller. The client $\Prog{M}$ itself has type
$\endm$: it does not return valuable information, but uses values and channels found in $\Gamma$.

Dually, $\exc T_S$ is the type of a server channel. \ruleref{Serv} constructs a server from three
components. $\Prog{L}$ computes the initial state of the server. Given the current state in
$\Name{z}$, and a client channel $\Name{y}$, $\Name{M}$ serves the client listening on $\Name{y}$,
and then returns the next state of the server. $\Prog{N}$ finalizes the server. Note that the
so-called server `state' here could well be a channel itself, enabling bidirectional interleaving
communication---a design we will explore in \cref{section:voting-game}.

The \ruleref{Serv} typing rule is quite restrictive, in that it does not allow anything from the
environments $\Delta$ and $\Sigma$ to be used in the term $M$ which computes the next state of the
server. Fortunately, the following derivable rule allows us to weave some non-linear values of types
$\vec{V}$ in the server.
\begin{mathpar}
	\inferrule{
		\IsProg{\Delta}{L}{T}
		\and
		\IsProg{\Name{\vec{v}}:\vec{V}, \Name{z}:T, \Name{y}:T_S}{M}{T}
		\and
		\IsProg{\Sigma, \Name{f} : T}{N}{U}
		\and
		\text{$\vec{V}$ unlimited}
	}{
		\IsProg{\Name{\vec{v}} : \vec{V}, \Delta, \Sigma, \Name{y} : \exc T_S}
					 {%
					 	\underbrace{
							\Serv{y}{
								\Pair{\Name{\vec{v}}}{L}
							}{
								z'
							}{
								\LetP{\vec{v}}{z}{\Name{z'}} \; \Pair{\Name{\vec{v}}}{M}
							}{
								f'
							}{
								\LetP{\Name{\vec{v}}}{f}{\Name{f'}}\ N
							}
						 }_{%
						 	\textstyle
							\Serv*{y}{L}{z}{M}{f}{N}
						 }
					}{U}
	}
\end{mathpar}
We will make crucial use of this derivable rule in a couple of our examples. We also also adopt the
common shorthands $\Let{x}{M}\; \Prog{N} \defeq (\Lam{x}{N})\Prog{M}$ and $\Let{\_}{M}\; \Prog{N}
\defeq (\Lam{z}{N})\Prog{M}$ for fresh $\Name{z} : \Unit$.

\subsection{Functional Data Structure Server}
	\label{section:csgv-data-server}

Our primitives can be used to protect a shared functional data structure. Without loss of
generality, we consider a server whose state is a purely functional queue $T$ with operations
\begin{mathpar}
	\Namesf{enq}: T \otimes A \rightarrow T
	\and
	\Namesf{deq} : T \rightarrow T \otimes (\textsf{Unit} + A)
	\and
	\Namesf{empty} : T
\end{mathpar}
In particular, $\Namesf{deq}$ could return $\textsf{Unit}$ if the queue is empty. The server will
talk to a client via a channel of type $T_S \defeq (\recv A.\endp) \with (\send (\textsf{Unit} +
A).\endp)$. One client receives an $A$ along $\Name{r_0}$, and enqueues it. The other one dequeues
an element, and sends it along $\Name{r_1}$.

\begin{minipage}{0.4\textwidth}
	\begin{align*}
		\Prog{L} \defeq\  &\Namesf{empty}  \\
		\Prog{M_{enq}} \defeq\  &\LetP{v}{y''}{\Recv{\Name{y'}}}  \\
					&\Let{\_}{\Term{\Name{y''}}} \; \Namesf{enq} (\Name{z}, \Name{v}) \\
		\Prog{M_{deq}} \defeq\ & \LetP{v}{z'}{\Namesf{deq}\; \Name{z}}  \\
				& \Let{\_}{\Term*{\Send{\Name{v}}{\Name{y'}}}}\; \Name{z'}\\
					\Prog{M} \defeq\  & \SCase{\Name{y}}{y'}{M_{enq}}{M_{deq}}
	\end{align*}
\end{minipage}
\begin{minipage}{0.5\textwidth}
\begin{align*}
	\Prog{C_0} \defeq\
		&\LetP{v}{r_0'}{\Recv{\Name{r_0}}} \\
			&\Let{\_}{\Term{\Name{r_0'}}} \\
			&\Let{x_0'}{\Send{\Name{v}}{\SelectL*{\Name{x_0}}}} \; \Name{x_0'} \\
	\Prog{C_1} \defeq\
		&\LetP{v}{x_1'}{\Recv{\SelectR*{\Name{x_0}}}} \\
			&\Let{\_}{\Term*{\Send{\Name{v}}{r1}}} \; \Name{x_1'} \\
	\Progsf{clients} \defeq\  & \Let{x}{\ReqA{x}{x_0}{C_0}} \\
		& \Let{x}{\ReqA{x}{x_1}{C_1}} \; \ReqW{x}
		\end{align*}
\end{minipage}

\noindent We then define $\Progsf{server} \defeq \Serv{y}{L}{z}{M}{f}{\Name{f}}$, and see that
\[
	\IsProg{\Name{r_0} : \recv A.\endp, \Name{r_1} : \send (\textsf{Unit} + A).\endp}
				 {\Conn{x}{\textsf{clients}}{y}{\textsf{server}}}
				 {T}
\]

\subsection{Nondeterminism}
	\label{section:csgv-nondeterminism}

Unsurprisingly, the races in our system suffice to implement nondeterministic choice. We define
$\boo \defeq \textsf{Unit} + \textsf{Unit}$. We implement $\True$ and $\False$ by the obvious
injections, and the conditional by
\begin{mathpar}
	\inferrule*[right=$+E$]{
		\IsProg{\Gamma}{B}{\boo} \\
			\IsProg{\Delta}{M}{V}
		\\
			\IsProg{\Delta}{N}{V}
	}{
		\IsProg{\Gamma, \Delta}{\IfThen{B} \; M \; \Else \; N \defeq \Match{L}{x}{M}{N}}{V}
	}
\end{mathpar}
The clients $C_0, C_1$ respectively send $\False$ and $\True$ over a channel. We also define a
server with a pair of Booleans as internal state. The first component records whether the server has
ever received a value. When a value is received it is stored in the second component, and any
further values received are discarded.

\begin{minipage}{0.40\textwidth}
\begin{align*}
	\Prog{C_0} \defeq\
		& \Send{\False}{\Name{x_0}} \\
	\Prog{C_1} \defeq\
		& \Send{\True}{\Name{x_1}} \\
	\Progsf{clients} \defeq\
		& \Let{x}{\ReqA{x}{x_0}{C_0}} \\
		& \Let{x}{\ReqA{x}{x_1}{C_1}}\\
		& \ReqW{x}
\end{align*}
\end{minipage}
\begin{minipage}{0.45\textwidth}
\begin{align*}
	\Prog{M} \defeq\
		& \LetP{z_0}{z_1}{\Name{z}} \\
		& \LetP{v}{y'}{\Recv{\Name{y}}} \\
		& \Let{\_}{\Term{\Name{y'}}} \\
		& \IfThen{\Name{z_0}}\;  \Name{z}\; \Else\;  \Pair{\True}{\Name{v}}\\
	\Prog{N} \defeq\
		& \LetP{f_0}{f_1}{\Name{f}} \; \Name{f_1}
	 \end{align*}
\end{minipage}

\noindent We define a $\Progsf{server} \defeq\ \Serv{y}{\Pair{\False}{\False}}{z}{M}{f}{N}$
beginning from $\Pair{\False}{\False}$. We then have that
\[
	\IsProg{}{\Progsf{flip} \defeq \Conn{x}{\textsf{clients}}{y}{\textsf{server}}}{\boo}
\]
This program is translated to $\IsProc{\bra[y]{\Progsf{flip}}}{\Name{y} : \boom}$, with reactions
$\red*{\bra[y]{\Progsf{flip}}}{\FalseP{y}}$ and $\red*{\bra[y]{\Progsf{flip}}}{\TrueP{y}}$. We can
use this to implement a nondeterministic choice operator:
\begin{mathpar}
	\inferrule{
		\IsProc{P}{\Gamma} \\
		\IsProc{Q}{\Gamma}
	}{
		\IsProc{\Choose{P}{Q} \defeq \New{x}{y}{\Par*{\bra[y]{\Progsf{flip}}}{\IfO{x}{P}{Q}}}}
					 {\Gamma}
	}
\end{mathpar}
such that $\red*{\Choose{P}{Q}}{P}$ and  $\red*{\Choose{P}{Q}}{Q}$.

\subsection{Fork–join Parallelism}
	\label{section:fork-join}

\emph{Fork-join parallelism} \citep{conway1963multiprocessor} is a common model of parallelism in
which child processes are \emph{forked} to perform computation simultanously. Once they have
finished, they are \emph{joined} by the parent process, which collects their work and produces the
final result. We assume a `heavyweight' function $\Prog{h} : A \to B$ that will run on forked
processes, and a relatively less expensive function $\Prog{g} : B \to B \to B$ that will combine
their answers. We also assume an initial value $\Prog{g_0} : B$, and a list of `tasks' $\Prog{xs} :
[A]$ to process. $[A]$ is the type of lists of $A$, and is supported by the operations:
\begin{align*}
	\Progsf{nil} &: [A]
	&
	\Progsf{cons} &: A \to [A] \to [A]
	&
	\Prog{\textsf{fold}_{C}} &: C \to (C \to A \to C) \to [A] \to C
\end{align*}
Let
\begin{align*}
	\Progsf{clients} & \defeq \
		\Let{y}{
			\textsf{fold}_{\que T_S} 
				\; \Name{c}
				\; \Progparens{
							\Lam{x}{\Lam{v}{
								\ReqA{x}{x'}{(\Let{v'}{\Prog{h} \; \Name{v}} \; \Send{\Name{v'}}{\Name{x'}})}
							}}}
				\; \Prog{xs}
		} 
				\; \ReqW{y} \\
	\Prog{M} &\defeq\
			\LetP{v}{y'}{\Recv{\Name{y}}} \; 
				\Let{\_}{\Term{\Name{y'}}} \; 
					\Prog{g}\; \Name{z}\; \Name{v}
\end{align*}
The client protocol is $T_S \defeq \send B.\endm$. To form the client pool, we begin with a shared
client channel $\Name{c} : \que T_S$. We fold over the list $\Prog{xs} : [A]$, adding a forked
process for each `task' $\Name{v} : A$ to the client pool. Each one of these forked processes will
compute $\Prog{h} \; \Name{v} : B$, and send it over its fresh channel $\Name{x'} : T_S$. We have
$\IsProg{\Name{c} : \que T_S}{\textsf{clients}}{\endm}$.

We let $\Progsf{server} \defeq\ \Serv*{y}{\Prog{g_0}}{z}{M}{f}{\Name{f}}$. The server begins with
internal state $\Prog{g_0} : B$. It nondeterministically receives the result of a computation of
$\Name{h}$ from each client, and `merges' it into its state using $\Prog{g}$. In the end, it returns
the result. We have $\IsProg{\Name{z} : B, \Name{y} : \overline{T_S}}{M}{B}$, and thus
$\IsProg{\Name{y} : \exc \overline{T_S}}{\textsf{server}}{B}$. We use
$\Derivsf{serve'}$ to pass unlimited parameters to the server internals.

Putting this system together, we get 
\[
	\IsProg{}{\ForkJoin{\Prog{h}}{\Prog{g_0}}{\Prog{g}}{\Prog{xs}} \defeq \Conn{x}{\textsf{clients}}{y}{\textsf{server}}}{B}
\]

The fork-join paradigm is often used in industrial parallelization frameworks
\citep{dagum1998openmp,reinders2007intel,blumofe1995cilk,leijen2009design}. The background languages
and type systems usually do not use any logical devices for concurrency. In particular, concurrent
behaviour is not controlled by the type system, as it is here. Note that fork-join requires each spawned process to be independent of each other and only communicate with the parent process, which is precisely caputured by the linearity restriction of our system. 

Another parallel computation model is that of \textit{async-finish}. It is more expressive than
fork-join, as it allows spawned processes to spawn further processes. The whole tree is then joined
at the root process, with no regard to the spawning thread of each child. Our system(s) does not
support that: in the \ruleref{ReqA} rule, the spawned process $\Prog{M}$ is only given a channel
$\Name{x'} : T_S$, which cannot be used to spawn further processes in the same pool. However, it is
well-known is that nested parallelism is still possible, but each child has to spawn its own
instance of a fork-join computation, which does not interfere with the root process.

An even more expressive model is that of \textit{futures} \citep{halstead1984implementation}. A
future is a first-class value that represents a computation running in parallel to the current
process. At any point it can be \textit{forced} to obtain its result; if it has not finished an
error may be returned, or the process forcing it may block. While fork-join or async-finish spawned
processes are independent of each other, futures may be passed around freely (in any
reasonably expressive language) and introduce rich interactions. This
seems to be in violation of the linearity restriction of our system(s), and thus cannot be
expressed. Nevertheless, the \ruleref{Conn} rule can be seen as a very restricted form of future,
where the spawned process can only communicate with the parent process. More discussions about the
difference between these models is given by \citet{acar2016tapp}.

\subsection{Keynes' beauty contest}
	\label{section:voting-game}

Until this point we have seen only relatively simple examples of server-client interaction. In all
cases, the `internal server protocol' we have used has consisted of an unlimited type, the values of
which we can replicate and discard. This leads to the false impression that clients access the
server one-by-one in a sequential manner, so that clients that connect later are unable to influence
the information observed by the earlier ones. In this section we present an example that shows this
to be untrue. In particular, if the `internal server protocol' consists of a session type itself,
then we witness bidirectional, interleaving behaviour. This distinguishes our systems from those
based on \emph{manifest sharing} \cite{balzer2017}.

We present a server implementing the umpire in a \emph{Keynesian beauty contest} \cite[\S
12]{keynes_1936}. Keynes' beauty contest works as follows. A newspaper runs a beauty contest in
which readers have to pick the prettiest faces from a set of photographs. The competitors are not
those pictured, but the readers themselves: if they pick the faces which are judged to be the
prettiest by the majority, they will win a prize. Thus, the readers are incentivized to estimate the
aesthetics of the majority. 

We will implement a restricted version of this scenario, where a pool of clients votes for a Boolean
value. The server then counts the votes, and awards a payoff of $0$ or $1$ (represented by $\False$
and $\True$ respectively) to each client, indicating whether they voted for the winner. This is
obviously impossible if the server handles requests sequentially. In fact, the server will be
implemented by spawning a network of interconnected processes, each of which will handle one vote.

We first define the following derived rule. Informally, this rule expresses that a process that uses
a channel of type $T_S$ is also exposing a channel of dual type $\overline{T_S}$.
\begin{mathpar}
	\inferrule*{
		\IsProg{\Gamma, \Name{x}:T_S}{M}{\endm} \\
		\IsProg{\Name{y}: \overline{T_S}}{\Name{y}}{\overline{T_S}}
	}{
		\IsProg{\Gamma}{\Inv{x}{M} \defeq \Conn{x}{M}{y}{\Name{y}}}{\overline{T_S}}
	}
\end{mathpar}

The client
session type is $C_S \defeq \send \boo. \recv \boo. \endm$, and the internal server protocol is $T_S
\defeq \recv (\nat \otimes \nat) . \send \boo. \endp$, where $\nat$ is the type of natural numbers.
We assume a bunch of standard functions:
\begin{mathpar}
	\Progsf{zero} : \nat \and
	\Progsf{succ} : \nat \rightarrow \nat \and
	\Prog{\leq} : \nat \rightarrow \nat \rightarrow \boo \and
	\Progsf{eq} : \boo \rightarrow \boo \rightarrow \boo
\end{mathpar}
where $\Progsf{eq}$ checks Boolean values for equality. We let
\begin{align*}
	L \defeq\
		& \Let{w'}{\Send{\Pair{\Namesf{zero}}{\Namesf{zero}}}{\Name{w}}}	\tag{send initial state}\\
		& \LetP{\_}{w''}{\Recv{\Name{w'}}}  \; \Name{w''} 							 	\tag{receive final value}\\
	N \defeq\
	  & \LetP{s}{f'}{\Recv{\Name{f}}} 				\tag{receive final count} \\
		& \LetP{n_0}{n_1}{\Name{s}} 										\tag{unpack state} \\
		& \Let{f''}{\Send{(\Name{n_0} \mathbin{\Prog{\leq}} \Name{n_1})}{\Name{f'}}}
																										\tag{compute winner and notify the last worker process} \\
		& \Term{\Name{f''}}															\tag{close channel}	\\	
	M \defeq\
		& \LetP{s}{z'}{\Recv{\Name{z}}}									\tag{get state} \\
		& \LetP{n_t}{n_f}{\Name{s}} 										\tag{unpack state} \\
		& \LetP{b}{y'}{\Recv{\Name{y}}} 								\tag{receive a vote} \\
		& \Let{s'}{
				\IfThen{\Name{b}} \; 
					\Pair{\Progsf{succ} \; \Name{n_t}}{\Name{n_f}} \;
				\Else \;
					\Pair{\Name{n_t}}{\Progsf{succ} \; \Name{n_f}}
			}																							\tag{increment the right counter} \\
		& \Let{w'}{\Send{\Name{s'}}{\Name{w}}} 					\tag{pass new state to next worker process} \\
		& \LetP{b'}{w''}{\Recv{\Name{w'}}} 							\tag{receive winner from next worker process} \\
		& \Let{\_}{\Term*{\Send{(\Progsf{eq}\; \Name{b}\; \Name{b'})}{\Name{y'}}}} 
																										\tag{tell competitor if they won, close channel} \\
		& \Let{\_}{\Term*{\Send{\Name{b'}}{\Name{z'}}}} \tag{forward winner on, close channel} \\
		& \Name{w''}
\end{align*}
We define $\Progsf{server} \defeq \Serv{y}{\Inv{w}{L}}{z}{\Inv{w}{M}}{f}{N}$. The components are typed as
\begin{equation*}
	\begin{aligned}[c]
		\IsProg*{\Name{w} : \overline{T_S}}{L}{\endm}\\
		\IsProg*{\Name{w} : \overline{T_S}, \Name{z} : T_S, \Name{y} : \overline{C_S}}{M}{\endm}
	\end{aligned}
		\qquad
	\begin{aligned}[c]
		\IsProg*{\Name{f} : T_S}{N}{\unit}\\
		\IsProg*{\Name{y} : \exc T_S}{\textsf{server}}{\unit}
	\end{aligned}
\end{equation*}
The details of this protocol are subtle. The construct $\Inv{x}{-}$ allows us to use programs which
only have side-effects as internal server state, by inverting the polarity of one of the channels.
The server is initialized by $\Prog{L}$, which sets the state to be $(0, 0)$. It then listens on the
same channel to receive the winner, which it promptly discards. The server finalization $\Prog{N}$
receives the final tally of the votes, computes the winner, sends back the result, and closes the
channel.

The component $\Prog{M}$ is used to communicate with each competitor. It receives the state of the
server, the competitor's vote, and increments the appropriate tally. It then passes on this new state to
the next worker process $\Prog{M}$, which will communicate with the next competitor. This sets up an entire network of
worker processes $\Prog{M}$, one to serve each competitor. When the competitors have all cast their votes, $\Prog{N}$
computes the winner, and sends it back to the last worker process. This process then tells the
competitor whether they won, closes the channel to the competitor, and passes on the result to the worker process
serving the previous competitor, and so on. At the very end, the winner is passed to the initialization
process $\Prog{L}$.

We can then define a number of competitors $\IsProg{\Name{x_i} : \send \boo. \recv \boo.
\endm}{C_i}{\endm}$ who will cast their votes by sending a Boolean value and receive a payoff along
$\Name{x_i}$. These can be combined into a client pool, much in the same way as in previous
examples.

\begin{wrapfigure}{R}{0.5\textwidth}
	\begin{tikzpicture}[
		process/.style={rectangle, draw=black!60, fill=black!5, very thick, minimum size=12.5mm},
		null/.style={rectangle,draw=black!0, minimum size = 0mm},
		node distance=0.5cm
	]
	\node[process] (L) {$\Prog{L}$};
	\node[process] (N) [below=of L] {$\Prog{N}$};
	\node[process] (M0) [right=1cm of L] {$\Prog{M}$};
	\node[process] (M1) [right=1cm of N] {$\Prog{M}$};
	\node[process] (C0) [right=of M0] {$\Prog{C_0}$};
	\node[process] (C1) [right=of M1] {$\Prog{C_1}$};
	\node[null] (C0-env) [right=of C0] {};
	\node[null] (C1-env) [right=of C1] {};
	\node[null] (topleft) [above=of L.west] {};
	\node[null] (botright) [below=of M1.east] {};

	\draw[-latex] (L.east) node [left] {$\Name{i}$} to [out=45, in=135]  node [above] {$\nat \otimes \nat$}
	(M0.west) node [right] {$\Name{z}$};
	\draw[latex-] (L.east) to  [out=-45,in=-135] node [below] {$\boo$} (M0.west);

	\draw[-latex] (M0.south) node [above] {$\Name{z'}$} to [out=-45,in=45] node [right] {$\nat \otimes \nat$} (M1.north) node [below] {$\Name{z}$};
	\draw[latex-] (M0.south) node [above] {$\Name{z'}$} to [out=-135,in=135] node [left] {$\boo$} (M1.north) node [below] {$\Name{z}$};

	\draw[-latex] (M1.west) node [right] {$\Name{z'}$} to [out=135,in=45] node [above] {$\nat \otimes \nat$} (N.east)  node [left] {$\Name{f}$};
	\draw[latex-] (M1.west) node [right] {$\Name{z'}$} to [out=-135,in=-45] node [below] {$\boo$} (N.east)  node [left] {$\Name{f}$};

	\draw[latex-] (M0.east) node [left] {$\Name{y}$} to [out=45,in=135] node [above] {$\boo$} (C0.west)  node [right] {$\Name{x_0}$};
	\draw[-latex] (M0.east) node [left] {$\Name{y}$} to [out=-45,in=-135] node [below] {$\boo$} (C0.west)  node [right] {$\Name{x_0}$};

	\draw[latex-] (M1.east) node [left] {$\Name{y}$} to [out=45,in=135] node [above] {$\boo$} (C1.west)  node [right] {$\Name{x_1}$};
	\draw[-latex] (M1.east) node [left] {$\Name{y}$} to [out=-45,in=-135] node [below] {$\boo$} (C1.west)  node [right] {$\Name{x_1}$};

	\begin{pgfonlayer}{bg}    %
		\fill[blue!25] (topleft) rectangle (botright) node [below] {server};
	\end{pgfonlayer}
	\end{tikzpicture}
	\caption{Layout of the voting game after coexponentials reactions but before other reactions. Boxes represents processes whose names are at the center of the boxes. Arrows represents directed mesages between processes with types of the data annotated. Labels on edges of boxes are the names of the channels to the processes.}
	\label{figure:voting-layout}
\end{wrapfigure}
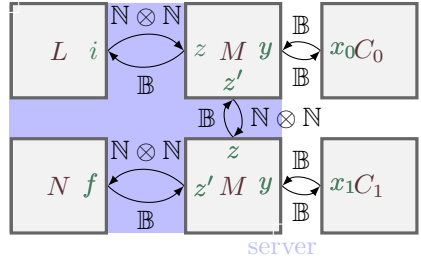

If we have two such competitors $C_0$ and $C_1$ merged in a pool, and we connect them to $\Progsf{server}$,
we will obtain a process topology of the form illustrated in the schematic diagram of
\cref{figure:voting-layout}. Compared with \cref{figure:cas-layout} this diagram is intuitive but
loose on accuracy. Details such as $\endp$ and $\endm$ are left out. We have also spelled out the
protocols internals. For example, the server internal protocol $T_S$ is indicated by a forward arrow
$\nat \otimes \nat$ and a backward arrow $\boo$.

\section{Related work}
	\label{section:related-work}

\paragraph{Hypersequents and Session Types.}

Hypersequents were introduced to process calculi and Classical Linear Logic by \citet{montesi2018}. Another version of that system was studied in detail by
\citet{kokke2019better}. A reaction semantics similar to the one used here was given in a later paper
\citep{kokke2019taking}.

\paragraph{Clients, Servers, and Races in Linear Logic.}

Typing client-server interaction has been a thorn in the side of session types and Linear Logic. All
previous attempts rely on some version of the \ruleref{Mix} rule. Both \citet[\S
3.4]{wadler2014} and \citet[Ex. 2.4]{caires_2017} use \ruleref{Mix} to
combine clients into client pools. Kokke et al. implicitly use \ruleref{Mix} to type an otherwise
untypable client pool in \textsc{HCP} \citep[Ex. 3.7]{kokke2019better}.\footnote{This has been
confirmed to us by the authors.} Remarkably, none of these calculi demonstrate stateful server
behaviour, as we predicted using a semantic argument in \cref{section:problem}.

\citet{lindley2016} explore the additional power bestowed upon CP by \emph{conflating} dual
connectives. The conflation of $\whynot$ and $\ofc$ leads to the notion of \emph{access point}, a
dynamic match-making communication service on a single end point. In fact, the rules look eerily
close to the list-like formulation of our servers and generators. Access points prove too powerful:
they introduce stateful nondeterminism, racy communication, and general recursion. This impairs the
safety of CP by introducing deadlock and livelock. Our work shows that we can still safely obtain
the former two features without introducing the third.

Adding nondeterminism to CLL in a controlled fashion is complex. \citet{lindley2016} express a form
of \emph{nondeterministic local choice} in CP by conflating $\with$ and $\oplus$. The resultant form
of nondeterministic choice cannot induce the racy behaviour normally exhibited in the $\pi$-calculus
\citep[\S 2]{kokke2019races}. \citet{caires_2017} present a dual-context system based on
CLL+\ruleref{Mix} in which the same kind of nondeterministic local choice is expressed through a new
set of modalities, $\oplus$ and $\with$.\footnote{This is an intentional clash with external and
internal choice in Linear Logic.} These bear a similarity to the coexponential modalities presented
here, but they are used for nondeterminism instead. Their $\with$ modality has a monadic flavour,
and hence can be used to encapsulate nondeterminism `in the monad' in the usual manner in which we
isolate effects.

\citet{kokke2019races} drew inspiration from Bounded Linear Logic
\citep{girard_1992} to formulate a system for nondeterministic client-server interaction. They use
types of the form $\whynot_n A$ (standing for $n$ copies of $A$ delimited by $\parr$) and $\ofc_n A$
(standing for $n$ copies of $A$ delimited by $\otimes$). $\ofc_n A$ represents a pool of $n$
disjoint clients with protocol $A$, and $\whynot_n A$ a server that can serve exactly $n$ clients
with protocol $A$. While this is consistent with disjoint-vs.-connected concurrency, their system is
limited to serving a specific number of clients in each session. Thus, it fails to satisfy criterion
(i) in \cref{section:problem}, and does not form a satisfactory model.

\paragraph{Fixed Points in Linear Logic.}

Inductive and coinductive types---presented proof-theoretically as least and greatest fixed
points---were introduced in the context of higher-order Classical Linear Logic by
\citet{baelde2012least}. Baelde formulates a weakly normalization cut elimination procedure, which
albeit does not satisfy the subformula property. \citet{ehrhard2021categorical} study categorical
models for a slight extension of the propositional fragment of Baelde's system, which allow them to
infer certain facts about the behaviour of (co)inductive linear data types.

The structure of Baelde's system has been used to extend CP with inductive and coinductive types by
\citet{lindley2016talking}. This is our starting point in \cref{section:deriving-coexp}, but with
significant alterations along the way. First, our system is based on HCP. Second, the server rule
has been reformulated to support hyperenvironments, and server finalization
(\cref{section:decision-lists}). Third, our client pools allow permutation in order to enable
nondeterminism (\cref{section:decision-permutation}). Finally, our reaction semantics are tailored
to the specific setting, and are consequently simpler. In a separate strand of work,
\citet{toninho2014corecursion} introduce coinduction in a system of session types based on
Intuitionistic Linear Logic (ILL); see \citet[\S\S 1, 7]{lindley2016talking} for a comparison. 

\paragraph{Manifest sharing.}

Closely related to our work is the notion of \emph{manifest sharing} \citep{balzer2017}. This work
starts from a very different premise: a channel is either \emph{linear} (as the usual channels in
session types), or \emph{shared} between processes. This leads to an ILL-based system,
$\textsf{SILL}_S$, with two \emph{modes} and two modalities shifting between them \citep{reed2009}.
The switch to a shared channel is punctuated by the modalities. Thus, sharing \emph{manifests} in
the types. In some ways, $\textsf{SILL}_S$ is a much stronger system, as it features
\emph{equi-synchronizing} recursive session types. The price to pay is the introduction of deadlock.
\citet{balzer2019} develop an additional layer of the type system that protects from it.

Our work attempts to solve the expressivity problem of LL-based session types beginning from
Curry-Howard: we seek the minimal extension to Linear Logic that will enable us to write server and
client processes. Unlike manifest sharing, we remain committed to CLL and its duality. The result is
that our system has simpler rules, avoiding the notions of linear and shared channels (as
(co)exponentials internalise them), as well as the lock-like primitives used to introduce modalities
in \citet{balzer2017}. Moreover, we have remained committed to the goal of retaining the good
properties ensured by cut elimination in CLL (e.g. deadlock freedom). A drawback of this approach is
that our system inherits the linearity constraint from linear logic, and is thus unable to express
circular structures (such as Dijkstra's dining philosophers).

Both systems provide atomicity, but in radically different ways. In $\textsf{SILL}_S$ users access
the service in a mutually exclusive manner. This is not compatible with the usual view of typical
client-server interaction, where mutliple clients need to access the server simultanously in order
to exchange information. A common workaround is to decompose an interleaved session into a
`stateless' protocol consisting of several mini-sessions. Every client is then required to send a
`cookie' to identify themselves across mini-sessions. In our system accesses to the shared service
are concurrent, but causally atomic (\cref{section:pools}). As a result, interleaved sessions can be
expressed natively (\cref{section:voting-game}).

\paragraph{Differential Linear Logic.}

The rules for $\que$ given in \S\ref{section:deriving-coexp} are almost the same as the
\emph{coweakening}, \emph{codereliction} and \emph{cocontraction} rules that are added to $\ofc$ in
Differential Linear Logic (DiLL) \citep{ehrhard2018}. DiLL is equipped with nondeterministic
reduction and formal sums, and is thus believed to have something to do with concurrency. \citet{ehrhard2010} have produced an embedding of the \emph{finitary} $\pi$-calculus into DiLL, though that encoding has been criticized \citep{mazza2018}. A type of
client-server interactions---namely the encoding of ML-style reference cells into session
types---has been encoded by \citet{castellan2020} in a system based on the rules of DiLL. This work
relies on both the costructural rules and \ruleref{Mix}, so it is not clear which device primarily
augments expressive power. Our work shows that something akin to the costructural rules of DiLL
arises from the wish to form client pools. The exact relationship between coexponentials and DiLL
remains to be determined.

\paragraph{Multiparty Session Types.}

There is a nontrivial connection between our work and \emph{Multiparty Session Types}
\citep{honda_2008,honda_2016,coppo_2016}, which comprise a $\pi$-calculus and a behavioural type
system specifying interaction between multiple agents. The kinds of protocols expressed by
multiparty session types are `fully' choreographed, and involve a \emph{fixed} number of
participants. As such, they cannot model interactions with an arbitrary number of clients; nor can
they introduce a controlled amount of nondeterminism. Some of these expressive limitations have been
remedied in systems of \emph{Dynamic Multirole Session Types} \citep{denielou_2011}, which come at
the price of introducing \emph{roles} that parties can dynamically join or leave, and a notion of
quantification over participants with a role. Our system captures certain use-cases of roles using
only tools from linear logic, with little additional complexity.

Closer to our work is the approach of \citet{carbone_2017} to multiparty session types
through \emph{coherence proofs}. In \emph{op. cit.} the authors develop \emph{Multiparty Classical
Processes}, a version of CP with role annotations and the \emph{MCut} rule. The latter is a version
of the \ruleref{MultiCut} rule annotated with a \emph{coherence} judgment derived from
\citet{honda_2008}, which generalises duality and ensures that roles match appropriately. MCP does
not allow dynamic sessions with arbitrary numbers of participants, and hence cannot model
client-server interactions. MCP was later refined into the system of Globally-governed Classical
Processes (GCP) by \citet{carbone_2016}. Unlike these calculi, our work does not
require any consideration of coherence or local vs. global types.

\section{Conclusions and Further Work}
\label{section:conclusion}

We presented the system of Client-Server Linear Logic, which features a novel form of modality, the
coexponentials. We then showed how \hcpcoexp~can be used to model client-server interactions without
falling down the slippery slope of introducing \ruleref{Mix}. We comment on some directions for
future work.

\paragraph{Termination.}
It would be interesting to establish a \emph{termination} result for \hcpcoexp. This would prove
that the resulting calculus does not generate \emph{livelock}. We expect this proof to be somewhat
involved, which is why most work on Linear Logic and session types either fails to produce a proof,
or defers to Girard's proof for CLL \cite{wadler2014,aschieri2019par}.

\paragraph{Syntax.}
The weak $\exc$ rule listed in \cref{section:deriving-coexp} is expressed by folding $\otimes$ over
the set of formulas. This obstructs a particular commuting conversion in cut elimination. Similarly, presentation of the \textit{strong} exponential and its computational interpretation is omitted due to its unsatisfactory rules. We believe these issues are due to the limitation of sequent calculus, and alternative techniques are necessary to solve them.

\bibliographystyle{ACM-Reference-Format}
\bibliography{ll}


\begin{thebibliography}{66}


\ifx \showCODEN    \undefined \def \showCODEN     #1{\unskip}     \fi
\ifx \showDOI      \undefined \def \showDOI       #1{#1}\fi
\ifx \showISBNx    \undefined \def \showISBNx     #1{\unskip}     \fi
\ifx \showISBNxiii \undefined \def \showISBNxiii  #1{\unskip}     \fi
\ifx \showISSN     \undefined \def \showISSN      #1{\unskip}     \fi
\ifx \showLCCN     \undefined \def \showLCCN      #1{\unskip}     \fi
\ifx \shownote     \undefined \def \shownote      #1{#1}          \fi
\ifx \showarticletitle \undefined \def \showarticletitle #1{#1}   \fi
\ifx \showURL      \undefined \def \showURL       {\relax}        \fi
\providecommand\bibfield[2]{#2}
\providecommand\bibinfo[2]{#2}
\providecommand\natexlab[1]{#1}
\providecommand\showeprint[2][]{arXiv:#2}

\bibitem[\protect\citeauthoryear{Abramsky}{Abramsky}{1993a}]%
        {abramsky_computational_1993}
\bibfield{author}{\bibinfo{person}{Samson Abramsky}.}
  \bibinfo{year}{1993}\natexlab{a}.
\newblock \showarticletitle{Computational interpretations of linear logic}.
\newblock \bibinfo{journal}{\emph{Theoretical Computer Science}}
  \bibinfo{volume}{111}, \bibinfo{number}{1-2} (\bibinfo{year}{1993}),
  \bibinfo{pages}{3--57}.
\newblock
\showISSN{03043975}
\urldef\tempurl%
\url{https://doi.org/10.1016/0304-3975(93)90181-R}
\showDOI{\tempurl}
\newblock
\shownote{ISBN: 0304-3975.}


\bibitem[\protect\citeauthoryear{Abramsky}{Abramsky}{1993b}]%
        {abramsky1993interaction}
\bibfield{author}{\bibinfo{person}{Samson Abramsky}.}
  \bibinfo{year}{1993}\natexlab{b}.
\newblock \showarticletitle{Interaction categories}.
\newblock In \bibinfo{booktitle}{\emph{Theory and Formal Methods 1993}}.
  \bibinfo{publisher}{Springer}, \bibinfo{pages}{57--69}.
\newblock


\bibitem[\protect\citeauthoryear{Abramsky}{Abramsky}{1994}]%
        {abramsky1994proofs}
\bibfield{author}{\bibinfo{person}{Samson Abramsky}.}
  \bibinfo{year}{1994}\natexlab{}.
\newblock \showarticletitle{Proofs as processes}.
\newblock \bibinfo{journal}{\emph{Theoretical Computer Science}}
  \bibinfo{volume}{135}, \bibinfo{number}{1} (\bibinfo{year}{1994}),
  \bibinfo{pages}{5--9}.
\newblock


\bibitem[\protect\citeauthoryear{Abramsky, Gay, and Nagarajan}{Abramsky
  et~al\mbox{.}}{1996}]%
        {abramsky1996interaction}
\bibfield{author}{\bibinfo{person}{Samson Abramsky}, \bibinfo{person}{Simon~J
  Gay}, {and} \bibinfo{person}{Rajagopal Nagarajan}.}
  \bibinfo{year}{1996}\natexlab{}.
\newblock \showarticletitle{Interaction categories and the foundations of typed
  concurrent programming}. In \bibinfo{booktitle}{\emph{Deductive {Program}
  {Design}}} \emph{(\bibinfo{series}{Nato {ASI} {Subseries} {F}})},
  \bibfield{editor}{\bibinfo{person}{Manfred Broy}} (Ed.).
  \bibinfo{publisher}{Springer-Verlag Berlin Heidelberg},
  \bibinfo{pages}{35--113}.
\newblock
\showISBNx{3-540-60947-4}
\urldef\tempurl%
\url{http://www.springer.com/us/book/9783540609476}
\showURL{%
\tempurl}


\bibitem[\protect\citeauthoryear{Abramsky and Jagadeesan}{Abramsky and
  Jagadeesan}{1994}]%
        {abramsky_games_1994}
\bibfield{author}{\bibinfo{person}{Samson Abramsky} {and}
  \bibinfo{person}{Radha Jagadeesan}.} \bibinfo{year}{1994}\natexlab{}.
\newblock \showarticletitle{Games and {Full} {Completeness} for
  {Multiplicative} {Linear} {Logic}}.
\newblock \bibinfo{journal}{\emph{The Journal of Symbolic Logic}}
  \bibinfo{volume}{59}, \bibinfo{number}{2} (\bibinfo{year}{1994}),
  \bibinfo{pages}{543}.
\newblock
\showISSN{00224812}
\urldef\tempurl%
\url{https://doi.org/10.2307/2275407}
\showDOI{\tempurl}
\showeprint{1311.6057}


\bibitem[\protect\citeauthoryear{Acar}{Acar}{2016}]%
        {acar2016tapp}
\bibfield{author}{\bibinfo{person}{Umut~A Acar}.}
  \bibinfo{year}{2016}\natexlab{}.
\newblock \bibinfo{booktitle}{\emph{Parallel Computing: Theory and Practice}}.
\newblock
\urldef\tempurl%
\url{http://www.cs.cmu.edu/afs/cs/academic/class/15210-f15/www/tapp.html}
\showURL{%
\tempurl}


\bibitem[\protect\citeauthoryear{Aschieri and Genco}{Aschieri and
  Genco}{2019}]%
        {aschieri2019par}
\bibfield{author}{\bibinfo{person}{Federico Aschieri} {and}
  \bibinfo{person}{Francesco~A. Genco}.} \bibinfo{year}{2019}\natexlab{}.
\newblock \showarticletitle{Par Means Parallel: Multiplicative Linear Logic
  Proofs as Concurrent Functional Programs}.
\newblock \bibinfo{journal}{\emph{Proc. ACM Program. Lang.}}
  \bibinfo{volume}{4}, \bibinfo{number}{POPL}, Article \bibinfo{articleno}{18}
  (\bibinfo{date}{Dec.} \bibinfo{year}{2019}), \bibinfo{numpages}{28}~pages.
\newblock
\urldef\tempurl%
\url{https://doi.org/10.1145/3371086}
\showDOI{\tempurl}


\bibitem[\protect\citeauthoryear{Atkey}{Atkey}{2017}]%
        {atkey2017observed}
\bibfield{author}{\bibinfo{person}{Robert Atkey}.}
  \bibinfo{year}{2017}\natexlab{}.
\newblock \showarticletitle{Observed communication semantics for classical
  processes}. In \bibinfo{booktitle}{\emph{European Symposium on Programming}}.
  Springer, \bibinfo{pages}{56--82}.
\newblock


\bibitem[\protect\citeauthoryear{Atkey, Lindley, and Morris}{Atkey
  et~al\mbox{.}}{2016}]%
        {lindley2016}
\bibfield{author}{\bibinfo{person}{Robert Atkey}, \bibinfo{person}{Sam
  Lindley}, {and} \bibinfo{person}{J.~Garrett Morris}.}
  \bibinfo{year}{2016}\natexlab{}.
\newblock \showarticletitle{Conflation {Confers} {Concurrency}}.
\newblock In \bibinfo{booktitle}{\emph{A {List} of {Successes} {That} {Can}
  {Change} the {World}}}, \bibfield{editor}{\bibinfo{person}{Sam Lindley},
  \bibinfo{person}{Conor McBride}, \bibinfo{person}{Phil Trinder}, {and}
  \bibinfo{person}{Don Sannella}} (Eds.). \bibinfo{series}{Lecture {Notes} in
  {Computer} {Science}}, Vol.~\bibinfo{volume}{9600}.
  \bibinfo{publisher}{Springer International Publishing},
  \bibinfo{pages}{32--55}.
\newblock
\showISBNx{978-3-319-30935-4}


\bibitem[\protect\citeauthoryear{Avron}{Avron}{1991}]%
        {avron1991hypersequents}
\bibfield{author}{\bibinfo{person}{Arnon Avron}.}
  \bibinfo{year}{1991}\natexlab{}.
\newblock \showarticletitle{Hypersequents, logical consequence and intermediate
  logics for concurrency}.
\newblock \bibinfo{journal}{\emph{Annals of Mathematics and Artificial
  Intelligence}} \bibinfo{volume}{4}, \bibinfo{number}{3-4}
  (\bibinfo{year}{1991}), \bibinfo{pages}{225--248}.
\newblock


\bibitem[\protect\citeauthoryear{Baelde}{Baelde}{2012}]%
        {baelde2012least}
\bibfield{author}{\bibinfo{person}{David Baelde}.}
  \bibinfo{year}{2012}\natexlab{}.
\newblock \showarticletitle{Least and greatest fixed points in linear logic}.
\newblock \bibinfo{journal}{\emph{ACM Transactions on Computational Logic}}
  \bibinfo{volume}{13}, \bibinfo{number}{1} (\bibinfo{year}{2012}),
  \bibinfo{pages}{1--44}.
\newblock


\bibitem[\protect\citeauthoryear{Balzer and Pfenning}{Balzer and
  Pfenning}{2017}]%
        {balzer2017}
\bibfield{author}{\bibinfo{person}{Stephanie Balzer} {and}
  \bibinfo{person}{Frank Pfenning}.} \bibinfo{year}{2017}\natexlab{}.
\newblock \showarticletitle{Manifest sharing with session types}.
\newblock \bibinfo{journal}{\emph{Proceedings of the ACM on Programming
  Languages}} \bibinfo{volume}{1}, \bibinfo{number}{ICFP}
  (\bibinfo{year}{2017}), \bibinfo{pages}{1--29}.
\newblock
\urldef\tempurl%
\url{https://doi.org/10.1145/3110281}
\showDOI{\tempurl}


\bibitem[\protect\citeauthoryear{Balzer, Toninho, and Pfenning}{Balzer
  et~al\mbox{.}}{2019}]%
        {balzer2019}
\bibfield{author}{\bibinfo{person}{Stephanie Balzer}, \bibinfo{person}{Bernardo
  Toninho}, {and} \bibinfo{person}{Frank Pfenning}.}
  \bibinfo{year}{2019}\natexlab{}.
\newblock \showarticletitle{Manifest {Deadlock}-{Freedom} for {Shared}
  {Session} {Types}}. In \bibinfo{booktitle}{\emph{Programming {Languages} and
  {Systems}}}, \bibfield{editor}{\bibinfo{person}{Luís Caires}} (Ed.),
  Vol.~\bibinfo{volume}{11423}. \bibinfo{publisher}{Springer International
  Publishing}, \bibinfo{address}{Cham}, \bibinfo{pages}{611--639}.
\newblock
\urldef\tempurl%
\url{https://doi.org/10.1007/978-3-030-17184-1\_22}
\showDOI{\tempurl}


\bibitem[\protect\citeauthoryear{Barr}{Barr}{1991}]%
        {barr1991}
\bibfield{author}{\bibinfo{person}{Michael Barr}.}
  \bibinfo{year}{1991}\natexlab{}.
\newblock \showarticletitle{$\ast$-{Autonomous} categories and linear logic}.
\newblock \bibinfo{journal}{\emph{Mathematical Structures in Computer Science}}
  \bibinfo{volume}{1}, \bibinfo{number}{2} (\bibinfo{year}{1991}),
  \bibinfo{pages}{159--178}.
\newblock
\showISSN{14698072}
\urldef\tempurl%
\url{https://doi.org/10.1017/S0960129500001274}
\showDOI{\tempurl}


\bibitem[\protect\citeauthoryear{Bellin}{Bellin}{1997}]%
        {bellin_subnets_1997}
\bibfield{author}{\bibinfo{person}{Gianluigi Bellin}.}
  \bibinfo{year}{1997}\natexlab{}.
\newblock \showarticletitle{Subnets of proof-nets in multiplicative linear
  logic with {MIX}}.
\newblock \bibinfo{journal}{\emph{Mathematical Structures in Computer Science}}
  \bibinfo{volume}{7}, \bibinfo{number}{6} (\bibinfo{year}{1997}),
  \bibinfo{pages}{663--669}.
\newblock
\urldef\tempurl%
\url{https://doi.org/10.1017/S0960129597002326}
\showDOI{\tempurl}


\bibitem[\protect\citeauthoryear{Bellin and Scott}{Bellin and Scott}{1994}]%
        {bellin1994p}
\bibfield{author}{\bibinfo{person}{G. Bellin} {and} \bibinfo{person}{P.~J.
  Scott}.} \bibinfo{year}{1994}\natexlab{}.
\newblock \showarticletitle{On the $\pi$-calculus and linear logic}.
\newblock \bibinfo{journal}{\emph{Theoretical Computer Science}}
  \bibinfo{volume}{135}, \bibinfo{number}{1} (\bibinfo{year}{1994}),
  \bibinfo{pages}{11--65}.
\newblock
\showISSN{03043975}
\urldef\tempurl%
\url{https://doi.org/10.1016/0304-3975(94)00104-9}
\showDOI{\tempurl}


\bibitem[\protect\citeauthoryear{Benton and Wadler}{Benton and Wadler}{1996}]%
        {benton_linear_1996}
\bibfield{author}{\bibinfo{person}{N Benton} {and} \bibinfo{person}{P Wadler}.}
  \bibinfo{year}{1996}\natexlab{}.
\newblock \showarticletitle{Linear logic, monads and the lambda calculus}. In
  \bibinfo{booktitle}{\emph{Proceedings 11th {Annual} {IEEE} {Symposium} on
  {Logic} in {Computer} {Science}}}. \bibinfo{publisher}{IEEE}.
\newblock
\urldef\tempurl%
\url{https://doi.org/10.1109/LICS.1996.561458}
\showDOI{\tempurl}


\bibitem[\protect\citeauthoryear{Blumofe, Joerg, Kuszmaul, Leiserson, Randall,
  and Zhou}{Blumofe et~al\mbox{.}}{1995}]%
        {blumofe1995cilk}
\bibfield{author}{\bibinfo{person}{Robert~D Blumofe},
  \bibinfo{person}{Christopher~F Joerg}, \bibinfo{person}{Bradley~C Kuszmaul},
  \bibinfo{person}{Charles~E Leiserson}, \bibinfo{person}{Keith~H Randall},
  {and} \bibinfo{person}{Yuli Zhou}.} \bibinfo{year}{1995}\natexlab{}.
\newblock \showarticletitle{Cilk: An efficient multithreaded runtime system}.
\newblock \bibinfo{journal}{\emph{ACM SigPlan Notices}} \bibinfo{volume}{30},
  \bibinfo{number}{8} (\bibinfo{year}{1995}), \bibinfo{pages}{207--216}.
\newblock


\bibitem[\protect\citeauthoryear{Caires and P{\'e}rez}{Caires and
  P{\'e}rez}{2017}]%
        {caires_2017}
\bibfield{author}{\bibinfo{person}{Lu{\'i}s Caires} {and}
  \bibinfo{person}{Jorge~A. P{\'e}rez}.} \bibinfo{year}{2017}\natexlab{}.
\newblock \showarticletitle{Linearity, Control Effects, and Behavioral Types}.
  In \bibinfo{booktitle}{\emph{Programming Languages and Systems. ESOP 2017}},
  \bibfield{editor}{\bibinfo{person}{Hongseok Yang}} (Ed.).
  \bibinfo{publisher}{Springer Berlin Heidelberg}, \bibinfo{address}{Berlin,
  Heidelberg}, \bibinfo{pages}{229--259}.
\newblock
\showISBNx{978-3-662-54434-1}
\urldef\tempurl%
\url{https://doi.org/10.1007/978-3-662-54434-1\_9}
\showDOI{\tempurl}


\bibitem[\protect\citeauthoryear{Caires and Pfenning}{Caires and
  Pfenning}{2010}]%
        {caires2010}
\bibfield{author}{\bibinfo{person}{Luís Caires} {and} \bibinfo{person}{Frank
  Pfenning}.} \bibinfo{year}{2010}\natexlab{}.
\newblock \showarticletitle{Session {Types} as {Intuitionistic} {Linear}
  {Propositions}}. In \bibinfo{booktitle}{\emph{{CONCUR} 2010 - {Concurrency}
  {Theory}}} \emph{(\bibinfo{series}{Lecture {Notes} in {Computer} {Science}},
  Vol.~\bibinfo{volume}{6269})}, \bibfield{editor}{\bibinfo{person}{Paul
  Gastin} {and} \bibinfo{person}{François Laroussinie}} (Eds.).
  \bibinfo{publisher}{Springer Berlin Heidelberg}, \bibinfo{address}{Berlin,
  Heidelberg}, \bibinfo{pages}{222--236}.
\newblock


\bibitem[\protect\citeauthoryear{Caires, Pfenning, and Toninho}{Caires
  et~al\mbox{.}}{2016}]%
        {caires_linear_2016}
\bibfield{author}{\bibinfo{person}{Luís Caires}, \bibinfo{person}{Frank
  Pfenning}, {and} \bibinfo{person}{Bernardo Toninho}.}
  \bibinfo{year}{2016}\natexlab{}.
\newblock \showarticletitle{Linear logic propositions as session types}.
\newblock \bibinfo{journal}{\emph{Mathematical Structures in Computer Science}}
  \bibinfo{volume}{26}, \bibinfo{number}{3} (\bibinfo{year}{2016}),
  \bibinfo{pages}{367--423}.
\newblock
\showISSN{0960-1295, 1469-8072}
\urldef\tempurl%
\url{https://doi.org/10.1017/S0960129514000218}
\showDOI{\tempurl}


\bibitem[\protect\citeauthoryear{Carbone, Lindley, Montesi, Sch{\"u}rmann, and
  Wadler}{Carbone et~al\mbox{.}}{2016}]%
        {carbone_2016}
\bibfield{author}{\bibinfo{person}{Marco Carbone}, \bibinfo{person}{Sam
  Lindley}, \bibinfo{person}{Fabrizio Montesi}, \bibinfo{person}{Carsten
  Sch{\"u}rmann}, {and} \bibinfo{person}{Philip Wadler}.}
  \bibinfo{year}{2016}\natexlab{}.
\newblock \showarticletitle{{Coherence Generalises Duality: A Logical
  Explanation of Multiparty Session Types}}. In \bibinfo{booktitle}{\emph{27th
  International Conference on Concurrency Theory (CONCUR 2016)}}
  \emph{(\bibinfo{series}{Leibniz International Proceedings in Informatics
  (LIPIcs)}, Vol.~\bibinfo{volume}{59})},
  \bibfield{editor}{\bibinfo{person}{Jos{\'e}e Desharnais} {and}
  \bibinfo{person}{Radha Jagadeesan}} (Eds.). \bibinfo{publisher}{Schloss
  Dagstuhl--Leibniz-Zentrum fuer Informatik}, \bibinfo{address}{Dagstuhl,
  Germany}, \bibinfo{pages}{33:1--33:15}.
\newblock
\showISBNx{978-3-95977-017-0}
\showISSN{1868-8969}
\urldef\tempurl%
\url{https://doi.org/10.4230/LIPIcs.CONCUR.2016.33}
\showDOI{\tempurl}


\bibitem[\protect\citeauthoryear{Carbone, Montesi, Schürmann, and
  Yoshida}{Carbone et~al\mbox{.}}{2017}]%
        {carbone_2017}
\bibfield{author}{\bibinfo{person}{Marco Carbone}, \bibinfo{person}{Fabrizio
  Montesi}, \bibinfo{person}{Carsten Schürmann}, {and} \bibinfo{person}{Nobuko
  Yoshida}.} \bibinfo{year}{2017}\natexlab{}.
\newblock \showarticletitle{Multiparty session types as coherence proofs}.
\newblock \bibinfo{journal}{\emph{Acta Informatica}}  \bibinfo{volume}{54}
  (\bibinfo{year}{2017}), \bibinfo{pages}{243--269}.
\newblock
\showISSN{0001-5903}
\urldef\tempurl%
\url{https://doi.org/10.1007/s00236-016-0285-y}
\showDOI{\tempurl}


\bibitem[\protect\citeauthoryear{Castellan, Yoshida, and Stefanesco}{Castellan
  et~al\mbox{.}}{2020}]%
        {castellan2020}
\bibfield{author}{\bibinfo{person}{Simon Castellan}, \bibinfo{person}{Nobuko
  Yoshida}, {and} \bibinfo{person}{L\'{e}o Stefanesco}.}
  \bibinfo{year}{2020}\natexlab{}.
\newblock \showarticletitle{Game {Semantics}: {Easy} as {Pi}}.
\newblock \bibinfo{journal}{\emph{arXiv:2011.05248 [cs]}} (\bibinfo{date}{Nov.}
  \bibinfo{year}{2020}).
\newblock
\showeprint{2011.05248}
\urldef\tempurl%
\url{http://arxiv.org/abs/2011.05248}
\showURL{%
\tempurl}


\bibitem[\protect\citeauthoryear{Conway}{Conway}{1963}]%
        {conway1963multiprocessor}
\bibfield{author}{\bibinfo{person}{Melvin~E Conway}.}
  \bibinfo{year}{1963}\natexlab{}.
\newblock \showarticletitle{A multiprocessor system design}. In
  \bibinfo{booktitle}{\emph{Proceedings of the November 12-14, 1963, fall joint
  computer conference}}. \bibinfo{pages}{139--146}.
\newblock


\bibitem[\protect\citeauthoryear{Coppo, Dezani-Ciancaglini, Yoshida, and
  Padovani}{Coppo et~al\mbox{.}}{2016}]%
        {coppo_2016}
\bibfield{author}{\bibinfo{person}{Mario Coppo}, \bibinfo{person}{Mariangiola
  Dezani-Ciancaglini}, \bibinfo{person}{Nobuko Yoshida}, {and}
  \bibinfo{person}{Luca Padovani}.} \bibinfo{year}{2016}\natexlab{}.
\newblock \showarticletitle{Global progress for dynamically interleaved
  multiparty sessions}.
\newblock \bibinfo{journal}{\emph{Mathematical Structures in Computer Science}}
  \bibinfo{volume}{26}, \bibinfo{number}{2} (\bibinfo{year}{2016}),
  \bibinfo{pages}{238--302}.
\newblock
\showISSN{0960-1295, 1469-8072}
\urldef\tempurl%
\url{https://doi.org/10.1017/S0960129514000188}
\showDOI{\tempurl}


\bibitem[\protect\citeauthoryear{Dagum and Menon}{Dagum and Menon}{1998}]%
        {dagum1998openmp}
\bibfield{author}{\bibinfo{person}{Leonardo Dagum} {and}
  \bibinfo{person}{Ramesh Menon}.} \bibinfo{year}{1998}\natexlab{}.
\newblock \showarticletitle{OpenMP: an industry standard API for shared-memory
  programming}.
\newblock \bibinfo{journal}{\emph{IEEE computational science and engineering}}
  \bibinfo{volume}{5}, \bibinfo{number}{1} (\bibinfo{year}{1998}),
  \bibinfo{pages}{46--55}.
\newblock


\bibitem[\protect\citeauthoryear{Deni\'{e}lou and Yoshida}{Deni\'{e}lou and
  Yoshida}{2011}]%
        {denielou_2011}
\bibfield{author}{\bibinfo{person}{Pierre-Malo Deni\'{e}lou} {and}
  \bibinfo{person}{Nobuko Yoshida}.} \bibinfo{year}{2011}\natexlab{}.
\newblock \showarticletitle{Dynamic {Multirole} {Session} {Types}}. In
  \bibinfo{booktitle}{\emph{Proceedings of the 38th {Annual} {ACM}
  {SIGPLAN}-{SIGACT} {Symposium} on {Principles} of {Programming} {Languages} -
  {POPL} '11}}. \bibinfo{publisher}{ACM Press}, \bibinfo{pages}{435}.
\newblock
\showISBNx{978-1-4503-0490-0}
\urldef\tempurl%
\url{https://doi.org/10.1145/1926385.1926435}
\showDOI{\tempurl}


\bibitem[\protect\citeauthoryear{Ehrhard}{Ehrhard}{2018}]%
        {ehrhard2018}
\bibfield{author}{\bibinfo{person}{Thomas Ehrhard}.}
  \bibinfo{year}{2018}\natexlab{}.
\newblock \showarticletitle{An introduction to differential linear logic:
  proof-nets, models and antiderivatives}.
\newblock \bibinfo{journal}{\emph{Mathematical Structures in Computer Science}}
  \bibinfo{volume}{28}, \bibinfo{number}{7} (\bibinfo{year}{2018}),
  \bibinfo{pages}{995--1060}.
\newblock
\urldef\tempurl%
\url{https://doi.org/10.1017/S0960129516000372}
\showDOI{\tempurl}


\bibitem[\protect\citeauthoryear{Ehrhard and Jafarrahmani}{Ehrhard and
  Jafarrahmani}{2021}]%
        {ehrhard2021categorical}
\bibfield{author}{\bibinfo{person}{Thomas Ehrhard} {and}
  \bibinfo{person}{Farzad Jafarrahmani}.} \bibinfo{year}{2021}\natexlab{}.
\newblock \bibinfo{title}{Categorical models of Linear Logic with fixed points
  of formulas}.
\newblock
\newblock
\showeprint[arxiv]{2011.10209}~[cs.LO]


\bibitem[\protect\citeauthoryear{Ehrhard and Laurent}{Ehrhard and
  Laurent}{2010}]%
        {ehrhard2010}
\bibfield{author}{\bibinfo{person}{Thomas Ehrhard} {and}
  \bibinfo{person}{Olivier Laurent}.} \bibinfo{year}{2010}\natexlab{}.
\newblock \showarticletitle{Interpreting a finitary pi-calculus in differential
  interaction nets}.
\newblock \bibinfo{journal}{\emph{Information and Computation}}
  \bibinfo{volume}{208}, \bibinfo{number}{6} (\bibinfo{year}{2010}),
  \bibinfo{pages}{606--633}.
\newblock
\urldef\tempurl%
\url{https://doi.org/10.1016/j.ic.2009.06.005}
\showDOI{\tempurl}


\bibitem[\protect\citeauthoryear{Fowler, Lindley, Morris, and Decova}{Fowler
  et~al\mbox{.}}{2019}]%
        {fowler2019exceptional}
\bibfield{author}{\bibinfo{person}{Simon Fowler}, \bibinfo{person}{Sam
  Lindley}, \bibinfo{person}{J.~Garrett Morris}, {and} \bibinfo{person}{Sára
  Decova}.} \bibinfo{year}{2019}\natexlab{}.
\newblock \showarticletitle{Exceptional asynchronous session types: session
  types without tiers}.
\newblock \bibinfo{journal}{\emph{Proceedings of the ACM on Programming
  Languages}} \bibinfo{volume}{3}, \bibinfo{number}{POPL}
  (\bibinfo{year}{2019}).
\newblock
\urldef\tempurl%
\url{https://doi.org/10.1145/3290341}
\showDOI{\tempurl}


\bibitem[\protect\citeauthoryear{Gay and Vasconcelos}{Gay and
  Vasconcelos}{2010}]%
        {gay2010linear}
\bibfield{author}{\bibinfo{person}{Simon~J Gay} {and} \bibinfo{person}{Vasco~T
  Vasconcelos}.} \bibinfo{year}{2010}\natexlab{}.
\newblock \showarticletitle{Linear type theory for asynchronous session types}.
\newblock \bibinfo{journal}{\emph{Journal of Functional Programming}}
  \bibinfo{volume}{20}, \bibinfo{number}{1} (\bibinfo{year}{2010}),
  \bibinfo{pages}{19}.
\newblock


\bibitem[\protect\citeauthoryear{Girard}{Girard}{1987}]%
        {girard1987linear}
\bibfield{author}{\bibinfo{person}{Jean-Yves Girard}.}
  \bibinfo{year}{1987}\natexlab{}.
\newblock \showarticletitle{Linear logic}.
\newblock \bibinfo{journal}{\emph{Theoretical Computer Science}}
  \bibinfo{volume}{50}, \bibinfo{number}{1} (\bibinfo{year}{1987}),
  \bibinfo{pages}{1--101}.
\newblock
\showISSN{03043975}
\urldef\tempurl%
\url{https://doi.org/10.1016/0304-3975(87)90045-4}
\showDOI{\tempurl}


\bibitem[\protect\citeauthoryear{Girard and Lafont}{Girard and Lafont}{1987}]%
        {girardlafont1987}
\bibfield{author}{\bibinfo{person}{J.~Y. Girard} {and} \bibinfo{person}{Y.
  Lafont}.} \bibinfo{year}{1987}\natexlab{}.
\newblock \showarticletitle{Linear logic and lazy computation}. In
  \bibinfo{booktitle}{\emph{{TAPSOFT} '87}} \emph{(\bibinfo{series}{Lecture
  {Notes} in {Computer} {Science}}, Vol.~\bibinfo{volume}{250})},
  \bibfield{editor}{\bibinfo{person}{Hartmut Ehrig}, \bibinfo{person}{Robert
  Kowalski}, \bibinfo{person}{Giorgio Levi}, {and} \bibinfo{person}{Ugo
  Montanari}} (Eds.). \bibinfo{publisher}{Springer-Verlag},
  \bibinfo{address}{Berlin/Heidelberg}, \bibinfo{pages}{52--66}.
\newblock
\showISBNx{978-3-540-17611-4}
\urldef\tempurl%
\url{https://doi.org/10.1007/BFb0014972}
\showDOI{\tempurl}


\bibitem[\protect\citeauthoryear{Girard, Scedrov, and Scott}{Girard
  et~al\mbox{.}}{1992}]%
        {girard_1992}
\bibfield{author}{\bibinfo{person}{Jean-Yves Girard}, \bibinfo{person}{Andre
  Scedrov}, {and} \bibinfo{person}{Philip~J. Scott}.}
  \bibinfo{year}{1992}\natexlab{}.
\newblock \showarticletitle{Bounded linear logic: a modular approach to
  polynomial-time computability}.
\newblock \bibinfo{journal}{\emph{Theoretical Computer Science}}
  \bibinfo{volume}{97}, \bibinfo{number}{1} (\bibinfo{year}{1992}),
  \bibinfo{pages}{1--66}.
\newblock
\showISSN{03043975}
\urldef\tempurl%
\url{https://doi.org/10.1016/0304-3975(92)90386-T}
\showDOI{\tempurl}


\bibitem[\protect\citeauthoryear{Halstead~Jr}{Halstead~Jr}{1984}]%
        {halstead1984implementation}
\bibfield{author}{\bibinfo{person}{Robert~H Halstead~Jr}.}
  \bibinfo{year}{1984}\natexlab{}.
\newblock \showarticletitle{Implementation of Multilisp: Lisp on a
  multiprocessor}. In \bibinfo{booktitle}{\emph{Proceedings of the 1984 ACM
  Symposium on LISP and functional programming}}. \bibinfo{pages}{9--17}.
\newblock


\bibitem[\protect\citeauthoryear{Herlihy and Shavit}{Herlihy and
  Shavit}{2012}]%
        {herlihy_art_2012}
\bibfield{author}{\bibinfo{person}{Maurice Herlihy} {and} \bibinfo{person}{Nir
  Shavit}.} \bibinfo{year}{2012}\natexlab{}.
\newblock \bibinfo{booktitle}{\emph{The {Art} of {Multiprocessor}
  {Programming}} (\bibinfo{edition}{revised first} ed.)}.
\newblock \bibinfo{publisher}{Morgan Kaufmann}.
\newblock
\showISBNx{978-0-12-397337-5}


\bibitem[\protect\citeauthoryear{Hinrichsen, Bengtson, and Krebbers}{Hinrichsen
  et~al\mbox{.}}{2019}]%
        {hinrichsen2019actris}
\bibfield{author}{\bibinfo{person}{Jonas~Kastberg Hinrichsen},
  \bibinfo{person}{Jesper Bengtson}, {and} \bibinfo{person}{Robbert Krebbers}.}
  \bibinfo{year}{2019}\natexlab{}.
\newblock \showarticletitle{Actris: Session-type based reasoning in separation
  logic}.
\newblock \bibinfo{journal}{\emph{Proceedings of the ACM on Programming
  Languages}} \bibinfo{volume}{4}, \bibinfo{number}{POPL}
  (\bibinfo{year}{2019}), \bibinfo{pages}{1--30}.
\newblock


\bibitem[\protect\citeauthoryear{Honda, Vasconcelos, and Kubo}{Honda
  et~al\mbox{.}}{1998}]%
        {honda1998}
\bibfield{author}{\bibinfo{person}{Kohei Honda}, \bibinfo{person}{Vasco~T
  Vasconcelos}, {and} \bibinfo{person}{Makoto Kubo}.}
  \bibinfo{year}{1998}\natexlab{}.
\newblock \showarticletitle{Language primitives and type discipline for
  structured communication-based programming}. In
  \bibinfo{booktitle}{\emph{Programming {Languages} and {Systems}:
  {Proceedings} of the 7th {European} {Symposium} on {Programming}
  ({ESOP}'98)}} \emph{(\bibinfo{series}{Lecture {Notes} in {Computer}
  {Science}}, Vol.~\bibinfo{volume}{1381})}. \bibinfo{publisher}{Springer,
  Berlin, Heidelberg}, \bibinfo{pages}{122--138}.
\newblock
\urldef\tempurl%
\url{https://doi.org/10.1007/BFb0053567}
\showDOI{\tempurl}


\bibitem[\protect\citeauthoryear{Honda, Yoshida, and Carbone}{Honda
  et~al\mbox{.}}{2008}]%
        {honda_2008}
\bibfield{author}{\bibinfo{person}{Kohei Honda}, \bibinfo{person}{Nobuko
  Yoshida}, {and} \bibinfo{person}{Marco Carbone}.}
  \bibinfo{year}{2008}\natexlab{}.
\newblock \showarticletitle{Multiparty {Asynchronous} {Session} {Types}}. In
  \bibinfo{booktitle}{\emph{Proceedings of the 35th {Annual} {ACM}
  {SIGPLAN}-{SIGACT} {Symposium} on {Principles} of {Programming}
  {Languages}}}. \bibinfo{publisher}{ACM Press}.
\newblock
\showISBNx{978-1-59593-689-9}
\urldef\tempurl%
\url{https://doi.org/10.1145/1328438.1328472}
\showDOI{\tempurl}


\bibitem[\protect\citeauthoryear{Honda, Yoshida, and Carbone}{Honda
  et~al\mbox{.}}{2016}]%
        {honda_2016}
\bibfield{author}{\bibinfo{person}{Kohei Honda}, \bibinfo{person}{Nobuko
  Yoshida}, {and} \bibinfo{person}{Marco Carbone}.}
  \bibinfo{year}{2016}\natexlab{}.
\newblock \showarticletitle{Multiparty {Asynchronous} {Session} {Types}}.
\newblock \bibinfo{journal}{\emph{J. ACM}} \bibinfo{volume}{63},
  \bibinfo{number}{1} (\bibinfo{year}{2016}), \bibinfo{pages}{1--67}.
\newblock
\showISSN{0004-5411, 1557-735X}
\urldef\tempurl%
\url{https://doi.org/10.1145/2827695}
\showDOI{\tempurl}


\bibitem[\protect\citeauthoryear{Keynes}{Keynes}{1936}]%
        {keynes_1936}
\bibfield{author}{\bibinfo{person}{John~Maynard Keynes}.}
  \bibinfo{year}{1936}\natexlab{}.
\newblock \bibinfo{booktitle}{\emph{The {General} {Theory} of {Employment},
  {Interest} and {Money}}}.
\newblock \bibinfo{publisher}{Macmillan \& Co. Ltd.},
  \bibinfo{address}{London}.
\newblock


\bibitem[\protect\citeauthoryear{Kokke, Montesi, and Peressotti}{Kokke
  et~al\mbox{.}}{2019a}]%
        {kokke2019better}
\bibfield{author}{\bibinfo{person}{Wen Kokke}, \bibinfo{person}{Fabrizio
  Montesi}, {and} \bibinfo{person}{Marco Peressotti}.}
  \bibinfo{year}{2019}\natexlab{a}.
\newblock \showarticletitle{Better late than never: a fully-abstract semantics
  for classical processes}.
\newblock \bibinfo{journal}{\emph{Proceedings of the ACM on Programming
  Languages}} \bibinfo{volume}{3}, \bibinfo{number}{POPL}
  (\bibinfo{year}{2019}), \bibinfo{pages}{1--29}.
\newblock


\bibitem[\protect\citeauthoryear{Kokke, Montesi, and Peressotti}{Kokke
  et~al\mbox{.}}{2019b}]%
        {kokke2019taking}
\bibfield{author}{\bibinfo{person}{Wen Kokke}, \bibinfo{person}{Fabrizio
  Montesi}, {and} \bibinfo{person}{Marco Peressotti}.}
  \bibinfo{year}{2019}\natexlab{b}.
\newblock \bibinfo{title}{Taking linear logic apart}.  (\bibinfo{year}{2019}).
\newblock
\urldef\tempurl%
\url{https://arxiv.org/abs/1904.06848}
\showURL{%
\tempurl}


\bibitem[\protect\citeauthoryear{Kokke, Morris, and Wadler}{Kokke
  et~al\mbox{.}}{2019c}]%
        {kokke2019races}
\bibfield{author}{\bibinfo{person}{Wen Kokke}, \bibinfo{person}{J.~Garrett
  Morris}, {and} \bibinfo{person}{Philip Wadler}.}
  \bibinfo{year}{2019}\natexlab{c}.
\newblock \bibinfo{title}{Towards Races in Linear Logic}.
\newblock
\newblock
\showeprint[arxiv]{1909.13376}~[cs.LO]


\bibitem[\protect\citeauthoryear{Leijen, Schulte, and Burckhardt}{Leijen
  et~al\mbox{.}}{2009}]%
        {leijen2009design}
\bibfield{author}{\bibinfo{person}{Daan Leijen}, \bibinfo{person}{Wolfram
  Schulte}, {and} \bibinfo{person}{Sebastian Burckhardt}.}
  \bibinfo{year}{2009}\natexlab{}.
\newblock \showarticletitle{The design of a task parallel library}.
\newblock \bibinfo{journal}{\emph{Acm Sigplan Notices}} \bibinfo{volume}{44},
  \bibinfo{number}{10} (\bibinfo{year}{2009}), \bibinfo{pages}{227--242}.
\newblock


\bibitem[\protect\citeauthoryear{Lindley and Morris}{Lindley and
  Morris}{2015}]%
        {lindley2015semantics}
\bibfield{author}{\bibinfo{person}{Sam Lindley} {and}
  \bibinfo{person}{J~Garrett Morris}.} \bibinfo{year}{2015}\natexlab{}.
\newblock \showarticletitle{A semantics for propositions as sessions}. In
  \bibinfo{booktitle}{\emph{European Symposium on Programming Languages and
  Systems}}. Springer, \bibinfo{pages}{560--584}.
\newblock


\bibitem[\protect\citeauthoryear{Lindley and Morris}{Lindley and
  Morris}{2016}]%
        {lindley2016talking}
\bibfield{author}{\bibinfo{person}{Sam Lindley} {and}
  \bibinfo{person}{J.~Garrett Morris}.} \bibinfo{year}{2016}\natexlab{}.
\newblock \showarticletitle{Talking bananas: structural recursion for session
  types}. In \bibinfo{booktitle}{\emph{Proceedings of the 21st {ACM} {SIGPLAN}
  {International} {Conference} on {Functional} {Programming} - {ICFP} 2016}}.
  \bibinfo{publisher}{ACM Press}, \bibinfo{address}{Nara, Japan},
  \bibinfo{pages}{434--447}.
\newblock
\showISBNx{978-1-4503-4219-3}
\urldef\tempurl%
\url{https://doi.org/10.1145/2951913.2951921}
\showDOI{\tempurl}


\bibitem[\protect\citeauthoryear{Lindley and Morris}{Lindley and
  Morris}{2017}]%
        {lindley_lightweight_2017}
\bibfield{author}{\bibinfo{person}{Sam Lindley} {and}
  \bibinfo{person}{J.~Garrett Morris}.} \bibinfo{year}{2017}\natexlab{}.
\newblock \showarticletitle{Lightweight {Functional} {Session} {Types}}.
\newblock In \bibinfo{booktitle}{\emph{Behavioural {Types}: from {Theory} to
  {Tools}}}, \bibfield{editor}{\bibinfo{person}{Simon Gay} {and}
  \bibinfo{person}{Antonio Ravara}} (Eds.). \bibinfo{publisher}{River
  Publishers}.
\newblock
\urldef\tempurl%
\url{https://doi.org/10.13052/rp-9788793519817}
\showDOI{\tempurl}


\bibitem[\protect\citeauthoryear{Maraist, Odersky, Turner, and Wadler}{Maraist
  et~al\mbox{.}}{1999}]%
        {maraist1999}
\bibfield{author}{\bibinfo{person}{J. Maraist}, \bibinfo{person}{M. Odersky},
  \bibinfo{person}{D.N. Turner}, {and} \bibinfo{person}{P. Wadler}.}
  \bibinfo{year}{1999}\natexlab{}.
\newblock \showarticletitle{Call-by-name, call-by-value, call-by-need and the
  linear lambda calculus}.
\newblock \bibinfo{journal}{\emph{Theoretical Computer Science}}
  \bibinfo{volume}{228}, \bibinfo{number}{1-2} (\bibinfo{year}{1999}),
  \bibinfo{pages}{175--210}.
\newblock
\urldef\tempurl%
\url{https://doi.org/10.1016/S0304-3975(98)00358-2}
\showDOI{\tempurl}


\bibitem[\protect\citeauthoryear{Maraist, Odersky, Turner, and Wadler}{Maraist
  et~al\mbox{.}}{1995}]%
        {maraistcall1995}
\bibfield{author}{\bibinfo{person}{John Maraist}, \bibinfo{person}{Martin
  Odersky}, \bibinfo{person}{David~N. Turner}, {and} \bibinfo{person}{Philip
  Wadler}.} \bibinfo{year}{1995}\natexlab{}.
\newblock \showarticletitle{Call-by-name, {Call}-by-value, {Call}-by-need, and
  the {Linear} {Lambda} {Calculus}}.
\newblock \bibinfo{journal}{\emph{Electronic Notes in Theoretical Computer
  Science}}  \bibinfo{volume}{1} (\bibinfo{year}{1995}),
  \bibinfo{pages}{370--392}.
\newblock
\showISSN{15710661}
\urldef\tempurl%
\url{https://doi.org/10.1016/S1571-0661(04)00022-2}
\showDOI{\tempurl}


\bibitem[\protect\citeauthoryear{Mazza}{Mazza}{2018}]%
        {mazza2018}
\bibfield{author}{\bibinfo{person}{Damiano Mazza}.}
  \bibinfo{year}{2018}\natexlab{}.
\newblock \showarticletitle{The true concurrency of differential interaction
  nets}.
\newblock \bibinfo{journal}{\emph{Mathematical Structures in Computer Science}}
  \bibinfo{volume}{28}, \bibinfo{number}{7} (\bibinfo{date}{Aug.}
  \bibinfo{year}{2018}), \bibinfo{pages}{1097--1125}.
\newblock
\urldef\tempurl%
\url{https://doi.org/10.1017/S0960129516000402}
\showDOI{\tempurl}


\bibitem[\protect\citeauthoryear{Melli\`{e}s}{Melli\`{e}s}{2009}]%
        {mellies2009categorical}
\bibfield{author}{\bibinfo{person}{Paul-Andr\'{e} Melli\`{e}s}.}
  \bibinfo{year}{2009}\natexlab{}.
\newblock \showarticletitle{Categorical {Semantics} of {Linear} {Logic}}.
\newblock In \bibinfo{booktitle}{\emph{Panoramas et synth\`{e}ses 27:
  {Interactive} models of computation and program behaviour}},
  \bibfield{editor}{\bibinfo{person}{Pierre-Louis Curien},
  \bibinfo{person}{Hugo Herbelin}, \bibinfo{person}{Jean-Louis Krivine}, {and}
  \bibinfo{person}{Paul-André Melliès}} (Eds.).
  \bibinfo{publisher}{Soci\'{e}t\'{e} Math\'{e}matique de France}.
\newblock
\showISBNx{978-2-85629-273-0}
\urldef\tempurl%
\url{http://www.pps.univ-paris-diderot.fr/$\sim$mellies/papers/panorama.pdf}
\showURL{%
\tempurl}


\bibitem[\protect\citeauthoryear{Melli\`{e}s, Tabareau, and Tasson}{Melli\`{e}s
  et~al\mbox{.}}{2018}]%
        {mellies2018}
\bibfield{author}{\bibinfo{person}{Paul-André Melli\`{e}s},
  \bibinfo{person}{Nicolas Tabareau}, {and} \bibinfo{person}{Christine
  Tasson}.} \bibinfo{year}{2018}\natexlab{}.
\newblock \showarticletitle{An explicit formula for the free exponential
  modality of linear logic}.
\newblock \bibinfo{journal}{\emph{Mathematical Structures in Computer Science}}
  \bibinfo{volume}{28}, \bibinfo{number}{7} (\bibinfo{year}{2018}).
\newblock
\showISSN{0960-1295, 1469-8072}
\urldef\tempurl%
\url{https://doi.org/10.1017/S0960129516000426}
\showDOI{\tempurl}


\bibitem[\protect\citeauthoryear{Merro and Sangiorgi}{Merro and
  Sangiorgi}{2004}]%
        {merro_2004}
\bibfield{author}{\bibinfo{person}{Massimo Merro} {and} \bibinfo{person}{Davide
  Sangiorgi}.} \bibinfo{year}{2004}\natexlab{}.
\newblock \showarticletitle{On asynchrony in name-passing calculi}.
\newblock \bibinfo{journal}{\emph{Mathematical Structures in Computer Science}}
  \bibinfo{volume}{14}, \bibinfo{number}{5} (\bibinfo{year}{2004}),
  \bibinfo{pages}{715--767}.
\newblock
\urldef\tempurl%
\url{https://doi.org/10.1017/S0960129504004323}
\showDOI{\tempurl}


\bibitem[\protect\citeauthoryear{Milner}{Milner}{1992}]%
        {milner_functions_1992}
\bibfield{author}{\bibinfo{person}{Robin Milner}.}
  \bibinfo{year}{1992}\natexlab{}.
\newblock \showarticletitle{Functions as processes}.
\newblock \bibinfo{journal}{\emph{Mathematical Structures in Computer Science}}
  \bibinfo{volume}{2}, \bibinfo{number}{2} (\bibinfo{year}{1992}),
  \bibinfo{pages}{119--141}.
\newblock
\urldef\tempurl%
\url{https://doi.org/10.1017/S0960129500001407}
\showDOI{\tempurl}


\bibitem[\protect\citeauthoryear{Milner}{Milner}{1999}]%
        {milner_communicating_1999}
\bibfield{author}{\bibinfo{person}{Robin Milner}.}
  \bibinfo{year}{1999}\natexlab{}.
\newblock \bibinfo{booktitle}{\emph{Communicating and {Mobile} {Systems}: {The}
  $\pi$-calculus}}.
\newblock \bibinfo{publisher}{Cambridge University Press},
  \bibinfo{address}{New York, NY, USA}.
\newblock
\showISBNx{0-521-65869-1}


\bibitem[\protect\citeauthoryear{Milner, Parrow, and Walker}{Milner
  et~al\mbox{.}}{1992}]%
        {milner1992calculus}
\bibfield{author}{\bibinfo{person}{Robin Milner}, \bibinfo{person}{Joachim
  Parrow}, {and} \bibinfo{person}{David Walker}.}
  \bibinfo{year}{1992}\natexlab{}.
\newblock \showarticletitle{A calculus of mobile processes, i}.
\newblock \bibinfo{journal}{\emph{Information and computation}}
  \bibinfo{volume}{100}, \bibinfo{number}{1} (\bibinfo{year}{1992}),
  \bibinfo{pages}{1--40}.
\newblock


\bibitem[\protect\citeauthoryear{Montesi and Peressotti}{Montesi and
  Peressotti}{2018}]%
        {montesi2018}
\bibfield{author}{\bibinfo{person}{Fabrizio Montesi} {and}
  \bibinfo{person}{Marco Peressotti}.} \bibinfo{year}{2018}\natexlab{}.
\newblock \bibinfo{title}{Classical {Transitions}}.  (\bibinfo{year}{2018}).
\newblock
\urldef\tempurl%
\url{https://arxiv.org/abs/1803.01049}
\showURL{%
\tempurl}


\bibitem[\protect\citeauthoryear{Reed}{Reed}{2009}]%
        {reed2009}
\bibfield{author}{\bibinfo{person}{Jason Reed}.}
  \bibinfo{year}{2009}\natexlab{}.
\newblock \bibinfo{title}{A {Judgmental} {Deconstruction} of {Modal} {Logic}}.
  (\bibinfo{year}{2009}).
\newblock
\urldef\tempurl%
\url{http://www.cs.cmu.edu/~jcreed/papers/jdml.pdf}
\showURL{%
\tempurl}


\bibitem[\protect\citeauthoryear{Reinders}{Reinders}{2007}]%
        {reinders2007intel}
\bibfield{author}{\bibinfo{person}{James Reinders}.}
  \bibinfo{year}{2007}\natexlab{}.
\newblock \bibinfo{booktitle}{\emph{Intel threading building blocks: outfitting
  C++ for multi-core processor parallelism}}.
\newblock \bibinfo{publisher}{" O'Reilly Media, Inc."}.
\newblock


\bibitem[\protect\citeauthoryear{Toninho, Caires, and Pfenning}{Toninho
  et~al\mbox{.}}{2014}]%
        {toninho2014corecursion}
\bibfield{author}{\bibinfo{person}{Bernardo Toninho},
  \bibinfo{person}{Lu{\'\i}s Caires}, {and} \bibinfo{person}{Frank Pfenning}.}
  \bibinfo{year}{2014}\natexlab{}.
\newblock \showarticletitle{Corecursion and non-divergence in session-typed
  processes}. In \bibinfo{booktitle}{\emph{International Symposium on
  Trustworthy Global Computing}}. Springer, \bibinfo{pages}{159--175}.
\newblock


\bibitem[\protect\citeauthoryear{van Steen and Tanenbaum}{van Steen and
  Tanenbaum}{2017}]%
        {van_steen_2017}
\bibfield{author}{\bibinfo{person}{Maarten van Steen} {and}
  \bibinfo{person}{Andrew~S. Tanenbaum}.} \bibinfo{year}{2017}\natexlab{}.
\newblock \bibinfo{booktitle}{\emph{Distributed {Systems}}
  (\bibinfo{edition}{3} ed.)}.
\newblock \bibinfo{publisher}{distributed-systems.net}.
\newblock
\urldef\tempurl%
\url{https://www.distributed-systems.net/}
\showURL{%
\tempurl}


\bibitem[\protect\citeauthoryear{Vasconcelos}{Vasconcelos}{2012}]%
        {vasconcelos_2012}
\bibfield{author}{\bibinfo{person}{Vasco~T. Vasconcelos}.}
  \bibinfo{year}{2012}\natexlab{}.
\newblock \showarticletitle{Fundamentals of session types}.
\newblock \bibinfo{journal}{\emph{Information and Computation}}
  \bibinfo{volume}{217} (\bibinfo{year}{2012}), \bibinfo{pages}{52--70}.
\newblock
\showISSN{08905401}
\urldef\tempurl%
\url{https://doi.org/10.1016/j.ic.2012.05.002}
\showDOI{\tempurl}


\bibitem[\protect\citeauthoryear{Wadler}{Wadler}{2014}]%
        {wadler2014}
\bibfield{author}{\bibinfo{person}{Philip Wadler}.}
  \bibinfo{year}{2014}\natexlab{}.
\newblock \showarticletitle{Propositions as sessions}.
\newblock \bibinfo{journal}{\emph{Journal of Functional Programming}}
  \bibinfo{volume}{24}, \bibinfo{number}{2-3} (\bibinfo{year}{2014}),
  \bibinfo{pages}{384--418}.
\newblock
\showISSN{0956-7968}
\urldef\tempurl%
\url{https://doi.org/10.1017/S095679681400001X}
\showDOI{\tempurl}


\end{thebibliography}

\appendix
\section{Coexponentials and Logical equivalences}

We may derive the following logical equivalences about coexponentials, which are dual to similar
laws for the exponentials.
\begin{mathpar}\label{coexp-logic-equiv}
	\exc \exc A \equiv \exc A \and
	\exc A \equiv \exc A \parr \exc A \and
	\exc \bot \equiv \bot \and
	\exc (A \with B) \equiv \exc A \parr \exc B \and
	\exc 0 \equiv \bot \and
	\que \que A \equiv \que A \and
	\que A \equiv \que A \otimes \que A \and
	\que 1 \equiv 1 \and
	\que (A \oplus B) \equiv \que A \otimes \que B \and
	\que \top \equiv 1
\end{mathpar}

\begin{theorem}
	In CLL with \rulenamestyle{Mix} and \rulenamestyle{BiCut}, exponentials and coexponentials coincide up to provability. That is: if we replace $\whynot$ and $\ofc$ in the rules for the exponentials with $\que$ and $\exc$ respectively, the resultant rule is provable using the coexponential rules, and vice versa.
\end{theorem}

\begin{proof}
	We first show that the exponential rules are derivable using coexponential rules under the substitution $\whynot \mapsto \que$.	The weakening rule $\inferrule*[right=$\whynot w$]{\vdash \Gamma}{\vdash \Gamma, \whynot A}$ is mapped to the derivation
	$
		\inferrule*[right=\rulenamestyle{Mix}]{
			\vdash \Gamma \\
			\inferrule*[right=$\que w$]{\strut}{\vdash \que A}
		}{\vdash \Gamma, \que A}
	$. The dereliction rule $\whynot d$ is just $\que d$, and the
	contraction rule $\inferrule*[right=$\whynot c$]{\vdash \Gamma, \whynot A, \whynot A}{\vdash  \Gamma, \whynot A}$ is
	mapped to
	\begin{mathpar}
		\inferrule*[right=\rulenamestyle{BiCut}]{
			\vdash \Gamma, \que A, \que A
			\\
			\inferrule*[right=$\que c$]{
				\inferrule*[right=Ax]{\strut}{\vdash \exc \negg A, \que A}
				\\
				\inferrule*[right=Ax]{\strut}{\vdash \exc \negg A, \que A}
			}{\vdash \exc \negg A, \exc \negg A, \que A}}
			{
				\Gamma, \que A
			}
	\end{mathpar}
  This leaves promotion. The forklores (iii--v) can be generalised to a
	bi-implication
	\begin{align*}
		A_1 \parr \dots \parr A_n \multimap A_1 \otimes \dots \otimes A_n &&
		A_1 \otimes \dots \otimes A_n \multimap A_1 \parr \dots \parr A_n
	\end{align*}
  and hence sequents $\vdash \bigotimes \negg{\Delta}, \bigotimes \Delta$ and $\vdash \bigparr
	\negg{\Delta}, \Delta$ for any $\Delta$. With these in hand, we can interpret the promotion rule
	$\inferrule*[right=$\ofc$]{\vdash \whynot \Gamma, A}{\vdash \whynot \Gamma, \ofc A}$ by the derivation
	\begin{mathpar}
		\inferrule*[right=Cut]{
			\inferrule*{ }{\vdash \que \Gamma, \bigparr \exc \negg{\Gamma}} \\
			\inferrule*[right=$\exc$]{
				\inferrule*[right=Cut]{
					\inferrule*{ }{\vdash \bigotimes \que \Gamma, \bigotimes \exc \negg{\Gamma}} \\
					\mprset{fraction={===}}
					\inferrule*[right=$\bigparr$]{\vdash \que \Gamma, A}{\vdash \bigparr \que \Gamma, A}
				}{
					\vdash \bigotimes \que \Gamma, A
				}
			}{
				\vdash \bigotimes \que \Gamma, \exc A
			}
		}{
			\vdash \que \Gamma, \exc A
		}
	\end{mathpar}

In the opposite direction, we show that the coexponentials rules are derivable using exponentials rules under
the substitution $\que \mapsto \whynot$. As the folklores ensure \rulenamestyle{Mix0}
is derivable in this system, we can interpret the weakening rule $\inferrule*[right=$\que w$]{\strut}{\vdash \que
A}$ by
$
	\inferrule*[right=$\whynot w$]{
		\inferrule*[right=\rulenamestyle{Mix0}]{\strut}{\vdash \cdot}
	}{
		\vdash \whynot A
	}
$. The dereliction rule $\que d$ is simply $\whynot d$, and the contraction rule $\inferrule*[right=$\que c$]{\vdash \Gamma, \que A \\ \vdash \Delta, \que A}{\vdash \Gamma, \Delta, \que A}$ is interpreted by the derivation
\begin{mathpar}
		\inferrule*[right=$\whynot c$]{
			\inferrule*[right=\rulenamestyle{Mix}]{
				\vdash \Gamma, \whynot A	\\
				\vdash \Delta, \whynot A
			}{
				\vdash \Gamma, \Delta, \whynot A, \whynot A
			}
		}{
			\vdash \Gamma, \Delta, \whynot A
		}
\end{mathpar}
Finally, the rule $\inferrule*[right=$\exc$]{\vdash \otimes \que \Gamma, A}{\vdash \otimes \que \Gamma, \exc A}$ is
interpreted in a way similar to promotion, but with the cuts replacing $\otimes$ with $\parr$
happening in the opposite order.
\end{proof}

\section{Translation of CSGV to CSLL: Omitted rules}
Of the functional fragment of CSGV the types are translated to CSLL:
\begin{equation*}
	\begin{aligned}[c]
		\bra*{T \multimap U} \defeq\ &\negg{\bra*{T}} \parr \bra*{U}\\
		\bra*{T \rightarrow U} \defeq\  &\ofc (\negg{\bra*{T}} \parr \bra*{U})
\end{aligned}
\qquad
\begin{aligned}[c]
	\bra*{T + U} \defeq\  &\bra*{T} \oplus \bra*{U} \\
	\bra*{T \otimes U} \defeq\  &\bra*{T} \otimes \bra*{U} \\
	\bra*{\textsf{Unit}} \defeq\  &\ounit 
\end{aligned}
\end{equation*}
Omitted rules of CSGV with their translation into CSLL are shown in \cref{figure:csgv-func,figure:csgv-linear-omitted}. Note that some translations use \cref{lemma:unlimited-positive}.
\begin{figure}
	\small
	\begin{flushleft}
		$
		\Trans{
		\inferhref{LImp-I}{$\multimap I$}{
			\IsProg{\Gamma, \Name{x} : T}{L}{U} 
		}{
			\IsProg{\Gamma}{\Lam{x}{L}}{T \multimap U}
		}
		}[z]
		\defeq\ 
		\inferrule{
			\IsProc{\bra[z]{L}}{\negg{\bra*{\Gamma}}, \Name{x} : \negg{\bra*{T}}, \Name{z} : \bra*{U}}
		}{ 
			\IsProc{\In{z}{x}{\bra[z]{L}}}{\negg{\bra*{\Gamma}}, \Name{z} : \negg{\bra*{T}} \parr \bra*{U}}
		}
		$
	\end{flushleft}
	\begin{flushleft}
		$
		\Trans{
		\inferhref{LImp-E}{$\multimap E$}{
			\IsProg{\Gamma}{M}{T \multimap U} \\
			\IsProg{\Delta}{N}{T}
		}{
			\IsProg{\Gamma,\Delta}{\Jux{M}{N}}{U}
		}
		}[x]
		\defeq\ 
		$
	\end{flushleft}
	\begin{mathpar}
		\inferrule{
			\IsProc{\bra[z]{M}}{\negg{\bra*{\Gamma}}, \Name{z} : \negg{\bra*{T}} \parr \bra*{U}} \\
				\inferrule*[right=$\otimes$]{
					\IsProc{\bra[z']{N}}{\negg{\bra*{\Delta}}, \Name{z'}:\bra*{T}} \\
					\IsProc{\Link{x}{y}}{\Name{x}:\bra*{U}, \Name{y}:\negg{\bra*{U}}}
				}{
					\IsProc{
						\Out{y}{z'}{\Par*{\bra[z']{N}}{\Link{x}{y}}}
					}{
						\negg{\bra*{\Delta}}, \Name{x}:\bra*{U}, \Name{y}:\bra*{T}\otimes\negg{\bra*{U}}
					}
				}
		}{
			\IsProc{
				\New{z}{y}{\Par*{
					\bra[z]{M}
				}{
					\Out{y}{z'}{\Par*{\bra[z']{N}}{\Link{x}{y}}}
				}}
			}{
				\negg{\bra*{\Gamma}}, \negg{\bra*{\Delta}}, \Name{x}:\bra*{U}
			}
		}
	\end{mathpar}
	\begin{flushleft}
		$
		\Trans{
		\inferhref{Imp-I}{$\rightarrow I$}{
			\IsProg{\Name{\vec{v}}:\vec{V}}{L}{T \multimap U}\\
			\text{$\vec{V}$ unlimited}
		}{
			\IsProg{\Name{\vec{v}}:\vec{V}}{L}{T \rightarrow U}
		}
		}[z]
		\defeq\ 
		\inferrule{
		\inferrule*{
			\inferrule*{
				\IsProc{\bra[z]{L}}{\Name{\vec{v}}:\negg{\bra{\vec{V}}}, \Name{z}:\negg{\bra{T}} \parr \bra{U}}
			}{
				\IsProc{\WhyD{\vec{v'}}{\vec{v}}{\bra[z]{L}}}{
					\Name{\vec{v'}}:\whynot \negg{\bra{\vec{V}}}, \Name{z}:\negg{\bra{T}} \parr \bra{U}
				}
			}
		}{
			\IsProc{\OfCP{z}{\vec{v'}}{\WhyD{\vec{v'}}{\vec{v}}{\bra[z]{L}}}}{\Name{\vec{v'}}:\whynot \negg{\bra{\vec{V}}}, \Name{z}:\ofc(\negg{\bra{T}} \parr \bra{U})}
		}
		}{
			\IsProc{\Coderelict{\vec{v}}{\vec{v'}}{\OfCP{z}{\vec{v}}{\WhyD{\vec{v'}}{\vec{v}}{\bra[z]{L}}}}}{\Name{\vec{v}}:\negg{\bra{\vec{V}}}, \Name{z}:\ofc(\negg{\bra{T}} \parr \bra{U})}
		}
		$
	\end{flushleft}
	\begin{flushleft}
		$
		\Trans{
		\inferhref{Imp-E}{$\rightarrow E$}{
			\IsProg{\Gamma}{L}{T \rightarrow U}
		}{
			\IsProg{\Gamma}{L}{T \multimap U}
		}
		}[x]
		\defeq\ 
		\inferrule{
			\IsProc{\bra[z]{L}}{\negg{\bra{\Gamma}}, \Name{z}:\ofc(\negg{\bra{T}} \parr \bra{U})}\\
			\inferrule*{
				\IsProc{\Link{x}{y}}{\Name{x}:\negg{\bra{T}} \parr \bra{U}, \Name{y}:\bra{T} \otimes \negg{\bra{U}}}
			}{
				\IsProc{\WhyD{y'}{y}{\Link{x}{y}}}{
					\Name{x}:\negg{\bra{T}} \parr \bra{U}, \Name{y'}: \whynot (\bra{T} \otimes \negg{\bra{U}})
				}
			}
		}{
			\IsProc{\New{z}{y'}{\Par*{
				\bra[z]{L}
			}{
				\WhyD{y'}{y}{\Link{x}{y}}
			}}}{
				\negg{\bra{\Gamma}}
				\Name{x}:\negg{\bra{T}} \parr \bra{U}
			}
		}
		$
	\end{flushleft}
	\begin{flushleft}
		$
		\Trans{
		\inferhref{W}{$w$}{
			\IsProg{\Gamma}{L}{T} \\
			\text{$U$ unlimited}
		}{
			\IsProg{\Gamma, x : U}{L}{T}
		}
		}[z]
		\defeq\ 
		\inferrule{
			\inferrule*{
				\IsProc{\bra[z]{L}}{\negg{\bra*{\Gamma}}, \Name{z} : \bra*{T}}
			}{
				\IsProc{\WhyW{x'}{\bra[z]{L}}}{\negg{\bra*{\Gamma}}, \Name{x'}:\whynot \negg{\bra*{U}}, \Name{z}:\bra*{T}}
			}
		}{
			\IsProc{\Coderelict{x}{x'}{\WhyW{x}{\bra[z]{L}}}}{\negg{\bra*{\Gamma}}, \Name{x} : \negg{\bra*{U}}, \Name{z} : \bra*{T}}
		}
		$
	\end{flushleft}
	\begin{flushleft}
		$
		\Trans{
		\inferhref{C}{$c$}{
			\IsProg{\Gamma, \Name{x_0} : U, \Name{x_1} : U}{L}{T} \\
			\text{$U$ unlimited}
		}{
			\IsProg{\Gamma, \Name{x} : U}{\Sb{\Sb{L}{x}{x_0}}{x}{x_1}}{T}
		}
		}[z]
		\defeq\ 
		\inferrule{
		\inferrule*{
			\inferrule*{
				\IsProc{\bra[z]{L}}{\negg{\bra*{\Gamma}}, \Name{x_0} : \negg{\bra*{U}}, \Name{x_1} : \negg{\bra*{U}}, \Name{z} : \bra*{T}}
			}{
				\IsProc{\WhyD{x'_0}{x_0}{\WhyD{x'_1}{x_1}{\bra[z]{L}}}}{
					\negg{\bra*{\Gamma}}, \Name{x'_0} : \whynot \negg{\bra*{U}}, \Name{x'_1} : \whynot \negg{\bra*{U}}, \Name{z} : \bra*{T}
				}
			}
		}{
			\IsProc{\WhyC{x'}{x'_0}{x'_1}{
				\WhyD{x'_0}{x_0}{\WhyD{x'_1}{x_1}{\bra[z]{L}}}
			}}{
				\negg{\bra*{\Gamma}}, \Name{x'} : \whynot \negg{\bra*{U}}, \Name{z} : \bra*{T}
			}
		}
		}{
			\IsProc{\Coderelict{x}{x'}{
				\WhyC{x'}{x'_0}{x'_1}{
				\WhyD{x'_0}{x_0}{\WhyD{x'_1}{x_1}{\bra[z]{L}}}
			}
			}}{
				\negg{\bra*{\Gamma}}, \Name{x} : \negg{\bra*{U}}, \Name{z} : \bra*{T}
			}
		}
		$
	\end{flushleft}
	\begin{flushleft}
		$
		\Trans{
		\inferhref{Prod-I}{$\otimes I$}{
			\IsProg{\Gamma}{M}{T} \\
			\IsProg{\Delta}{N}{U}
		}{
			\IsProg{\Gamma,\Delta}{\Pair{M}{N}}{T \otimes U}
		}
		}[z']
		\defeq\ 
		\inferrule{
			\IsProc{\bra[z]{M}}{\negg{\bra*{\Gamma}}, \Name{z} : \bra*{T}} \\
			\IsProc{\bra[z']{N}}{\negg{\bra*{\Delta}}, \Name{z'} : \bra*{U}} \\
		}{
			\IsProc{\Out{z'}{z}{\Par*{\bra[z]{M}}{\bra[z']{N}}}}{\negg{\bra*{\Gamma}}, \negg{\bra*{\Delta}}, \Name{z'} : \bra*{T} \otimes \bra*{U}}
		}
		$
	\end{flushleft}
	\begin{flushleft}
		$
		\Trans{
		\inferhref{Prod-E}{$\otimes E$}{
			\IsProg{\Gamma}{M}{T \otimes U} \\
			\IsProg{\Delta, \Name{x}:T, \Name{y}:U}{N}{V}
		}{
			\IsProg{\Gamma,\Delta}{\LetP{x}{y}{M} \; N}{V}
		}
		}[z]
		\defeq\ 
		$
	\end{flushleft}
	\begin{mathpar}
		\inferrule{
			\IsProc{\bra[z']{M}}{\negg{\bra*{\Gamma}}, \Name{z'}:\bra*{T} \otimes \bra*{U}} \\
			\inferrule*[right=$\parr$]{
				\IsProc{\bra[z]{N}}{\negg{\bra*{\Delta}}, \Name{x}:\negg{\bra*{T}}, \Name{y}:\negg{\bra*{U}},\Name{z}:\bra*{V}}
			}{
				\IsProc{\In{y}{x}{\bra[z]{N}}}{\negg{\bra*{\Delta}}, \Name{y}:\negg{\bra*{T}} \parr \negg{\bra*{U}},\Name{z}:\bra*{V}}
			}
		}{
			\IsProc{\New{z'}{y}{\Par*{
				\bra[z']{M}
			}{
				\In{y}{x}{\bra[z]{N}}
			}}}{
				\negg{\bra*{\Gamma}},\negg{\bra*{\Delta}},\Name{z}:\bra*{V}
			}
		}
	\end{mathpar}
	\begin{flushleft}
		$	
		\Trans{
		\inferhref{Sum-IL}{$+I_L$}{
			\IsProg{\Gamma}{M}{T}
		}{
			\IsProg{\Gamma}{\Left{M}}{T + U}
		}
		}[z]
		\defeq\ 
		\inferrule{
			\IsProc{\bra[z]{M}}{\negg{\bra*{\Gamma}}, \Name{z}:\bra*{T}}
		}{
			\IsProc{\Inl{z}{\bra[z]{M}}}{\negg{\bra*{\Gamma}}, \Name{z}:\bra*{T} \oplus \bra*{U}}
		}
		$
	\end{flushleft}

	\begin{flushleft}
		$
		\Trans{
		\inferhref{Plus-E}{$+E$}{
			\IsProg{\Gamma}{L}{T + U} \\
			\IsProg{\Delta,\Name{x}:T}{M}{V} \\
			\IsProg{\Delta,\Name{x}:U}{N}{V}
		}{
			\IsProg{\Gamma,\Delta}{\Match{L}{x}{M}{N}}{V}
		}
		}[z]
		\defeq\ 
		$
	\end{flushleft}
	\begin{mathpar}
		\inferrule{
			\IsProc{\bra[y]{L}}{\negg{\bra*{\Gamma}}, \Name{y}:\bra*{T} \oplus \bra*{U}}
			\\
			\inferrule*[right=$\with$]{
				\IsProc{\bra[z]{M}}{\negg{\bra*{\Delta}}, \Name{x}:\negg{\bra*{T}}, \Name{z}:\bra*{V}}
				\\
				\IsProc{\bra[z]{N}}{\negg{\bra*{\Delta}}, \Name{x}:\negg{\bra*{U}}, \Name{z}:\bra*{V}}
			}{
				\IsProc{\Case{x}{\bra[z]{M}}{\bra[z]{N}}}{\negg{\bra*{\Delta}}, \Name{x}:\negg{\bra*{T}} \with \negg{\bra*{U}}, \Name{z}:V}
			}
		}{
			\IsProc{\New{x}{y}{\Par*{
				\bra[y]{L}
			}{
				\Case{x}{\bra[z]{M}}{\bra[z]{N}}
			}}}{\negg{\bra*{\Gamma}}, \negg{\bra*{\Delta}}, \Name{z}:\bra*{V}}
		}
\end{mathpar}
\caption{Translation from CSGV to CSLL, functional fragments}
\label{figure:csgv-func}
\end{figure}

\begin{figure}
	\small
	\begin{flushleft}
		$
		\Trans{
		\inferH{SelectL}{
			\IsProg{\Gamma}{M}{T_S \oplus U_S}
		}{
			\IsProg{\Gamma}{\SelectL{M}}{T_S}
		}
		}[y]
		\defeq\ 
		\inferrule{
			\IsProc{\bra[z]{M}}{\negg{\bra*{\Gamma}}, \Name{z} : \bra*{T_S} \with \bra*{U_S}}\\
			\inferrule*[right=$\oplus$]{
				\IsProc{\Link{x}{y}}{\Name{x}: \negg{\bra*{T_S}}, \Name{y}: \bra*{T_S}}
			}{
				\IsProc{\Inl{x}{\Link{x}{y}}}{\Name{x}: \negg{\bra*{T_S}} \oplus \negg{\bra*{U_S}}, \Name{y}: \bra*{T_S}}
			}
		}{
			\IsProc{\New{x}{z}{\Par*{
				\bra[z]{M}
			}{
				\Inl{x}{\Link{x}{y}}
			}}}{\negg{\bra*{\Gamma}}, \Name{y} : \bra*{T_S}}
		}
		$
	\end{flushleft}
	\begin{flushleft}
		$
		\Trans{
		\inferH{Case}{
			\IsProg{\Gamma}{L}{T_S \with U_S}\\
			\IsProg{\Delta, \Name{x} : T_S}{M}{V}\\
			\IsProg{\Delta, \Name{x} : U_S}{N}{V}
		}{
			\IsProg{\Gamma,\Delta}{\SCase{L}{x}{M}{N}}{V}
		}
		}[z]
		\defeq\ 
		$
	\end{flushleft}
	\begin{mathpar}
		\inferrule*[right=Cut]{
			\IsProc{\bra[y]{L}}{\negg{\bra*{\Gamma}}, \Name{y} : \bra*{T_S} \oplus \bra*{U_S}}\\
			\inferrule*[right=$\with$]{
				\IsProc{\bra[z]{M}}{\negg{\bra*{\Delta}}, \Name{x}:\negg{\bra*{T_S}}, \Name{z} : \bra*{V}}\\
				\IsProc{\bra[z]{N}}{\negg{\bra*{\Delta}}, \Name{x}:\negg{\bra*{U_S}}, \Name{z} : \bra*{V}}
			}{
				\IsProc{\Case{x}{\bra[z]{M}}{\bra[z]{N}}}{\negg{\bra*{\Delta}}, \Name{x}:\negg{\bra*{T_S}} \with \negg{\bra*{U_S}}, \Name{z} : \bra*{V}}
			}
		}{
			\IsProc{\New{x}{y}{\Par*{
				\bra[y]{L}
			}{
				\Case{x}{\bra[z]{M}}{\bra[z]{N}}
			}}}{\negg{\bra*{\Gamma}},\negg{\bra*{\Delta}},\Name{z}:\bra*{V}}
		}
	\end{mathpar}
	\caption{Translation from CSGV to CSLL, omitted rules of linear fragments}
	\label{figure:csgv-linear-omitted}
\end{figure}

\section{CSLL: Metatheoretic Proofs}

\begin{lemma}
	If $\Equiv{P}{Q}$, then $\IsProc{P}{\Ctxs{G}}$ if and only if $\IsProc{Q}{\Ctxs{G}}$.
\end{lemma}
\begin{proof}
	By induction on $\Equiv{P}{Q}$. We prove one direction, the other one being entirely analogous. Moreover, the congruence cases are trivial. $\Equiv{\Par{P}{\Stop}}{\Proc{P}}$, commutativity, and associativity follow from the structure of hyperenvironments. Link-commutativity follows from the involutive property of $\negg{(-)}$.
	\begin{indproof}

		\indcase{\rulenamestyle{Res-Par}}

		Then $\Proc{P} = \New{x}{y}{\Par*{R}{S}}$ and $\Proc{Q} = \Par{R}{\New{x}{y}{S}}$ where
		$\Name{x}, \Name{y} \notin \fn{\Proc{R}}$. We must then have that $\IsProc{R}{\Ctxs{H}}$ where
		$\Name{x}, \Name{y} \not\in \Ctxs{H}$ (using \cref{lemma:fn-match}) and $\IsProc{S}{\Ctxs{I} \mid \Gamma, \Name{x} : A \mid
		\Delta, \Name{y}:\negg A}$, where $\Ctxs{G} = \Ctxs{H} \mid \Ctxs{I} \mid \Gamma, \Delta$.
		Hence, we can derive that $\Proc{Q} = \IsProc{\Par{R}{\New{x}{y}{S}}}{\Ctxs{G}}$.

		\indcase{\rulenamestyle{Res-Res}}

		Then $\Proc{P} = \New{x}{y}{\New{z}{w}{R}}$ and $\Proc{Q} = \New{z}{w}{\New{x}{y}{R}}$ for some
		$\Proc{R}$. We must invert $\IsProc{P}{\Ctxs{G}}$. This generates many cases: for example, it
		could be that $\IsProc{R}{\Ctxs{G}' \mid \Gamma, \Name{x} : A, \Name{z} : B \mid \Delta,
		\Name{y} : \negg A \mid \Sigma, \Name{w} : \negg B}$ where $\Ctxs{G} = \Ctxs{G}' \mid \Gamma,
		\Delta, \Sigma$, whence $\Proc{Q} = \IsProc{\New{z}{w}{\New{x}{y}{\Proc{R}}}}{\Ctxs{G}}$. The
		other cases are similar.

		\indcase{\rulenamestyle{Res-Pre}}
We show the case for $\New{x}{y}{\Out{z}{w}{P}} \equiv \Out{z}{w}{\New{x}{y}{P}}$, with that of other prefixes being similar. We have
\begin{mathpar}
	\inferrule*{
		\inferrule*{
			\IsProc{P}{
				\Ctxs{G} \mid \Sigma, \Name{y}: \negg C \mid \Gamma, \Name{z}:A, \Name{x}:C \mid \Delta, \Name{w}:B}
		}{
			\IsProc{\Out{z}{w}{P}}{
				\Ctxs{G} \mid \Sigma, \Name{y}: \negg C \mid \Gamma, \Name{z}:B \otimes A, \Name{x}:C, \Delta
			}
		}
	}{
		\IsProc{\New{x}{y}{\Out{z}{w}{P}}}{
			\Ctxs{G} \mid \Sigma, \Gamma, \Name{z}:B \otimes A, \Delta
		}
	}
\end{mathpar}
and therefore
\begin{mathpar}
	\inferrule*{
		\inferrule*{
			\IsProc{P}{
				\Ctxs{G} \mid \Sigma, \Name{y}: \negg C \mid \Gamma, \Name{z}:A, \Name{x}:C \mid \Delta, \Name{w}:B}
		}{
			\IsProc{\New{x}{y}{P}}{
				\Ctxs{G} \mid \Sigma, \Gamma, \Name{z}:A \mid \Delta, \Name{w}:B
			}
		}
	}{
		\IsProc{\Out{z}{w}{\New{x}{y}{P}}}{
			\Ctxs{G} \mid \Sigma, \Gamma, \Name{z}:B \otimes A, \Delta
		}
	}
\end{mathpar}
the other case has $\Name{x,y,z,w}$ in separate environments and are simpler.
\indcase{\rulenamestyle{Pre-Par}}
We show the case for $\Out{x}{y}{\Par*{P}{Q}} \equiv \Par{P}{\Out{x}{y}{Q}}$, with that of other prefixes being similar. We have
\begin{mathpar}
	\inferrule*{
		\inferrule*{
			\IsProc{P}{\Ctxs{G}}\\
			\IsProc{Q}{\Ctxs{H} \mid \Gamma, \Name{x}:A \mid \Delta, \Name{y}:B}
		}{
			\IsProc{\Par{P}{Q}}{\Ctxs{G} \mid \Ctxs{H} \mid \Gamma, \Name{x}:A \mid \Delta, \Name{y}:B}
		}
	}{
			\IsProc{\Out{x}{y}{\Par*{P}{Q}}}{\Ctxs{G} \mid \Ctxs{H} \mid \Gamma, \Delta, \Name{x}: B\otimes A}
	}
\end{mathpar}
and therefore
\begin{mathpar}
	\inferrule*{
		\IsProc{P}{\Ctxs{G}}\\
		\inferrule*{
			\IsProc{Q}{\Ctxs{H} \mid \Gamma, \Name{x}:A \mid \Delta, \Name{y}:B}
		}{
			\IsProc{\Out{x}{y}{Q}}{\Ctxs{H} \mid \Gamma, \Delta, \Name{x}: B\otimes A}
		}
	}{
		\IsProc{\Par{P}{\Out{x}{y}{Q}}}{
			\Ctxs{G} \mid \Ctxs{H} \mid \Gamma, \Delta, \Name{x}: B\otimes A
		}
	}
\end{mathpar}
\end{indproof}
\end{proof}

\begin{lemma} \label{lem:equiv-preserves-canon}
	If $\Proc{P} \equiv \Proc{Q}$, then $\Proc{P}$ is canonical if and only if $\Proc{Q}$ is canonical.
\end{lemma}
\begin{proof}
	Straightforward by induction on $\Proc{P} \equiv \Proc{Q}$.
\end{proof}

\begin{lemma}[Separation]
	\label{lem:separation}
	If $\IsProc{T}{\Gamma_0 \mid \dots \mid \Gamma_{n-1}}$, then there exist $\IsProc{T_i}{\Gamma_i}$ for $0 \leq i < n$ such that $\Proc{T} \equiv \Proc{T_0 \mid \dots \mid T_{n-1}}$. Moreover, if $\Proc{T}$ is canonical, then every $\Proc{T_i}$ is canonical.
\end{lemma}
\begin{proof}
	We prove the first claim by induction on $\IsProc{T}{\Gamma_0 \mid \dots \mid \Gamma_{n-1}}$. We show only the following cases; all other cases are either trivial or similar.
	\begin{indproof}
		\indcase{\rulenamestyle{HMix2}} Then $\Proc{T} = \Par{P}{Q}$, and after appropriately reordering the
		hyperenvironment we have $\IsProc{P}{\Gamma_0 \mid \dots \mid \Gamma_{m-1}}$ and
		$\IsProc{Q}{\Gamma_{m} \mid \dots \mid \Gamma_{n-1}}$ with $m \le n$. By the IH we have
		$\IsProc{T_i}{\Gamma_i}$ for $0 \leq i < n$, with $\Proc{P} \equiv \Proc{T_0} \mid \dots \mid
		\Proc{T_{m-1}}$, and $\Proc{Q} \equiv \Proc{T_m} \mid \dots \mid \Proc{T_{n-1}}$. We then have
		$\Proc{\Par{P}{Q}} \equiv \Proc{T_0} \mid \dots \mid \Proc{T_{n-1}} \equiv T$, as $\equiv$ is a
		congruence.

		\indcase{\rulenamestyle{Cut}}  Then $\Proc{T} = \New{x}{y}{P}$, and after appropriately reordering the
		hyperenvironment we have 
		\[
		\IsProc{P}{\Gamma_0 \mid \dots \mid \Gamma_{n-2} \mid \Delta_0, \Name{x} : A \mid \Delta_1, \Name{y} : \negg A}
		\]
		 where $\Gamma_{n-1} = \Delta_0,\Delta_1$. By the IH we have $\IsProc{P_i}{\Gamma_i}$ for $0 \leq i < n-1$, $\IsProc{P_{n-1}}{\Delta_0 ,
		\Name{x} : A}$, and $\IsProc{P_n}{\Delta_1, \Name{y} : \negg A}$, with $\Equiv{P}{P_0 \mid \dots
		\mid P_n}$. The result follows, as $\IsProc{\New*{x}{y}{\Par{P_{n-1}}{P_n}}}{\Gamma_{n-1}}$. and by (Res-Par)
		\[
			\New{x}{y}{P} \equiv \New*{x}{y}{P_0 \mid \dots \mid P_{n-1} \mid P_n}
										\equiv \Proc{P_0 \mid \dots \mid \New*{x}{y}{\Par{P_{n-1}}{P_n}}}
		\]

		\indcase{\rulenamestyle{Tensor}}
		Then $\Proc{T} = \Out{x}{y}{P}$, and after appropriately reordering the hyperenvironment we have 
		\[
			\IsProc{P}{\Gamma_0 \mid \dots \Gamma_{n-2} \mid \Delta_0, \Name{x}:A \mid \Delta_1, \Name{y}:B}
		\] where $\Gamma_{n-1} = \Delta_0,\Delta_1$. By the IH we have $\IsProc{P_i}{\Gamma_i}$ for $0 \le i < n-1$, $\IsProc{P_{n-1}}{\Delta_0, \Name{x}:A}$ and $\IsProc{P_n}{\Delta_1, \Name{y}:B}$ with $\Equiv{P}{P_0 \mid \dots \mid P_n}$. The result follows, as $\IsProc{\Out{x}{y}{\Par*{P_{n-1}}{P_n}}}{\Gamma_{n-1}, \Name{x}: B \otimes A}$, and by ()
	\end{indproof}
	The second claim follows by \cref{lem:equiv-preserves-canon}, and the fact subterms of canonical terms are canonical.
\end{proof}

\begin{lemma}[Local Progress]
	\label{lem:local-progress}
	If $\IsProc{P}{\Gamma, \Name{x} : A}$ and $\IsProc{Q}{\Delta, \Name{y} : \negg A}$ and both
	$\Proc{P}$ and $\Proc{Q}$ are canonical, then there exists an $\Proc{R}$ such that
	$\red{\New{x}{y}{\Par*{P}{Q}}}{R}$.
\end{lemma}
\begin{proof}
	By induction on $\Proc{P}$ and $\Proc{Q}$. Note the two are symmetric which we will exploit to omit some cases. The type judgment implies neither $\Proc{P}$ nor $\Proc{Q}$ can be $\Stop$. They cannot be of the form $\New{x}{y}{S}$ either, for they would not be canonical.
	\begin{itemize}
		\item If $\Proc{P} = \Par{P_0}{P_1}$, then it must be that $\Proc{P_1} = \Stop$ without loss of
		generality. We have that $\Equiv{\Par{P_0}{\Stop}}{P_0}$ by \rulenamestyle{Par-Unit}. Apply induction
		hypothesis on $\Proc{P_0}$ we get $\red{\New{x}{y}{\Par*{P_0}{Q}}}{R}$. Use \rulenamestyle{Eq} we have
		$\red{\New{x}{y}{\Par*{P}{Q}}}{R}$. Similar when $\Proc{Q} = \Par{Q_0}{Q_1}$.
		\item If $\Proc{P} = \Link{a}{b}$, it must be that $\Name{x} = \Name{b}$, so we can reduce by \rulenamestyle{Link}. Similar for $\Proc{Q}$.
		\item The remaining scenarios are where $\Proc{P} = \Pre{z}[P']$, or $\Proc{P} = \Case{z}{P_0}{P_1}$ or $\Proc{P} = \OfCP{z}{\vec{z'}}{P'}$, and similarly $\Proc{Q} = \Pre{w}[Q']$, or $\Proc{Q} = \Case{w}{Q_0}{Q_1}$ or $\Proc{Q} = \OfCP{w}{\vec{w'}}{Q'}$. If $\Name{z} = \Name{x}$ and $\Name{w} = \Name{y}$, then one of the reaction axioms apply; otherwise we can assume WLOG that $\Name{z} \neq \Name{x}$ (and of course $\Name{z} \neq \Name{y}$), and take cases of $\Proc{P}$.
		\begin{itemize}
			\item If $\Proc{P} = \Pre{z}[P']$, we have that $\New{x}{y}{\Par*{\Pre{z}[P']}{Q}} \equiv \Pre{z}[\New{x}{y}{\Par*{P'}{Q}}]$ by \rulenamestyle{Pre-Par} and \rulenamestyle{Res-Pre}. By induction hypothesis we have $\red{\New{x}{y}{\Par*{P'}{Q}}}{R}$, we therefore have $\red{\New{x}{y}{\Par*{P}{Q}}}{\Pre{z}[R]}$ by \rulenamestyle{Pre} and \rulenamestyle{Eq}.
			\item If $\Proc{P} = \Case{z}{P_0}{P_1}$, the commuting conversion \rulenamestyle{Case-Comm} applies.
			\item If $\Proc{P} = \OfCP{z}{\vec{z'}}{P'}$ where $\Name{x} \in \Name{\vec{z'}}$. We know $\Name{x} : A = \whynot B$ for some $B$ and thus $\Name{y}:\negg A = \ofc \negg B$. We check if $\Name{w} = \Name{y}$. If so, we have $\Proc{Q} = \OfCP{w}{\vec{w'}}{Q'}$ and thus \rulenamestyle{OfCourse-Comm} applies; otherwise, we know that $\Proc{Q}$ cannot be of the form $\OfCP{w}{\vec{w'}}{Q'}$ where $\Name{y} \in \Name{\vec{w'}}$ (because $\Name{y}:\ofc \negg B$ breaks \rulenamestyle{OfCourse} requirement that $\Name{\vec{w'}}:\whynot \vec{B}$), which means $\Proc{Q} = \Pre{w}{Q'}$  or $\Proc{Q} = \Case{w}{Q_0}{Q_1}$ and can be handled similarly as the previous two cases.
		\end{itemize}
	\end{itemize}
\end{proof}

\begin{theorem}[Progress]
				If $\IsProc{R}{\Ctxs{G}}$ then either $\Proc{R}$ is canonical, or there exists $\Proc{R'}$ such that	$\red{R}{R'}$.
\end{theorem}
\begin{proof}
	By induction on $\IsProc{R}{\Ctxs{G}}$.
	\begin{itemize}
		\item If $\Proc{R} = \Stop$ or $\Link{x}{y}$, then it is canonical.

		\item Suppose $\Proc{R} = \Par{P}{Q}$. If both $\Proc{P}$ and $\Proc{Q}$ are canonical, then so
		is $\Proc{R}$. Otherwise, if $P$ is not canonical, then by the IH we have $\red{P}{P'}$ for some
		$\Proc{R'}$, and thus $\red{\Par{P}{Q}}{\Par{P'}{Q}}$ by \rulenamestyle{ParL}. Similarly for
		$\Proc{Q}$.

		\item Suppose $\Proc{R} = \Pre{y}[P]$. If $\Proc{P}$ is canonical then so is $\Proc{R}$. Otherwise $\Proc{P}$ is not canonical, and by the IH $\red{P}{P'}$, and thus $\red{\Pre{y}[P]}{\Pre{y}[P']}$ for some $\Proc{P'}$ by \rulenamestyle{Pre}.

		\item Suppose $\Proc{R} = \Case{y}{P}{Q}$ or $\Proc{R} = \OfCP{y}{\vec{w}}{P}$, then it is canonical.

		\item Suppose $\Proc{R} = \New{x}{y}{P}$, with $\IsProc{P}{\Ctxs{G} \mid \Gamma, \Name{x} : A
		\mid \Delta, \Name{y} : \negg A}$. If $P$ is not canonical then by the IH we have $\red{P}{P'}$
		for some $\Proc{P'}$, and thus $\red{\New{x}{y}{P}}{\New{x}{y}{P'}}$ by \rulenamestyle{Res}. If
		$\Proc{P}$ is canonical, by \cref{lem:separation} we have that $\Proc{P} \equiv \Proc{P_0} \mid
		\dots \mid \Proc{P_{n}} \mid \Proc{P_{n+1}}$ where $\IsProc{P_n}{\Gamma, \Name{x}:A}$ and
		$\IsProc{P_{n+1}}{\Delta,\Name{y}:\negg A}$. Note that both $\Proc{P_n}$ and $\Proc{P_{n+1}}$
		are canonical. Hence we have $\Proc{R} \equiv \New*{x}{y}{P_0 \mid \dots \mid P_{n+1}}$. By
		\rulenamestyle{Res-Par} and \cref{lemma:fn-match} we obtain $\Proc{R} \equiv
		\Proc{P_0} \mid \dots \mid \Proc{P_{n-1}} \mid \New*{x}{y}{\Par{P_n}{P_{n+1}}}$. Local progress
		(\cref{lem:local-progress}) yields $\red{\New*{x}{y}{\Par{P_n}{P_{n+1}}}}{R'}$, which gives
		$\red{R}{P_0 \mid \dots \mid P_{n-1} \mid R'}$ by \rulenamestyle{ParL}.
	\end{itemize}
\end{proof}

\begin{definition}[Free Names] \label{def:fv}
	The free names $\fn{\Proc{P}}$ of a process $\Proc{P}$ is inductively defined as follows.
	\begin{align*}
 	\fn{\Stop} &= \emptyset \\
 	\fn{\Par{P}{Q}} &= \fn{\Proc{P}} \cup \fn{\Proc{Q}} \\
 	\fn{\New{x}{y}{P}} &= \fn{\Proc{P}} \setminus \{\Name{x},\Name{y}\} \\
 	\fn{\Link{x}{y}} &= \{x,y\} \\
 	\fn{\In{y}{x}{P}} = \fn{\Out{y}{x}{P}} &= \fn{\Proc{P}} \setminus \{\Name{x}\} \\
 	\fn{\Inl{x}{P}} = \fn{\Inr{x}{P}} &= \fn{\Case{x}{P}{Q}} = \fn{\Proc{P}} \\
	\fn{\In{x}{}{P}} = \fn{\Out{x}{}{P}} &= \fn{\Proc{P}} \cup \{\Name{x}\} \\
	\fn{\QueW{x}{P}} = \fn{\Proc{P}} \cup \{ \Name{x}\} \\
	\fn{\QueA{x}{x'}{P}} &= \fn{\Proc{P}} \setminus \{\Name{x'} \} \cup \{\Name{x}\}\\
	\fn{\ExcP{y}{P}{z}{w}{z'}{w'}{y'}{Q}} &= \fn{\Proc{P}} \cup \{\Name{y}\}
	\end{align*}
	The free names $\fn{\Gamma}$ of an environment $\Gamma$ is defined to be the names in $\Gamma$. The free names $\fn{\mathcal{G}}$ of an hyperenvironment $\mathcal{G}$ is defined to be the union of the free names of each environment in $\mathcal{G}$. Note that we stipulated names in each environment must not overlap.
\end{definition}

\begin{lemma}\label{lemma:fn-match}
	If $\IsProc{P}{\mathcal{G}}$, then $\fn{\Proc{P}} = \fn{\mathcal{G}}$.
\end{lemma}
\begin{proof}
	Straightforward by induction on $\IsProc{P}{\mathcal{G}}$.
\end{proof}

\begin{lemma}\label{lemma:renaming-preserves-types}
	If $\IsProc{P}{\Ctxs{G} \mid \Gamma, \Name{y}:A}$ and $\Name{x} \notin \Ctxs{G},\Gamma$, then $\IsProc{\Sb{P}{x}{y}}{\Ctxs{G} \mid \Gamma, \Name{x}:A}$.
\end{lemma}
\begin{proof}
	Straightforward by induction on $\IsProc{P}{\Ctxs{G} \mid \Gamma, \Name{y}:A}$.
\end{proof}

\begin{theorem}[Preservation]
	If $\IsProc{P}{\Ctxs{G}}$ and $\red{P}{Q}$, then $\IsProc{Q}{\Ctxs{G}}$.
\end{theorem}
\begin{proof}
By induction on $\red{P}{Q}$. We show the nontrivial cases of top-level cuts, and the commuting conversions.
\begin{indproof}

\indcase{\rulenamestyle{Eq}}
Suppose $\Proc{P} \equiv \red{P'}{Q'} \equiv \Proc{Q}$. Then the result follows by the IH
and two applications of \cref{lem:equiv-preserves-types}.

\indcase{\rulenamestyle{Case-Comm}} The redex is $\New{x}{y}{\Par*{\Case{z}{P_0}{P_1}}{Q}}$ and typed.
\begin{mathpar}
  \inferrule*{
    \inferrule*{
      \IsProc{P_0}{\Gamma,\Name{x}:C,\Name{z}:A}
      \\
      \IsProc{P_1}{\Gamma,\Name{x}:C, \Name{z}:B}
    }{
      \IsProc{\Case{z}{P_0}{P_1}}{\Gamma, \Name{x}:C, \Name{z}:A \with B}
    }
\\  
    \IsProc{Q}{\Delta, \Name{y} : \negg C}
  }{
    \IsProc{\New{x}{y}{\Par*{\Case{z}{P_0}{P_1}}{Q}}}{\Gamma,\Delta,\Name{z}:A \with B}
  }
\end{mathpar}
and therefore
\begin{mathpar}
  \inferrule*{
    \IsProc{\New{x}{y}{\Par*{P_0}{Q}}}{\Gamma,\Delta,\Name{z}:A}
    \\
    \IsProc{\New{x}{y}{\Par*{P_1}{Q}}}{\Gamma,\Delta,\Name{z}:B}
  }{
    \IsProc{\Case{z}{\New{x}{y}{\Par*{P_0}{Q}}}{\New{x}{y}{\Par*{P_1}{Q}}}}{\Gamma,\Delta,\Name{z}:A \with B}
  }
\end{mathpar}

\indcase{\rulenamestyle{OfCourse-Comm}}
The redex is $ \New{x}{y}{\Par*{
  \OfCP{z}{x\vec{w}}{P}
}{
  \OfCP{y}{\vec{v}}{Q}
}}$ and typed.
\begin{mathpar}
  \inferrule*{
    \inferrule*{
      \IsProc{P}{\Name{\vec{w}}:\whynot \vec{B}, \Name{z}:A, \Name{x}:\whynot C}
    }{
      \IsProc{\OfCP{z}{x\vec{w}}{P}}{
        \Name{\vec{w}}:\whynot \vec{B}, \Name{z}:\ofc A, \Name{x}:\whynot C
      }
    }
    \\
    \inferrule*{
      \IsProc{Q}{\Name{\vec{v}}:\whynot \vec{D},\Name{y}:\ofc \negg{C}}
    }{
      \IsProc{\OfCP{y}{\vec{v}}{Q}}{
        \Name{\vec{v}}:\whynot \vec{D},\Name{y}:\ofc \negg{C}
      }
    }
  }{
    \IsProc{
      \New{x}{y}{\Par*{
        \OfCP{z}{x\vec{w}}{P}
      }{
        \OfCP{y}{\vec{v}}{Q}
      }}}{
      \Name{\vec{w}}:\whynot \vec{B},
      \Name{z}:\whynot A,
      \Name{\vec{v}}:\whynot \vec{D}
    }
  }
\end{mathpar}
and therefore
\begin{mathpar}
  \inferrule*{
    \inferrule*{
      \IsProc{P}{\Name{\vec{w}}:\whynot \vec{B}, \Name{z}:A, \Name{x}:\whynot C}\\
      \IsProc{\OfCP{y}{\vec{v}}{Q}}{\Name{\vec{v}}:\whynot \vec{D},\Name{y}:\ofc \negg{C}}
    }{
      \IsProc{\New{x}{y}{\Par*{P}{\OfCP{y}{\vec{v}}{Q}}}}{
        \Name{\vec{w}}:\whynot \vec{B},
        \Name{z}:A,
        \Name{\vec{v}}:\whynot \vec{D}
      }
    }
  }{
    \IsProc{\OfCP{z}{\vec{v}\vec{w}}{\New{x}{y}{\Par*{P}{\OfCP{y}{\vec{v}}{Q}}}}}{
      \Name{\vec{w}}:\whynot \vec{B},
      \Name{z}:\ofc A,
      \Name{\vec{v}}:\whynot \vec{D}
    }
  }
\end{mathpar}

\indcase{\rulenamestyle{Link}}
Then the redex is $\New*{x}{y}{\Par{\Link{z}{x}}{P}}$ and the last steps of the typing derivation must have been
\begin{mathpar}
    \inferrule*{
      \IsProc{\Link{z}{x}}{\Name{z} : \negg A, \Name{x} : A} \\
      \IsProc{P}{\Ctxs{G} \mid \Gamma, \Name{y} : \negg A}
    }{
      \inferrule*{
        \IsProc{\Par{\Link{z}{x}}{P}}{\Name{z} : \negg A, \Name{x} : A \mid \Ctxs{G} \mid \Gamma, \Name{y} : \negg A}
      }{
        \IsProc{\New*{x}{y}{\Par{\Link{z}{x}}{P}}}{\Ctxs{G} \mid \Gamma, \Name{z} : \negg A}
      }
    }
\end{mathpar}
and therefore $\IsProc{\Sb{P}{z}{y}}{\Ctxs{G} \mid \Gamma, \Name{z}:\negg A}$ by \cref{lemma:renaming-preserves-types} because $\Name{z} \notin \fn{P}$. (otherwise the redex would not be well-typed)

\indcase{\rulenamestyle{One-Bot}}
Then the redex is $\New*{x}{y}{\Par{\In{x}{}{P}}{\Out{y}{}{Q}}}$ and the last steps of the typing derivation must have been
\begin{mathpar}
    \inferrule*{
        \inferrule*{
            \IsProc{P}{\Ctxs{G} \mid \Gamma}
        }{
            \IsProc{\In{x}{}{P}}{\Ctxs{G} \mid \Gamma, \Name{x} : \punit}
        } \\
        \inferrule*{
            \IsProc{Q}{\Ctxs{H}}
        }{
            \IsProc{\Out{y}{}{Q}}{\Ctxs{H} \mid \Name{y}:\ounit}
        }
    }{
        \IsProc{\New*{x}{y}{\Par{\In{x}{}{P}}{\Out{y}{}{Q}}}}{\Ctxs{G} \mid \Ctxs{H} \mid \Gamma}
    }
\end{mathpar}
Hence, we have
\begin{mathpar}
    \inferrule*{
      \IsProc{P}{\Ctxs{G} \mid \Gamma} \\ 
      \IsProc{Q}{\Ctxs{H}}
    }{
      \IsProc{\Par{P}{Q}}{\Ctxs{G} \mid \Ctxs{H} \mid \Gamma}
    }
\end{mathpar}

\indcase{\rulenamestyle{Tensor-Par}}
Then the redex is $\New{x}{y}{\Par{\Out{x}{z}{P}}{\In{y}{w}{Q}}}$, and the last steps of the typing derivation must have been 
\begin{mathpar}
    \inferrule*{
        \inferrule*{
          \IsProc{P}{\Ctxs{G} \mid \Gamma, \Name{z} : A \mid \Delta, \Name{x} : B}
        }{
          \IsProc{\Out{x}{z}{P}}{\Ctxs{G} \mid \Gamma, \Delta, \Name{x} : A \otimes B}
        }
        \\
        \inferrule*{
          \IsProc{Q}{\Ctxs{H} \mid \Sigma, \Name{w} : \negg A, \Name{y} : \negg B}
        }{
          \IsProc{\In{y}{w}{Q}}{\Ctxs{H} \mid \Sigma, \Name{y} : \negg A \parr \negg B}
        }
    }{
      \IsProc{\New*{x}{y}{\Par{\Out{x}{z}{P}}{\In{y}{w}{Q}}}}{\Ctxs{G} \mid \Ctxs{H} \mid \Gamma, \Delta, \Sigma}
    }
\end{mathpar}
so that $\Ctxs{G} = \Gamma, \Delta, \Sigma$. Therefore, we can infer that
\begin{mathpar}
    \inferrule*{
        \IsProc{P}{\Ctxs{G} \mid \Gamma, \Name{z} : A \mid \Delta, \Name{x} : B} \\
        \IsProc{Q}{\Ctxs{H} \mid \Sigma, \Name{w} : \negg A, \Name{y}: \negg B}
    }{
      \IsProc{\New{x}{y}{\New{z}{w}{\Par*{P}{Q}}}}{\Ctxs{G} \mid \Ctxs{H} \mid \Gamma,\Delta,\Sigma}
    }
\end{mathpar}

\indcase{\rulenamestyle{PlusL-With}}
Then the redex is $\New{x}{y}{\Par*{\Inl{x}{P}}{\Case{y}{Q_l}{Q_r}}}$, and the last steps of the typing derivation must have been
\begin{mathpar}
    \inferrule*{
        \inferrule*{
            \IsProc{P}{\Ctxs{G} \mid \Gamma,\Name{x}:A}
            }{
            \IsProc{\Inl{x}{P}}{\Ctxs{G} \mid \Gamma,\Name{x}: A \oplus B}
            }
        \\
        \inferrule*{
            \IsProc{Q_l}{\Delta,\Name{y}:\negg A}
            \\
            \IsProc{Q_r}{\Delta,\Name{y}:\negg B}
        }{
            \IsProc{\Case{y}{Q_l}{Q_r}}{\Delta, \Name{y}:\negg A \with \negg B}
        }
    }{
        \IsProc{\New{x}{y}{\Par*{\Inl{x}{P}}{\Case{y}{Q_l}{Q_r}}}}{\Ctxs{G} \mid \Gamma,\Delta}
    }	
\end{mathpar}
Hence,
\begin{mathpar}
    \inferrule*{
        \IsProc{P}{\Ctxs{G} \mid \Gamma,\Name{x}:A} 
        \\
        \IsProc{Q_l}{\Delta,\Name{y}:\negg A}
    }{
        \IsProc{\New{x}{y}{\Par*{P}{Q_l}}}{\Ctxs{G} \mid \Gamma,\Delta}
    }
\end{mathpar}

\indcase{\rulenamestyle{Claro-QueW}} 
This is the case of an empty client pool. The redex must be 
\[
  \New*{x}{y}{\Par{\QueW{x}{S}}{\ExcP{y}{P}{i}{f}{z}{z'}{y'}{Q}}}
\]
and the last steps in the typing
derivation must have been
\begin{mathpar}
  \inferrule*[right=Cut]{
    \IsProc{\mathcal{D}}{\Ctxs{G} \mid \Name{x}:\que A} \\
    \inferrule*[right=Claro]{
      \IsProc{P}{\Ctxs{H} \mid \Delta,\Name{i}:B \mid \Sigma, \Name{f} :\negg B} \\
      \IsProc{Q}{\Name{z} :\negg B, \Name{z'}:B, \Name{y'} : \negg A}
    }{
      \IsProc{\ExcP{y}{P}{i}{f}{z}{z'}{y'}{Q}}{\Ctxs{H} \mid \Delta, \Name{y} : \exc \negg A, \Sigma}
    }
  }{
    \IsProc{\New{x}{y}{\Par*{\mathcal{D}}{\ExcP{y}{P}{i}{f}{z}{z'}{y'}{Q}}}}{\Ctxs{G} \mid \Ctxs{H} \mid \Delta,\Sigma}
  }
\end{mathpar}
where $\mathcal{D} \defeq 
\inferrule{
  \IsProc{S}{\Ctxs{G}}
}{
  \IsProc{\QueW{x}{S}}{\Ctxs{G} \mid \Name{x} : \exc A}
}
$. Hence,
\begin{mathpar}
  \inferrule*{
  \IsProc{S}{\Ctxs{G}} \\
  \inferrule*{
    \IsProc{P}{\Ctxs{H} \mid \Delta,\Name{i}:B \mid \Sigma, \Name{f} :\negg B}
  }{
    \IsProc{\New{i}{f}{P}}{\Ctxs{H} \mid \Delta,\Sigma}
  }    
  }{
    \IsProc{\Par{S}{\New{i}{f}{P}}}{\Ctxs{G} \mid \Ctxs{H} \mid \Delta,\Sigma}
  }
\end{mathpar}
        
\indcase{\rulenamestyle{Claro-QueA}}
Then the redex is
\[
  \New{x}{y}{
    \Par*{
      \QueA{x}{x'}{S}
    }{
        \ExcP{y}{P}{i}{f}{z}{z'}{y'}{Q}
    }
  }
\]
The last few steps in the typing derivation must have been
\begin{mathpar}
  \inferrule*[right=Cut]{
    \IsProc{\mathcal{D}}{\Ctxs{G} \mid \Gamma, \Gamma', \Name{x}:\que A} \\
    \inferrule*[right=Claro]{
       \IsProc{P}{\Ctxs{H} \mid \Delta,\Name{i}:B \mid \Sigma, \Name{f} :\negg B} \\
        \IsProc{Q}{\Name{z} :\negg B, \Name{z'}:B, \Name{y'} : \negg A}
    }{
      \IsProc{\ExcP{y}{P}{i}{f}{z}{z'}{y'}{Q}}{\Ctxs{H} \mid \Delta, \Name{y} : \exc \negg A, \Sigma}
    }
  }{
    \IsProc{
      \New{x}{y}{
        \Par*{
          \mathcal{D}
          }{
            \ExcP{y}{P}{i}{f}{z}{z'}{y'}{Q}
          }
      }
    }{
      \Ctxs{G} \mid \Ctxs{H} \mid \Gamma,\Gamma',\Delta,\Sigma
    }
  }
\end{mathpar}
where 
\[
\mathcal{D} \defeq 
  \inferrule{
    \IsProc{S}{\Ctxs{G} \mid \Gamma, \Name{x}:\que A \mid \Gamma',\Name{x'}:A}
  }{
    \IsProc{
    \QueA{x}{x'}{S}}{\Ctxs{G} \mid \Gamma,\Gamma',\Name{x}:\que A}
  }
\]
Therefore,
\begin{mathpar}
  \inferrule*{
    \IsProc{S}{\Ctxs{G} \mid \Gamma, \Name{x}:\que A \mid \Gamma',\Name{x'}:A} \\
    \IsProc{\mathcal{D}}{\Ctxs{H} \mid \Delta, \Name{y'}:A, \Sigma, \Name{y}:\exc \negg A
    }
  }{
    \IsProc{
      \New{x}{y}{\New{x'}{y'}{\Par*{S}{\mathcal{D}}}}
    }{
      \Ctxs{G} \mid \Ctxs{H} \mid \Gamma,\Gamma',\Delta,\Sigma
    }
  }
\end{mathpar}
where $\mathcal{D}$ is
\begin{mathpar}
\inferrule*{
  \IsProc{
      \New{i}{z}{
          \Par*{P}{Q}
      }
  }{
      \Ctxs{H} \mid \Delta, \Name{z'}:B, \Name{y'}:\negg A \mid \Sigma, \Name{f}:\negg B
    }\\
    \IsProc{Q}{\Name{z} :\negg B, \Name{z'}:B, \Name{y'} : \negg A}
}{
  \IsProc{\ExcP*{y}{
      \New{i}{z}{
          \Par*{P}{Q}
      }
  }{z'}{f}{z}{z'}{y'}{Q}}{
    \Ctxs{H} \mid \Delta, \Name{y'}:\negg A, \Sigma, \Name{y}:\exc \negg A
  }
}
\end{mathpar}

\indcase{\rulenamestyle{OfCource-WhyNotW}}
\begin{mathpar}
  \inferrule*{
    \inferrule*{
      \IsProc{P}{\Ctxs{G} \mid \Gamma}
    }{
      \IsProc{\WhyW{x}{P}}{\Ctxs{G} \mid \Gamma, \Name{x}:\whynot A}
    }\\
    \inferrule*{
      \IsProc{Q}{\Name{\vec{z}} : \whynot \vec{B}, \Name{y}:\negg A}
    }{
      \IsProc{\OfCP{y}{\vec{z}}{Q}}{\Name{\vec{z}}:\whynot \vec{B}, \Name{y}:\ofc \negg A}
    }
  }{
    \IsProc{\New{x}{y}{\Par*{
      \WhyW{x}{P}
    }{
      \OfCP{y}{\vec{z}}{Q}
    }}}{
      \Ctxs{G} \mid \Gamma, \Name{\vec{z}}:\whynot \vec{B}
    }
  }
\end{mathpar}
and therefore
\begin{mathpar}
  \inferrule*{
    \IsProc{P}{\Ctxs{G} \mid \Gamma}
  }{
    \IsProc{\WhyW{\vec{z}}{P}}{\Ctxs{G} \mid \Gamma, \Name{\vec{z}}:\whynot \vec{B}}
  }
\end{mathpar}

\indcase{\rulenamestyle{OfCourse-WhyNotD}}
\begin{mathpar}
  \inferrule*{
    \inferrule*{
      \IsProc{P}{\Ctxs{G} \mid \Gamma, \Name{x'}:A}
    }{
      \IsProc{\WhyD{x}{x'}{P}}{\Ctxs{G} \mid \Gamma, \Name{x}:\whynot A}
    }\\
    \inferrule*{
      \IsProc{Q}{\Name{\vec{z}} : \whynot \vec{B}, \Name{y}:\negg A}
    }{
      \IsProc{\OfCP{y}{\vec{z}}{Q}}{\Name{\vec{z}}:\whynot \vec{B}, \Name{y}:\exc \negg A}
    }
  }{
    \IsProc{\New{x}{y}{\Par*{
      \WhyD{x}{x'}{P}
    }{
      \OfCP{y}{\vec{z}}{Q}
    }}}{
      \Ctxs{G} \mid \Gamma, \Name{\vec{z}}:\whynot \vec{B}
    }
  }
\end{mathpar}
and therefore

\begin{mathpar}
  \inferrule*{
    \IsProc{P}{\Ctxs{G} \mid \Gamma, \Name{x'}:A}\\
    \IsProc{Q}{\Name{\vec{z}} : \whynot \vec{B}, \Name{y}:\negg A}
  }{
    \IsProc{\New{x'}{y}{\Par*{P}{Q}}}{
      \Ctxs{G} \mid \Gamma, \Name{\vec{z}}:\whynot \vec{B}
    }
  }
\end{mathpar}

\indcase{\rulenamestyle{OfCourse-WhyNotC}}
\begin{mathpar}
\inferrule*{
  \inferrule*{
    \IsProc{P}{\Ctxs{G} \mid \Gamma, \Name{y_0}:\whynot A, \Name{y_1}:\whynot A}
  }{
    \IsProc{\WhyC{x}{y_0}{y_1}{P}}{\Ctxs{G} \mid \Gamma, \Name{x}:\whynot A}
  }\\
  \inferrule*{
      \IsProc{Q}{\Name{\vec{z}} : \whynot \vec{B}, \Name{y}:\negg A}
  }{
      \IsProc{\OfCP{y}{\vec{z}}{Q}}{\Name{\vec{z}}:\whynot \vec{B}, \Name{y}:\ofc \negg A}   
  }
}{
  \IsProc{\New{x}{y}{\Par*{
    \WhyC{x}{y_0}{y_1}{P}
  }{
    \OfCP{y}{\vec{z}}{Q}
  }}}{
    \Ctxs{G} \mid \Gamma, \Name{\vec{z}}:\whynot \vec{B}
  }
}
\end{mathpar}
and therefore
\begin{mathpar}
\inferrule*{
  \IsProc{\mathcal{D}}{\Ctxs{G} \mid \Gamma, \Name{\vec{z_0}}:\whynot \vec{B}, \Name{\vec{z_1}}:\whynot \vec{B}}
}{
  \IsProc{\WhyC{\vec{z}}{\vec{z_0}}{\vec{z_1}}{\mathcal{D}}}{
    \Ctxs{G} \mid \Gamma, \Name{\vec{z}}:\whynot \vec{B}
  }
}
\end{mathpar}
where $\mathcal{D}$ is 
\begin{mathpar}
  \inferrule*{
    \IsProc{P}{\Ctxs{G} \mid \Gamma, \Name{y_0}:\whynot A, \Name{y_1}:\whynot A}\\
    \IsProc{\Sb{Q}{z_0}{z}}{
      \Name{\vec{z_0}} : \whynot \vec{B}, \Name{y}:\negg A
    }\\
    \IsProc{\Sb{Q}{z_1}{z}}{
      \Name{\vec{z_1}} : \whynot \vec{B}, \Name{y}:\negg A
    }
  }{
    \IsProc{\New{y_0}{z_0}{\New{y_1}{z_1}{
      \Par*{
      P
    }{
      \Par{
        \Sb{Q}{z_0}{z}
      }{
        \Sb{Q}{z_1}{z}
    }}}}}{
      \Ctxs{G} \mid \Gamma, \Name{\vec{z_0}}:\whynot \vec{B}, \Name{y_1}:\whynot A
    }
  }
\end{mathpar}
\end{indproof}
\end{proof}

\begin{lemma}
	$\bra*{\overline{T_S}} = \negg{\bra*{T_S}}$.
\end{lemma}	
\begin{proof}
	By simple induction.
\end{proof}

\begin{lemma} \label{lemma:unlimited-positive}
	If $T$ is unlimited, we have the following derivable rule in CSLL.
	\begin{mathpar}
		\inferH{Positive}{
			\IsProc{P}{\Ctxs{G} \mid \Gamma, \Name{x'}:\whynot \negg{\bra*{T}}}
		}{
			\IsProc{\Coderelict{x}{x'}{P}}{\Ctxs{G} \mid \Gamma, \Name{x}:\negg{\bra*{T}}}
		}
	\end{mathpar}
\end{lemma}
The above lemma means that all unlimited types enjoy contraction and weakening. Some session types,
such as $end_?$, also enjoy such properties: see e.g. \cite[\S 5]{gay2010linear}. However, in order
to retain the good properties of termination and deadlock-freedom, we insist that all channels are
used linearly, and carefully closed at the end.

\begin{proof}
	It is given by the well-known fact that $\ofc A$, $\ounit$ and $\cunit$ are always positive, that $\otimes$, $\oplus$ preserve positivity, and that our system (without server and client) is equivalent to linear logic in terms of expressivity \citep[\S 2.3]{kokke2019better}. More concretely, we derive the rule by induction on $T$. 
	\begin{itemize}
		\item If $T$ is $\textsf{Unit}$ we have
		\begin{mathpar}
				\inferrule*{
					\inferrule*{
						\IsProc{\Out{y}{}{\Stop}}{\Name{y}:\ounit}
					}{
						\IsProc{\OfCP{y}{}{\Out{y}{}{\Stop}}}{\Name{y}:\ofc \ounit}
					}
				}{
					\IsProc{\In{x}{}{\OfCP{y}{}{\Out{y}{}{\Stop}}}}{
					\Name{y}: \ofc \ounit, \Name{x}: \punit
					}
				}
		\end{mathpar}
		Cut $\Name{y}$ of this with $\Name{x'}$ of $\Proc{P}$ and we are done.
		\item If $T$ is $U \rightarrow V$, we have
			\begin{mathpar}
					\inferrule*{
						\IsProc{\Link{y}{x}}{\Name{y}:\ofc (\bra*{U}\otimes \negg{\bra*{V}}), \Name{x}:\whynot (\negg{\bra*{U}} \parr \bra*{V})}
					}{
						\IsProc{\OfCP{y}{x}{\Link{y}{x}}}{
							\Name{y}:\ofc \ofc (\bra*{U}\otimes \negg{\bra*{V}}), \Name{x}:\whynot (\negg{\bra*{U}} \parr \bra*{V})
						}
					}
			\end{mathpar}
		Cut $\Name{y}$ of this with $\Name{x'}$ of $\Proc{P}$ and we are done.
		\item If $T$ is $U + V$, and both $U$ and $V$ are unlimited. First we apply \cref{lem:separation} on $\Proc{P}$ and acquire $\IsProc{P_0}{\Ctxs{G}}$ and $\IsProc{P_1}{\Gamma, \Name{x}:\whynot(\negg{\bra*{U}} \with \negg{\bra*{V}})}$, and we have $\Proc{\mathcal{D}_U}$ defined as
		\begin{mathpar}
				\inferrule*{
					\inferrule*{
						\inferrule*{
							\IsProc{\Link{y}{x}}{
								\Name{y}:\bra*{U},
								\Name{x}:\negg{\bra*{U}}
							}
						}{
							\IsProc{\Inl{y}{\Link{y}{x}}}{
								\Name{y}:\bra*{U} \oplus \bra*{V},
								\Name{x}:\negg{\bra*{U}}
							}
						}
					}{
						\IsProc{\WhyD{x'}{x}{\Inl{y}{\Link{y}{x}}}}{
							\Name{y}: \bra*{U} \oplus \bra*{V},
							\Name{x'}:\whynot \negg{\bra*{U}}
						}
					}
				}{
					\IsProc{\Coderelict{x}{x'}{\OfCP{y}{x'}{
						\WhyD{x'}{x}{\Inl{y}{\Link{y}{x}}}
					}}}{
						\Name{y}:\ofc(\bra*{U} \oplus \bra*{V}),
						\Name{x}:\negg{\bra*{U}}
					}
				}
		\end{mathpar}
		and similarly for $\Proc{\mathcal{D}_V}$. Finally we have
		\begin{mathpar}
			\inferrule*{
				\IsProc{\mathcal{D}_U}{
					\Name{y}:\ofc(\bra*{U} \oplus \bra*{V}),
						\Name{x}:\negg{\bra*{U}}
				}\\
				\IsProc{\mathcal{D}_V}{
					\Name{y}:\ofc(\bra*{U} \oplus \bra*{V}),
						\Name{x}:\negg{\bra*{V}}
				}
			}{
				\IsProc{\Case{x}{\mathcal{D}_U}{\mathcal{D}_V}}{
					\Name{y}:\ofc(\bra*{U} \oplus \bra*{V}),
						\Name{x}:\negg{\bra*{U}} \with \negg{\bra*{V}}
				}
			}
		\end{mathpar}
		Cut $\Name{y}$ of this with $\Name{x'}$ of $\Proc{P_1}$ then combine with $\Proc{P_0}$ and we are done.
		\item If $T$ is $U \otimes V$, and both $U$ and $V$ are unlimited, we have
		\begin{mathpar}
			\inferrule*{
			\inferrule*{
			\inferrule*{
				\IsProc{\Link{v}{v'}}{\Name{v'}:\bra*{V}, \Name{v}:\negg{\bra*{V}}}\\
				\IsProc{\Link{u}{u'}}{\Name{u'}:\bra*{U}, \Name{u}:\negg{\bra*{U}}}
			}{
				\IsProc{\Out{v'}{u'}{\Par*{
					\Link{v}{v'}
				}{
					\Link{u}{u'}
				}}}{
					\Name{v'}:\bra*{U} \otimes \bra*{V},
					\Name{u}:\negg{\bra*{U}}, 
					\Name{v}:\negg{\bra*{V}}
				}
			}}{
				\IsProc{
							\WhyD{u''}{u}{\WhyD{v''}{v}{
								\Out{v'}{u'}{\Par*{
					\Link{v}{v'}
				}{
					\Link{u}{u'}
				}}}}
				}{
					\Name{v'}:\bra*{U} \otimes \bra*{V},
					\Name{u''}:\whynot \negg{\bra*{U}}, 
					\Name{v''}:\whynot \negg{\bra*{V}}
				}
			}
			}{
				\IsProc{
					\In{v}{u}{
					\Coderelict{u}{u''}{\Coderelict{v}{v''}{
					\OfCP{v'}{u,v}{
						\WhyD{u''}{u}{\WhyD{v''}{v}{
								\Out{v'}{u'}{\Par*{
					\Link{v}{v'}
				}{
					\Link{u}{u'}
				}}}}
					}}
				}}}{
					\Name{v'}:\ofc(\bra*{U} \otimes \bra*{V}),
					\Name{v}:\negg{\bra*{U}} \parr \negg{\bra*{V}}	
				}
			}
		\end{mathpar}
		Cut $\Name{v'}$ of this with $\Name{x}$ of $\Proc{P}$, rename $\Name{v}$ to $\Name{x}$ and we are done.
	\end{itemize}
\end{proof}
\section{More Examples}
\subsection{Compare-And-Set} \label{section:csgv-cas}
We now recover the example of CAS server/client in CSGV. 
We define the server-client protocol to be $T_S \defeq \send \boo.\send \boo.\recv \boo.\endm$. The choice of $\endp$ vs. $\endm$ is purely driven by well-typedness.  

	\begin{minipage}{0.5\linewidth}
	\begin{align*}
		\Deriv{C_0}  \defeq \  &  \Let{x_d}{\Send{\False}{\Name{x_0}}} \\
			  &  \Let{x_r}{\Send{\True}{\Name{x_d}}} \\
			  & \LetP{r}{x'}{\Recv{\Name{x_r}}} \\
			  & \Let{\_}{\Term{\Send*{\Name{r}}{\Name{r0}}}} \\
			  & \Name{x'} \\
		\Deriv{C_1} \defeq \ &  \dots \\
		\Derivsf{clients} \defeq \ & \Let{x}{\ReqA{x}{x_0}{C_0}} \\
					   & \Let{x}{\ReqA{x}{x_1}{C_1}}\\
					   & \ReqW{x}
	\end{align*}
	\end{minipage}
	\begin{minipage}{0.5\linewidth}
	\begin{align*}
		\Deriv{L} \defeq \ & \False \\
		\Deriv{M} \defeq \ 
			& \LetP{exp}{y'}{\Recv{\Name{y}}} \\
			& \LetP{des}{y''}{\Recv{\Name{y'}}} \\
			& \IfThen{\Name{exp} = \Name{w}} \\
			& \Let{\_}{\Term{\Send*{\True}{\Name{y''}}}} \; \Name{des} \\
			& \Else \\
			& \Let{\_}{\Term{\Send*{\False}{\Name{y''}}}} \; \Name{w} \\
		\Deriv{N} \defeq\  & \Name{z} \\
		\Derivsf{server} \defeq\  &\Serv{y}{L}{w}{M}{z}{N}
	\end{align*}
	\end{minipage}

typed as:
\begin{equation*}
	\begin{aligned}[c]
	 \IsProg*{\Name{x_0} : T_S, \Name{r_0} : \send \boo.\endp}{C_0}{\endm} \\
	 \IsProg*{\Name{x_1} : T_S, \Name{r_1} : \send \boo.\endp}{C_1}{\endm} \\
	 \IsProg*{\Name{x} : \que T_S,\Name{r_0}:\send \boo.\endp,\Name{r_1}:\send \boo.\endp}{\textsf{clients}}{\endm} 
	\end{aligned}
	\qquad
	\begin{aligned}[c]
		\IsProg*{}{L}{\boo} \\
		\IsProg*{\Name{w} : \boo, \Name{y} : \recv \boo.\recv \boo.\recv \boo.\endp}{M}{\boo} \\
		\IsProg*{\Name{z} : \boo}{N}{\boo} \\
		\IsProg*{\Name{y} : \exc \overline{T_S}}{\textsf{server}}{\boo}
	\end{aligned}
\end{equation*}

	where $\overline{T_S} = \recv \boo.\recv \boo.\send \boo.\endp$. Finally we have
	\begin{mathpar}
		\IsProg{\Name{r_0}:\send \boo.\endp, \Name{r_1}:\send \boo.\endp}{\Conn{x}{\textsf{clients}}{y}{\textsf{server}}}{\boo}
	\end{mathpar}

	\subsection{List Shuffling}\label{section:list-shuffling}
	We use the racing behaviour of clients to shuffle a list. 
	 We define server/client protocol to be $T_S \defeq \send A.\endm$, meaning each client sends a value of $A$ and ends the session. Each clients are defined the same way: they simply take the value $A$ from the environment and send it over the channel. Clients are forked by folding the list.	The server simply receives values from clients and reforms the list.

	\begin{minipage}{0.5\linewidth}
	\begin{align*}
		\Prog{C} \defeq\ 
			& \Send{\Name{v}}{\Name{x'}} \\
		\Progsf{clients} \defeq \ 
			& \Let{y}{\Namesf{fold}_{\que T_S}\; \Name{x}\; \Lam{x}{\Lam*{v}{\ReqA{x}{x'}{C}}}\; \Name{l}} \\
			& \ReqW{y}
	\end{align*}
\end{minipage}
\begin{minipage}{0.5\linewidth}
	\begin{align*}
	 \Prog{L} \defeq\  & \Namesf{nil} \\
	 \Prog{M} \defeq\  & \LetP{v}{y'}{\Recv{\Name{y}}} \\
			  & \Let{\_}{\Term{\Name{y'}}} \\
			  & \Progsf{cons}\; \Name{v}\; \Name{z'}\\
	 \Prog{N} \defeq\  & \Name{z}\\
	 \Progsf{server} \defeq\  &\Serv{y}{L}{z'}{M}{z}{N}
	\end{align*}
\end{minipage}

	and we have the following typing
	\begin{equation*}
		\begin{aligned}[c]
		 \IsProg*{\Name{v}:A, \Name{x'}:T_S}{C}{\endm}\\
		 \IsProg*{\Name{l}:[A], \Name{x}:\que T_S}{\textsf{clients}}{\endm}\\
		 \IsProg*{}{L}{[A]} 
		\end{aligned}
		\qquad
		\begin{aligned}[c]
		 \IsProg*{\Name{z'}:[A], \Name{y}:\overline{T_S}}{M}{[A]} \\
		 \IsProg*{\Name{z}:[A]}{N}{[A]}\\
		 \IsProg*{\Name{y}:\exc \overline{T_S}}{\textsf{server}}{[A]}
		\end{aligned}
	\end{equation*}
	and finally we define shuffling as 
	\begin{align*}
	 \IsProg{\Name{l}:[A]}{\Conn{x}{\textsf{clients}}{y}{\textsf{server}}}{[A]}
	\end{align*}

	\subsection{Merge Sort} \label{section:merge-sort}
	Using fork-join, we can define parallel merge sort. We first have to assume general recursion (and therefore not expressible in the vanilla CSGV), and two functions on lists of $A$. $\textsf{split}$ splits a list of $A$ into several (supposedly two) lists, and $\textsf{merge}$ merges several sorted lists into one list. $\textsf{isend}$ returns $\True$ if the list is empty or singleton.
	\begin{mathpar}
	 \IsProg{}{\textsf{split}}{[A] \rightarrow [[A]]}
	 \and
	 \IsProg{}{\textsf{merge}}{[[A]] \rightarrow [A]}
	 \and
	 \IsProg{}{\textsf{isend}}{[A] \rightarrow \boo}
	\end{mathpar}
	And we define merge sort as follows. We first check if the list $\Name{l}$ is too short to sort; if not we $\Progsf{split}$ the list and $\Progsf{sort}$ each of the sub-lists. The sorted sub-lists are collected into $\Name{l'}$ which we will $\Progsf{merge}$. Note that the racing behaviour of client/server means $\Name{l'}$ could be any ordering, which does not matter for merge sort. For scnerios where it does matter, each sub-result should be accompanied by its index to get re-ordered.
	\begin{align*}
	 \Progsf{sort} \; \Name{l} \defeq\  & \IfThen{\Progsf{isend} \; \Name{l}} \; \Name{l} \\
							 & \Else \\
							 & \Let{l'}{\ForkJoin{\textsf{sort}}{\textsf{cons}} {\textsf{nil}}{(\textsf{split} \; \Name{l})}} \\
							 & \Progsf{merge} \; \Name{l'}
	\end{align*}
	which gives us
	\begin{align*}
	 \IsProg{}{\textsf{sort}}{[A] \rightarrow [A]}
	\end{align*}

		\subsection{Map-Reduce} %
	   The purpose of the map-reduce model is to transform input of type $[A]$ into output of type $[D]$ using two functions (using the functorial formulation given by \citet{hinrichsen2019actris}). Take the example of counting the frequency of each word in an article which contains several paragraphs $[A]$. The \textit{map} function $\Name{f}$ counts the frequency $C$ of each word $B$ in a paragraph. The \textit{reduce} function $\Name{g}$ takes a word $B$ and its frequency $[C]$ in all paragraphs and simply returns the word with the sum frequency $D \defeq B \otimes C$. In the end we hope to get $[D]$.
	   \begin{mathpar}
		\Name{f} : A \rightarrow [B \otimes C]
		\and
		\Name{g} : B \otimes [C] \rightarrow D
	   \end{mathpar}
	   We first define parallelized $\Progsf{flatMap}$ with fork-join:
	   \begin{align*}
			\IsProg{}{\Progsf{flatMap}_{A,B} \defeq \Lam{f}{\Lam{l}{\ForkJoin{f}{\Progsf{nil}}{\Progsf{concat}}{\Name{l}}}}}{(A \rightarrow [B]) \rightarrow [A] \rightarrow [B]}
	   \end{align*}
	   where $\Progsf{concat}$ is the standard function that concatenate two lists. Based on this we define $\Progsf{map-reduce}$:
	   \begin{align*}
		\IsProg{\Name{f}, \Name{g}}{\textsf{map-reduce} \defeq (\Progsf{flatMap}_{B \otimes [C], D} \; \Name{g}) \circ \Progsf{group}_{B,C} \circ (\Progsf{flatMap}_{A,B \otimes C} \; \Name{f})}{[A] \rightarrow [D]}
	   \end{align*}
	   where $\IsProg{}{\textsf{group}_{B,C}}{[B \otimes C] \rightarrow [B \otimes [C]]}$ is the standard function that groups list of pairs by their keys. 
	   
	   There is a notable difference between our version of map-reduce and the version in \citet{hinrichsen2019actris} (and other related literatures ). Usually a fixed number of threads (that usually corresponds to the number of cpu cores/nodes) are spawned, who will then repeatedly retrieve tasks from and send result to the main thread. In our version, however, the number of threads is the number of tasks, and each thread will handle one task only. The former approach seems lower-level, allowing optimizing the number of threads according to the hardware reality. Our language is higher-level, and it is up to the implementation to coordinate threads with cores/nodes. Implementing it at a lower-level seems to be difficult because of the linearity constraints.
	\subsection{Interleaving clients} \label{section:interleaving-clients}
	Another interleaving clients example (but simpler than the beauty contest example) is one where each client submits a boolean to the server, who calculates the XOR of all the submissions and sends the result back to all clients.  The internal protocol of the server, as well as the server interface, will be $T_S \defeq \recv \boo. \send \boo. \endp$.
	We define
	\begin{align*}
	 L \defeq & \Let{w'}{\Send{\False}{\Name{w}}} \tag{send initial value}\\
			  & \LetP{\_}{w''}{\Recv{\Name{w'}}} \tag{recv the final value} \\
			  & w'' \\
	 M \defeq & \LetP{s}{z''}{\Recv{\Name{z'}}} \tag{recv the last value} \\
			  & \LetP{b}{y'}{\Recv{\Name{y}}} \tag{recv the boolean from client} \\
			  & \Let{s'}{\Progsf{xor}(\Name{s}, \Name{b})} \tag{calculate the xor} \\
			  & \Let{w'}{\Send{\Name{s'}}{\Name{w}}} \tag{send to next worker process}\\
			  & \LetP{f}{w''}{\Recv{\Name{w'}}} \tag{recv the final boolean} \\
			  & \Let{\_}{\Term*{\Send{\Name{f}}{\Name{z''}}}} \tag{send the final to previous worker process} \\
			  & \Let{\_}{\Term*{\Send{\Name{f}}{\Name{y'}}}} \tag{send the final to client} \\
			  & \Name{w''}\\
	 N \defeq & \LetP{f}{z'}{\Recv{\Name{z}}} \tag{recv the final value} \\
			  & \Let{z''}{\Send{\Name{f}}{\Name{z'}}} \tag{simply send it back} \\
			  & \Term{\Name{z''}} \\
	 \Progsf{server} \defeq & \Serv{y}{\Inv{w}{L}}{z'}{\Inv{w}{M}}{z}{N} 
	\end{align*}
	where $\IsProg{}{\textsf{xor}}{\boo \otimes \boo \rightarrow \boo}$ is the standard xor function on booleans.
	\begin{mathpar}
	 \IsProg{\Name{w}: \overline{T_S}}{L}{\endm}
	 \and
	 \IsProg{\Name{w}:\overline{T_S}, \Name{z'}:T_S, \Name{y}:T_S}{M}{\endm} 
	 \and
	 \IsProg{\Name{z}:T_S}{N}{\textsf{Unit}} 
	 \and
	 \IsProg{\Namesf{xor}, \Name{y}:\exc T_S}{\textsf{server}}{\textsf{Unit}}
	\end{mathpar}
	We omit defining the clients as they will be very similar to the ones in previous examples.

	\subsection{Symbol Generator} \label{section:symbol-generator}
	Another simple scenario is where server acts like a generator of unique symbols (essentially natural numbers $\nat$) and clients race to acquire those symbols. The server protocol is $\send \nat. \endp$, meaning the server simply sends an number to the client and ends the session; the server internal state is $\nat$.
	\begin{align*}
	 L \defeq \ & \Progsf{zero} \tag{starts with zero}\\
	 M \defeq \ & \Let{\_}{\Term*{\Send{\Name{z'}}{\Name{y}}}} \tag{send the counter to client}\\
			   & \Progsf{succ} \; \Name{z'} \tag{increase the counter}\\
	 N \defeq\  & \Name{z} \tag{output the final counter}\\
	 \Progsf{server} \defeq & \Serv{y}{L}{z'}{M}{z}{N}
	\end{align*}
	typed as
	\begin{mathpar}
	 \IsProg{}{L}{\nat} 

	 \IsProg{\Name{z'}:\nat, \Name{y}:\send N. \endp}{M}{\nat} 

	 \IsProg{\Name{z}:\nat}{N}{\nat}

	 \IsProg{\Name{y}:\exc(\send \nat. \endp)}{\Progsf{server}}{N}
	\end{mathpar}
	We omit defining the clients as they would be similar to previous ones; but we note that it is impossible for a process to act as multiple clients and aggregate two symbols. The reason is that informally speaking, in our system clients are not allowed to communicate with each other besides via the server as indicated by the functor. More concretely, supposed we are to define a process acting as multiple clients, it would be typed as $\IsProg{\Gamma, \Name{x_0}: \recv N. \endp, \Name{x_1}: \recv N. \endp}{K}{T}$; but there is no way in CSGV to combine $\Name{x_0}$ and $\Name{x_1}$.

\end{document}